\documentclass[11pt,preprint]{imsart}

\RequirePackage{amsthm,amsmath,amsfonts,amssymb}
\RequirePackage[numbers]{natbib}
\RequirePackage[colorlinks=true,citecolor=blue,urlcolor=blue]{hyperref}
\RequirePackage[top=1in, bottom=1in, left=1in, right=1in]{geometry}

\startlocaldefs
\numberwithin{equation}{section}
\theoremstyle{plain}   
\DeclareMathOperator*{\esssup}{ess\,sup}
\usepackage{xcolor}
\newtheorem{theorem}{Theorem}[section]
\newtheorem{proposition}[theorem]{Proposition}
\newtheorem{lemma}[theorem]{Lemma}

\newtheorem{corollary}[theorem]{Corollary}
\newtheorem{definition}[theorem]{Definition}

\newtheorem{remark}[theorem]{Remark}

\newcommand{\wh}{\widehat}
\newcommand{\ov}{\overline}
\newcommand{\wt}{\widetilde}
\def\sign{\mbox{sign}}

\newcommand\cB{{\cal B}}
\newcommand\cE{{\cal E}}
\newcommand\cH{{\cal H}}

\newcommand\cF{{\cal F}}
\newcommand\cL{{\cal L}}
\newcommand\cI{{\cal I}}

\newcommand\cR{{\cal R}}

\newcommand\cD{{\cal D}}

\newcommand\cV{{\cal V}}
\newcommand\cU{{\cal U}}

\newcommand\cZ{{\cal Z}}

\def\bbr{{\mathbb R}}

\def\text#1{\hbox{#1}}
\def\proof{{\noindent \bf Proof. }}

\def\essinf{{\text essinf}}
\def\endproof{\mbox{\ $\qed$}}

\def\E{{\bf E}}

\def\P{{\bf P}}

\def\I{{\bf I}}

\def\R{{\bf R}}
\def\Q{{\bf Q}}

\def\d{\mathrm{d}}
\def\build #1_#2{\mathrel{\mathop{\kern 0pt #1}\limits_\zs{#2}}}
\newcommand{\zs}[1]{{\mathchoice{#1}{#1}{\lower.25ex\hbox{$\scriptstyle#1$}}
{\lower0.25ex\hbox{$\scriptscriptstyle#1$}}}}

\numberwithin{equation}{section}

\def\proof{{\noindent \bf Proof. }}
\def\endproof{\mbox{\ $\qed$}}
\newtheorem{example}{Example}
\newtheorem*{example*}{Example}

\newcommand{\disc}[1]{\textcolor{blue}{#1}}

\newcommand{\tc}[1]{\textcolor{blue}{#1}}
\newcommand{\tn}[1]{\textcolor{blue}{#1}}

\def\F{\mathcal{F}}
\def\R{\mathbb{R}}

\begin{document}

\begin{frontmatter}
%%%%%%%%%%%%%%%%%%%%%%%%%%%%%%%%%%%%%%%%%%%%%%
%%                                          %%
%% Enter the title of your article here     %%
%%                                          %%
%%%%%%%%%%%%%%%%%%%%%%%%%%%%%%%%%%%%%%%%%%%%%%
%\title{aaaa???}
%%\title{A sample article title with some additional note\thanksref{T1}}
%\runtitle{bbb???}
%\thankstext{T1}{ccccA sample of additional note to the title.}

\title{%Forward BSDEs and backward SPDEs for utility maximization under endogenous pricing \\ 
Utility maximization under endogenous pricing \thanks{An earlier draft of the paper was disseminated under the title ``Forward BSDEs and backward SPDEs for utility maximization under endogenous pricing.'' We thank Antoon Pelsser, Patrick Cheridito, Michail Anthropelos, Frank Riedel, Sergio Pulido, Michael Kupper, Dirk Becherer, Roger Laeven, Giorgio Ferrari, Mario Ghossoub, and the participants of the GPSD 2023, the 8th Eastern Conference on Mathematical Finance, and seminars in UvA, Waterloo, Bielefeld, and ETH Zürich for their comments and feedback where the usual caveat applies.}}
%\footnotetext{As earlier draft of the paper was disseminated under the title ``Forward BSDEs and backward SPDEs for utility maximization under endogenous pricing.''}
 
%\thaicomment{Which titile do we want to use finally?}}
%\begin{aug}
%%%%%%%%%%%%%%%%%%%%%%%%%%%%%%%%%%%%%%%%%%%%%%
%%Only one address is permitted per author. %%
%%Only division, organization and e-mail is %%
%%included in the address.                  %%
%%Additional information can be included in %%
%%the Acknowledgments section if necessary. %%
%%%%%%%%%%%%%%%%%%%%%%%%%%%%%%%%%%%%%%%%%%%%%%

%\runtitle{}
\begin{aug}
%%%%%%%%%%%%%%%%%%%%%%%%%%%%%%%%%%%%%%%%%%%%%%
%%Only one address is permitted per author. %%
%%Only division, organization and e-mail is %%
%%included in the address.                  %%
%%Additional information can be included in %%
%%the Acknowledgments section if necessary. %%
%%%%%%%%%%%%%%%%%%%%%%%%%%%%%%%%%%%%%%%%%%%%%%

\author[A]{\fnms{Thai } \snm{Nguyen}\ead[label=e1]{thai.nguyen@act.ulaval.ca}}\and 
\author[B]{\fnms{Mitja } \snm{Stadje}\ead[label=e2]{mitja.stadje@uni-ulm.de}}
%\author[A]{\fnms{ttt???} \snm{nnn???}\ead[label=e1]{thai@act.edu.vn???}},
%\author[B]{\fnms{???} \snm{???}\ead[label=e2,mark]{???@???}}
%\and
%\author[B]{\fnms{???} \snm{???}\ead[label=e3,mark]{???@???}}
%%%%%%%%%%%%%%%%%%%%%%%%%%%%%%%%%%%%%%%%%%%%%%
%% Addresses                                %%
%%%%%%%%%%%%%%%%%%%%%%%%%%%%%%%%%%%%%%%%%%%%%%
\address[A]{Université Laval, \'{E}cole d'Actuariat, }
\address[B]{University of Ulm, Faculty of Mathematics and Economics, Institute of Insurance Science and Institute of Financial Mathematics}

%\address[B]{Canada???, \printead{e2,e3}}
\end{aug}

%\author{Thai Nguyen\footnote{Universit\'{e} Laval, \'{E}cole d'Actuariat, Qu\'ebec city, Canada. %and School of Economic Mathematics-Statistics, University of Economics Hochiminh city, Vietnam. 
%Email: thai.nguyen@act.ulaval.ca}  \, and Mitja Stadje\footnote{University of Ulm, Institute of Insurance Science and Institute of Financial Mathematics, Ulm, Germany. Email: mitja.stadje@uni-ulm.de}}
%\date{\today}
%\maketitle
%\author[A]{\fnms{Thai } \snm{Nguyen}\ead[label=e1]{thai.nguyen@act.ulaval.ca}}
%\and
%\author[B]{\fnms{Mitja } \snm{Stadje}\ead[label=e2]{mitja.stadje@uni-ulm.de}}
%\author[A]{\fnms{ttt???} \snm{nnn???}\ead[label=e1]{thai@act.edu.vn???}},
%\author[B]{\fnms{???} \snm{???}\ead[label=e2,mark]{???@???}}
%\and
%\author[B]{\fnms{???} \snm{???}\ead[label=e3,mark]{???@???}}
%%%%%%%%%%%%%%%%%%%%%%%%%%%%%%%%%%%%%%%%%%%%%%
%% Addresses                                %%
%%%%%%%%%%%%%%%%%%%%%%%%%%%%%%%%%%%%%%%%%%%%%%
%\address[A]{Université Laval, \'{E}cole d'Actuariat, \printead{e1}}
%\address[B]{University of Ulm, Institute of Insurance Science and Institute of Financial Mathematics, \printead{e2}}
%\address[A]{ulm???, \printead{e1}}

%\address[B]{Canada???, \printead{e2,e3}}
%\end{aug}

\begin{abstract}
We study the expected utility maximization problem of a large investor who is allowed to make transactions on tradable assets in an incomplete financial market with endogenous permanent market impacts.
The asset prices are assumed to
follow a nonlinear price curve quoted in the market as the utility indifference curve
of a representative liquidity supplier. {Using generalized subgradients}, we show that optimality can
be fully characterized via a system of coupled forward-backward stochastic differential
equations (FBSDEs) which corresponds to a non-linear backward stochastic partial differential
equation (BSPDE). We show existence of solutions to {the optimal investment problem and the} FBSDEs in the case where the
driver function of the representative market maker grows at least quadratically or the utility function of the large investor falls faster than quadratically or is exponential. Furthermore, we derive smoothness results for the existence of solutions of BSPDEs. Examples are provided when the market is complete, {the driver is positively homogeneous} or the utility function is exponential.

\end{abstract}



\end{frontmatter}

\section{Introduction} \label{sec:int}

%\section{Introduction}
This paper studies a stochastic optimal control problem with feedback effect using coupled forward backward stochastic differential equations (FBSDEs) and backward stochastic partial differential equations (BSPDEs). %\red{We show that FBSDEs and BSPDEs characterize the optimal wealth process and the value function, respectively, for optimal investment under endogenous pricing with permanent price impacts. Furthermore, we prove that for large classes of utility functions an optimal solution exists and the associated FBSDE has a solution. Moreover, we show that the value function is smooth if the market is complete or the utility function is exponential.}\\
{Our findings demonstrate that FBSDEs and BSPDEs can serve to define optimal wealth and value functions in the context of investment with endogenous pricing and permanent price impacts. Furthermore, we establish the existence of optimal solutions for various utility functions, accompanied by solutions for associated FBSDEs. Additionally, we reveal that the value function exhibits smoothness when the market is complete or when the utility function is exponential.}

\vspace{2mm}
{To comprehend our market impact model, consider a substantial trader or financial investor dealing in risky assets like financial derivatives within the financial market. Traditional financial mathematics assumes traders are price takers, meaning asset prices and their stochastic processes are predefined and unaffected by trader activities. However, it's widely acknowledged that significant buying or selling impacts prices by altering supply and asset volatility, as discussed in \cite{bookstaber2007demon,stoll1978supply,ho1981optimal}.} %Consider a large trader or a financial investor who trades risky assets like financial derivatives in a financial market. Classical financial mathematics supposes that the trader is a price taker in the sense that the prices of the risky assets (and the stochastic process representing these) at every time instance are exogenously given and are in particular not affected by the trades of the investor. However, it is well known that the buying or selling of larger positions actually influences prices by affecting the supply or the volatility of the underlying asset, see for example \cite{bookstaber2007demon,stoll1978supply,ho1981optimal}.  %This paper solves a utility maximization problem under price impacts with a analytically tractable structure.
{In this paper,} we consider a model with (permanent) market impact and analyze a
utility maximization problem for a large investor who is trading risky assets in a financial market where the trading influences the future prices, and the price curves are non-linear in volume. Thus, our model naturally encompasses phenomena such as nonlinearity in liquidation and market contractions resulting from illiquidity. Despite its significant nonlinearity, our model unexpectedly facilitates a comprehensive mathematical analysis of the optimal portfolio choice problem and enables the characterization of the optimal strategy, suitable for numerical computations.  {As far as we are aware, this paper marks a pioneering endeavor to address the utility maximization problem in the presence of fully non-linear endogenous market impact using robust expectations.}
%our model captures endogenously such phenomena as nonlinearity in liquidation and market contractions due to illiquidity. Although being highly non-linear, our model surprisingly allows for a thorough mathematical analysis of the problem of optimal portfolio choice and a characterisation of the optimal strategy which can be used for numerical computation. 

\vspace{2mm}
Our contribution can be summarized as follows: a) Building on \cite{stoll1978supply,ho1981optimal,S07,S14,BK1,BK2,KP16a} we introduce a new type of intrinsically non-linear optimal control problem with a solid foundation in endogoneous pricing. b) We show that the problem is well defined in a continuous-time framework which arises naturally as the limit of discrete-time trading with feedback effects. c) In general incomplete markets with endowments of the investor and the market maker, we characterize, {using generalized subgradients}, optimal solutions in terms of coupled FBSDEs and of BSPDEs, {thereby highlighting the interplay between these approaches.} Furthermore, the FBSDEs dependence on the utility function of the investor can be completely captured through the level of prudence the utility function induces. d) We give new existence results for the arising FBSDEs and BSPDEs.  e) {We furnish examples in scenarios where the market is complete, or the driver function is positively homogeneous, and/or the utility function adopts an exponential form. Notably, under the circumstances of linear pricing without market impact, our problem simplifies to the classical utility maximization problem.} % We give examples in the case that the market is complete with an exponential utility function. \red{In the case of {linear pricing without market impact}, our problem reduces to the classical utility maximization problem.}

\vspace{2mm}

The classical Merton utility maximization problem \cite{Merton1971} with and without transaction costs has been extensively studied and we recommend \cite{muhle2017primer} for an overview. In the absence of transaction costs and feedback effects, the classical solution methods are based on convex duality which first appeared in {\cite{bismut}} and was extended for example in {\cite{bismut1973,pliska,karatzas87,karatzas1991,cvitanic2,kramkov1999asymptotic} and (with endowment) in }\cite{https://doi.org/10.1111/j.1467-9965.2008.00360.x}. {Despite their broad applicability, dual methods come with a drawback: they are not well-suited for numerical approximations.}
%\red{Although quite general, dual methods have the disadvantage of not lending themselves to numerical approximations.} 
A somewhat newer approach {in particular} for exponential, power, or logarithmic utility functions is based on BSDEs with and without convex duality methods; see {\cite{rouge2000,horst2010,
	hu2005,sekine, mania2010,laeven2013}} among others. {BSDEs can, for instance through Monte Carlo simulation, be approximated efficiently. They offer advantages such as independence from convex hedging constraints and unlike PDEs a broader applicability beyond Markovian frameworks.}
	%\red{BSDEs can for instance through Monte Carlo simulation be approximated efficiently, do not rely on convex hedging constraints, and contrary to PDEs are in general not restricted to Markovian frameworks.}
	%The existence and uniqueness of an optimal strategy is usually characterized via the solution to the dual variational problem. 
	In \cite{horst2014,santacroce2014,santacroce2023}, the optimal strategy for an unhedgeable terminal condition and general utility functions can be described through a solution to a fully coupled forward backward system. %Similar results with slightly weaker technical conditions are also discussed in \cite{santacroce2014}. 
{By using dynamic programming, the authors in \cite{mania2003,mania2010,mania2017} show that the optimal strategy can be represented  in terms of the value function related to the problem and its derivatives. Furthermore, under appropriate assumptions, the value function is characterized as the solution to a BSPDE.} %\red{(IS OUR ANALYSIS MORE DIFFICULT THAN THE CLASSICAL CASE?)} 
%{In comparison to these traditional cases, our framework, characterized by a fully non-linear driver, poses considerably greater challenges. Our innovative contribution lies in establishing the existence of solutions for both the optimal investment problem and FBSDEs in scenarios where the driver function of the representative market maker grows at a minimum quadratic rate, or the utility function of the significant investor decays at a rate faster than quadratic, or is exponential.}

\vspace{2mm}
For utility maximization problems where the trading itself influences the prices most works focus on temporary or transient price impact of exogenous-type; see \cite{bank2017hedging} and the references therein. %% see for instance \cite{chandra2019,ekren2019,garleanu2013,garleanu2016,
%%guasoni2015,guasoni2018,guasoni2017,kallsen2014,moreau2017,mete2013utility,bank2017hedging,bank2018linear}.
%Temporary market impacts have been extensively investigated and typically have the same effect to traders' P\&L as
%transaction costs see e.g., Cetin et al~\cite{Cetin} and
%Fukasawa~\cite{F}. 
{We note that our paper focuses on investigating the permanent price impact, rather than temporary market impact, and assumes that trades are fragmented into smaller portions, allowing their effects to dissipate gradually over time. This concept is discussed comprehensively in \cite{GathSc}.}
%We remark that neglecting temporary market impact, \red{and instead only studying the permanent price impact caused by the change of the inventory as we do in our paper,} assumes that the trades are split into smaller trades so that their effects disappear gradually; see \cite{GathSc} for an overview. %For other works on the optimal execution of a trade, we refer to  \cite{gokay2011liquidity}.
{In this paper, we consider an endogenously based liquidity model where the permanent price impact is due to an inventory (supply-side) change of the security. 
It is shown in \cite{Fukasawa2017} that the profit and loss (P\&L) process of any trading strategy can then be expressed as a non-linear stochastic integral or a $g$-expectation and the pricing and hedging can be done by solving a semi-linear PDE in the special case of a Markovian setting. \cite{Fukasawa2017} also provides a completeness condition under which any derivative can be perfectly replicated by a dynamic trading strategy, see also \cite{S14}.
% induced by the previous trades. 
Contrary to exogenously based-liquidity models, an endogenously liquidity-based model may give better economic understanding on the liquidity risk.} Non-utility based additive permanent market impact models are also
studied in many other works, see for instance \cite{bouchard2017hedging} and the references therein. 
Feedback effects with BSDEs are studied in \cite{CM}. In \cite{drapeau2019fbsde}, a Nash equilibrium for a market impact game is characterized in terms of a fully coupled system of FBSDEs. 

\vspace{2mm}
%In this paper, we consider an endogenously based liquidity model where the permanent price impact is due to an inventory (supply-side) change of the security.  It is shown in \cite{Fukasawa2017} that the profit and loss (P\&L) process of any trading strategy can then be expressed as a non-linear stochastic integral or a $g$-expectation and the pricing and hedging can be done by solving a semi-linear PDE in the special case of a Markovian setting. \cite{Fukasawa2017} also provides a completeness condition under which any derivative can be perfectly replicated by a dynamic trading strategy.
% induced by the previous trades. Contrary to exogenously based-liquidity models, an endogenously liquidity-based model may give better economic understanding on the liquidity risk. 

The closest price impact model to ours is the one introduced by Bank and
Kramkov~\cite{BK1,BK2} who allow the Market Makers to have utility functions of von
Neumann-Morgenstern type (see also the earlier works \cite{stoll1978supply,ho1981optimal,S07,S14}).
They show the existence of the representative
liquidity supplier and construct a nonlinear stochastic integral to describe the Large Trader's
P\&L. 
Our utility function stems from time-consistent convex risk measures represented as a $g$-expectation, and in terms of decision theory, it represents an ambiguity-averse preference with the exponential utility being the only intersection to von Neumann-Morgenstern utility functions.
The focus of this paper is on optimal portfolio choice which is not addressed in \cite{BK1,BK2,Fukasawa2017}. The special case of an exponential utility function for the investor and for the representative liquidity supplier is analyzed in \cite{anthropelos2018} without considering coupled FBSDEs or BSPDEs. {For the static case see \cite{ABG18}. In \cite{KP16a}, the authors show that the price system arising from an exponential utility function can be characterized as a multidimensional quadratic BSDE for which existence and uniqueness results are shown.}

\vspace{2mm}
 We remark that the terminal wealth in our case lies in a space of feasible terminal conditions of $g$-expectations and depends on the strategy in a complex and non-convex way. Due to the non-linearity in the target function, the control variable and the feedback effects, the mathematical proofs are delicate. {We first show that the wealth process corresponds to a non-linear Kunita integral where every admissible trading strategy induces a process, say $\mathcal{Z}$, which is based on a parametric family of $Z^y$ processes from the solution of a parametric BSDE. We prove that trading can be extended to continuous time by deriving from appropriate assumptions the existence of a version of $y\mapsto Z^y$ which can be shown to be continuous in the trading strategy and satisfies suitable growth conditions. Through calculus of variations, we derive a coupled FBSDE for the optimal wealth process which is necessary and sufficient for a trading strategy to be optimal. Existence of a solution under conditions covering most known examples is shown by proving weak compactness of a space of suitable $Z$'s or wealth processes. %or of the terminal wealth feasible to attain the supremum.
To the best of our knowledge such a $Z$-based optimal control approach has not been used yet to show existence of a solution of a coupled FBSDE. 

%thaicomment{%The solution existence of a coupled FBSDE in a Markovian setting can be linked to that of a parabolic PDE, see e.g. \cite{ma1994solving}. 
\vspace{2mm}
	
	While for decoupled quadratic FBSDEs there is a rich theory available, for \emph{coupled} FBSDEs the literature is much more sparse. For non-Markovian systems, existence for solutions of coupled FBSDEs
	for sufficiently small time horizons $T$ have been obtained by \cite{delarue2002existence} using a contraction
	method. Well-posedness of the system has been investigated by \cite{ma2015well} using the so-called decoupling field method introduced in \cite{douglas1996numerical}.  \cite{L17} study the
	well-posedness for multidimensional and coupled systems of FBSDEs with a generator that can be separated into a quadratic and a subquadratic part. Using Malliavin calculus arguments \cite{kupper2019multidimensional} obtains local and global existence and uniqueness results for multidimensional coupled FBSDEs for monotone generators with arbitrary growth in the control variable. \cite{Jackson} obtains some recent results for Markovian FBSDEs with quadratic growth. Existence results for multidimensional coupled FBSDEs with diagonally quadratic generators can be found in \cite{chen2023existence}. \cite{herdegen21} study equilibria of asset prices with quadratic transaction costs, and solve an arising coupled FBSDE using stochastic Riccati equations. In contrast to these works %more recent studies \cite{Jackson,chen2023existence}, 
	our forward process boasts a more versatile volatility term, denoted as $\cH$ in Equation \eqref{eq:bijective}, which depending on the utility function can fully depend on the variables $X, Y, M$, and $\omega$. %This unique feature empowers our model with heightened adaptability, allowing it to comprehensively capture a broader spectrum of financial dynamics and scenarios.}
%\thaicomment{To the best of our knowledge, our existence results represent a novel contribution to the existing literature. For instance, our paper distinguishes itself from reference \cite{kupper2019multidimensional} in several concrete ways. Firstly, whereas the forward process in \cite{kupper2019multidimensional} can depend on variable $Y$ but explicitly excludes any reliance on $Z$ (which corresponds to our $M$), our model embraces the influence of $Z$, thus providing a more comprehensive analytical framework. Secondly, the terminal condition in \cite{kupper2019multidimensional} imposes the requirement of Lipschitz continuity on the Markov process, a constraint not imposed by our $H_L$, thereby expanding the scope of processes within our analytical purview. Lastly, our driver function has the capacity to incorporate the variable $\omega$, a flexibility not found in the referenced paper, thus offering a more adaptable and versatile approach for modeling a wide array of scenarios.}
%\mitjacomment{Thai: State if or if not our results are covered?} our forward process's volatility term, denoted as $\cH$ in Equation \ref{eq:FBSDE}, is notably more versatile. It can fully depend on the variables $X,Y,Z$ and $\omega$, making the our model more challenging. 
Moreover, for the value function, we derive a dynamic programming principle and prove that under smoothness conditions the value function satisfies a certain BSPDE which can be connected to the previously derived FBSDE.
%, if the value function is smooth we then derive a BSPDE using the It{\^o}-Ventzel formula and give connections to previously derived FBSDE. 
The smoothness of the value function is then shown using duality methods.
%In Section 7 we show that the value function is smooth using duality results.
 \emph{All results} except in Section 7.2, are shown for an \emph{incomplete} financial market with the randomness being generated by a $d$-dimensional Brownian motion.} %with $n\leq d.$ 
 %This is possible through considering in the proofs the space of attainable $\mathcal{Z}$'s (instead of the set of admissible trading strategies) and through proving that under our conditions this set has a workable structure, see Proposition 2.5.} 

\vspace{2mm}
The remainder of the paper is organized as follows. Section \ref{sec:mod} 
outlines a model with permanent market impact, depicting hedging through $g$-expectations. We specify key assumptions validating our results and our expected utility maximization problem under permanent market impacts. Solutions are characterized by FBSDEs in Section \ref{sec:FBSDE}. FBSDE existence results are derived in Section \ref{se:existence}. The connection to BSPDEs is studied in Section \ref{se:BSPDE}. The results on regularity necessary for the existence of a solution of the BSPDEs together with examples in complete markets are presented in Section \ref{se:Regu}. The main proofs and extra technical results are reported in the appendix. Section \ref{sec: conclusion} concludes.

\section{The model setting}\label{sec:mod}

In a limit order book of specific asset, the roles of liquidity suppliers and liquidity demanders are not symmetric. Every liquidity supplier submits a price quote for a specific volume and trades with the other liquidity suppliers until an equilibrium is achieved. The remaining limit order form a price curve which is a nonlinear function in volume. Taking a Bertrand-type competition among liquidity suppliers into account, it would then be reasonable to begin with modelling the price curve as the utility indifference curve of a representative liquidity supplier. %, see for instance Bank and Kramkov~\cite{BK1,BK2}, Fukasawa and Stadje~\cite{Fukasawa2017} and Robertson et al. (2018). 
While Bank and Kramkov \cite{BK1,BK2} used Neumann-Morgenstern utility functions for the representative agent we will as in \cite{Fukasawa2017} use time-consistent convex risk measures instead. The exponential utility function assumed by  \cite{KP16a,anthropelos2018} is in the intersection of these two frameworks. Modulo a compactness assumption using a time-consistent convex risk measure is equivalent to using a $g$-expectation, providing a powerful stochastic calculus tool. A further advantage of our approach from an economic point of view is that ambiguity aversion is taken into account, see the discussion below.
In the present paper, we therefore simply assume that there is a representative liquidity supplier, called the Market Maker, who quotes a price for each volume based on the utility indifference principle and her utility is a $g$-expectation with a cash-invariance property. The existence of the representative agent under such utility
functions follows from  Horst et al.~\cite{horst2010}. While the cash invariance axiom might not be realistic for an individual investor, market makers are often financial institutions for which a cash invariance axiom seems more reasonable.
If the driver of the $g$-expectation is a linear function, then 
the price curve becomes linear in volume and
we recover the standard framework of financial engineering.
For the individual investor we assume a strictly concave Neumann-Morgenstern utility function $U:\mathbb{R}\to \mathbb{R}$ satisfying standard conditions. 

{In Subsection \ref{marketsetting} we first present our market setting. 
In Subsection \ref{pricingrule} we describe the utility function of a market maker (which is given by a $g$-expectation) under some assumptions on the function $g$ and the {traded} securities. In Subsection \ref{tradingimpact} we give the pricing rule of a market maker which arises as utility indifference price. %We further show that our assumptions lead to a workable structure of a space of $\mathcal{Z}$'s which arise from the optimal control variables. 
We show that, the P$\&$L process of a trader can be described as a linear  stochastic integral and a non-linear Lebesgue integral in $\mathcal{Z}$. In Subsection \ref{sec:Tech} we define the expected utility maximization problem of the Large Trader.}
%\thaicomment{If we want to shorten the paper, the above paragraph can be removed?}
\vspace{2mm}
\subsection{The basic setting}\label{marketsetting}
%% The setting below follows Fukasawa and Stadje~\cite{Fukasawa2017}.
We assume zero risk-free rates meaning in particular that cash is risk free and does not have any bid-ask spread.
Let $T>0$ be the end of an accounting period.
Each agent evaluates her utility based on her wealth at $T$.
Consider $n$ securities $S=(S^1,\ldots ,S^n)^\intercal$ whose values at time $T$ are exogenously determined, where '$\intercal$' denotes transpose.
We denote the value by $S$  and regard it as an $\mathcal{F}_T$-measurable random variable defined on a filtered probability space $(\Omega, \mathcal{F}, \mathbf{P}, (\mathcal{F}_t)_{t\in [0,T]})$ with $(\F_t)$ being the completion of the filtration generated by a $d$-dimensional Browninan motion $W=(W^1,\ldots,W^d)^\intercal$ with $n\leq d$
satisfying the usual conditions. %%Throughout this paper equalities and inequalities should hold almost surely. 
The security $S^i$ can for instance be a zero-coupon bond, an asset backed security, or a derivative with an underlying which is not traded.
The price of this security at $T$ is trivially $S$, but the price at  $t < T$ should be $\mathcal{F}_t$-measurable and will be endogenously determined by a utility-based mechanism. 
There are two agents in our model: A Large Trader and a Market Maker.
The Market Maker quotes a price for each volume of the security.
She can be risk-averse and so her quotes can be nonlinear in volume and
depend on her inventory of this security.
The Large Trader refers to the quotes and makes a decision. She cannot avoid affecting the quotes by her trading due to the inventory consideration of the Market Maker, and seeks an optimal strategy under this endogenous market impact.
\subsection{The evaluation method of the market maker and examples}\label{pricingrule}
As the pricing rule of the Market Maker, our model adopts the utility indifference principle, using an evaluation $\Pi$ given by the solution of a $g$-expectation defined as follows.
Let $g(t,\omega,z):[0,T]\times \Omega  \times \mathbb{R}^d \to
\mathbb{R}$ be a $\mathcal{P}\otimes \mathcal{B}(\mathbb{R}^d)$
measurable function, where $\mathcal{P}$ is the progressively measurable $\sigma$ field,
such that 
$z \mapsto g(t,\omega,z)$ is a convex function\footnote{{In the sequel we will, in this case, simply state that $g(t,z)$ is convex.}} with $g(t,\omega,0)=0$ for each $(t,\omega ) \in  [0,T]\times \Omega $. As usual the $\omega$ will typically be suppressed. For a stopping time $\tau$ we denote by $\cD_\tau$ a linear space of
$\cF_\tau$-measurable random variables to be specified below. We denote by $L^0(\cF_\tau)$ the space of $\cF_\tau$-measurable random variables, by $L^2(\cF_\tau) = L^2(\Omega, \cF_\tau, \P)$ the space of $\cF_\tau$-measurable random variables which are square integrable with respect to $\mathbf{P}$, by $L^2(\d \P \times \d t)$ the space of progressively measurable processes which are square integrable with respect to $\d\mathbf{P}\times \d t$ and by $L^\infty (\cF_\tau)$ the space of $\cF_\tau$ -measurable random variables which are bounded. Products of vectors will be understood as vector products. Equalities, inequalities and inclusions are understood a.s. unless stated otherwise. Throughout the rest of this paper, we always assume either one of the following conditions:

\vspace{2mm}
{\em
\noindent{\bf (Hg)} %$g(t,z)=\alpha |z|^2+l(t,z)$ where $\alpha>0$ is convex in $z$ and $l$ is Lipschitz in $z$ (also uniformly in $t$ and $\omega$) and continuously differentiable in $z$ with $l(t,0)=0$. 
{$g(t,z)$ is convex in $z$, and there exists a $K>0$ such that $|g(t,z)|\leq K(1+|z|^2)$, $g(t,0)=0$ and $$|g(t,z_1)-g(t,z_0)|\leq K(1+|z_1|\vee |z_0|)|z_1-z_0|.$$} Furthermore, 
$$\cD_\tau=\big\{H \in L^0(\mathcal{F}_\tau)\,\vert\, \E[\exp(a\vert H\vert)] < \infty \text{ for all } a >0\big\}.$$
}
%% for the derivative of $g$ with respect to $z$ it holds that $\sup_{t\in[0,T]} \vert g_z(t,z)\vert\le K(1+\vert z\vert)$ for some $K$ independent of $t$ and
%\thaicomment{The condition $\sup_{t\in[0,T]} \vert g_z(t,z)\vert\le K(1+\vert z\vert)$ seems to be a direct consequence of the form $g(t,z)=\alpha |z|^2+l(t,z)$ is convex, $l$ is Lipschitz in $z$??}
%{\em
%	\noindent{\bf (Hg)}: $g(t,z)$ is convex, quadratic growth and continuously differentiable in $z$ with $g(t,0)=0$. Furthermore, $\sup_{t\in[0,T]} \vert g_z(t,z)\vert\le K(1+\vert z\vert)$ for some $K$ independent of $t$ and $\cD_T=L^{\infty}(\Omega, \mathcal{F}_T,\P)$.
%}
\vspace{2mm}
{\em
\noindent{\bf(HL)} $g(t,z)$ is convex in $z$, Lipschitz in $z$ (also uniformly in $t$ and $\omega$)  with $g(t,0)=0$ and $\cD_\tau=L^{2}(\Omega, \mathcal{F}_\tau,\P)$.
}	
\vspace{5mm}
%\vspace{3mm}
%{\em
%\noindent{\bf Condition 1}: 
%\begin{itemize}
	%%\item The filtration $\cF=(\cF_t)_\zs{0\le t\le T}$ is generated by a standard Brownian motion $W$. 
	%\item Let $g:\Omega\times[0,T]\times\bbr\to\bbr$ be a $\cR\otimes \cB(\bbr)$ measurable function such that for each $(\omega,t)\in \Omega\times[0,T]$, $z\mapsto g(t,z)(\omega)$ is a smooth convex function in $z$ with $g(t,0)=0$.
	%\item For each $X\in\cD_T$ there exists a progressively measurable $Z(X)$ such that
	%\begin{eqnarray}
	%\E[\int_0^T\vert Z_t(X)\vert^2 \d t]<\infty 
	%\label{eq:}
	%\end{eqnarray}
	%and 
	%\begin{equation}
	%X=\Pi_t(X)+\int_t^T g(s,Z_s(X))\d s-\int_t^T Z_s(X) \d W_s,
	%\label{eq:}
	%\end{equation}
%\end{itemize}
%}
%Note that under the Condition 1, the $g$-expectation  $\Pi$ is convex.
%
%\vspace{2mm}
%{\em
%\noindent{\bf Condition 2}: There exist $\Omega_0$ with $\P(\Omega_0)=1$ and $\cR\otimes\cB(\bbr)$ measurable function 
%$$
%\cZ:\Omega\times[0,T]\times \bbr\to \bbr%; \quad \wt{\cZ}:\Omega\times[0,T]\times \bbr\to \bbr
%$$
%such that 
%\begin{equation}
%\cZ(\omega,t,y))=Z(H_{\mathrm{M}}-yS)\quad \mbox{for all}\quad (\omega,t,y)\in \Omega_0\times[0,T]\times \bbr. 
%\label{eq:Z1}
%\end{equation}
%For our convenience we denote $\cZ^Y_t(\omega):=Z(\omega,t,Y_t(\omega))$.
%}
%\vspace{2mm}

It is well-known, see e.g. \cite{kobylanski2000,BH2} that under {\bf (HL)} or {\bf (Hg)} for each $H \in \mathcal{D}_T$,
there exist unique progressively measurable square integrable processes $ (\Pi(H),Z(H)) = (\Pi(H),Z^1(H), \ldots, Z^d(H))$ such that
%\begin{equation*}
%\E\left[\int_0^T|Z_t(H)|^2 \mathrm{d}t\right] < \infty \quad\mbox{and}\quad
%\sup_{0 \leq t \leq T}|\Pi_t(H)| \in \mathcal{D}_T,
%\end{equation*}and 
for all $t\in [0,T]$, it holds that
\begin{equation}
H = \Pi_t(H) + \int_t^T g\bigl(s,Z_s(H)\bigr)\mathrm{d}s - \int_t^T Z_s(H) \mathrm{d}W_s,
\label{bsdeX}
\end{equation}
where
$
\int_t^T Z_s(H) \mathrm{d}W_s = \sum_{j=1}^d \int_t^T Z_s^j(H) \d W_s^j.
$
Equation \eqref{bsdeX} in differential form can equivalently be written as 
\begin{equation} 
\label{bsdeX2}
\mathrm{d}\Pi_s(H)= g\bigl(s,Z_s(H)\bigr)\mathrm{d}s - Z_s(H) \mathrm{d}W_s, \quad \Pi_T(H)=H .
\end{equation}
\eqref{bsdeX} or \eqref{bsdeX2} is also called a backward stochastic differential equation (BSDE) with driver function $g$ and terminal condition $H$. Sometimes $\Pi(H)$ is also called a $g$-expectation of $H$. %%We remark that $\int_0^{\cdot} Z_s(X) \mathrm{d} W_s$ is a BMO martingale under ${\bf (HL)}$. 
{Under our assumptions it is shown for instance in \cite{jiang2008convexity} that }$\Pi$ is concave, cash invariant (meaning that $\Pi_t(H+m)=\Pi_t(H)+m$ for $m\in {\cD_t}$) and monotone (meaning that $\Pi_t(H_1)\leq\Pi_t(H_2)$ if $H_1\leq H_2$). Furthermore, $\Pi$ satisfies the tower property (time-consistency).

\begin{example}[$g$-expectation]
	\label{example0}
 %\notag
\begin{enumerate}
	%\item
%The simplest example of a $g$-expectation is if $g=0$ and thus
%\begin{equation}\label{rn0}
%\Pi_t(H) = \E[H|\mathcal{F}_t]
%\end{equation}
%with $ \mathcal{D}_t = L^2(\Omega, \mathcal{F}_t,\P)$. This evaluation can be interpreted as the orthogonal projection of
%future cash-flows.
\item
The simplest example is given by the case where $g$ is linear in $z$, i.e. $g(t,z)=-\eta_t  z$ for a bounded progressively measurable process $\eta$. By the Girsanov Theorem, 
\begin{equation}\label{rn}
\Pi_t(H) = \E^\Q[H\mid \mathcal{F}_t]
\end{equation}
with $\Q$ having the Radon-Nikodym derivative with stochastic logarithm $\eta$, i.e., 
\begin{equation}\label{qq}
\frac{\d\Q}{\d\P}:=\exp\{\int_0^T -\eta_s dW_s-\frac{1}{2}\int_0^T |\eta_s|^2 ds\}.
\end{equation}
%In this case we set   $ \mathcal{D}_t = L^2(\Omega, \mathcal{F}_t,\P)$. %When $p=2$,
%this evaluation can be interpreted as the orthogonal projection of
%future cash-flows.
 \item
Another example of $g$-expectation is the following exponential utility-based equivalent:
$$\Pi_t(H) = -\frac{1}{\gamma} \log \E^\Q\left[ \exp\left(-\gamma H\right) \mid \mathcal{F}_t\right], $$
with $  \mathcal{D}_t =\{H \in L^0(\mathcal{F}_t)\vert\, \E[\exp(a\vert H\vert)] < \infty \text{ for all } a >0\} $,
where $\gamma > 0$ is a parameter of risk-aversion {and $\Q$ is defined through \eqref{qq}}. In this case $\Pi$ is a $g$-expectation with driver function $g(t,z)=\frac{1}{2}\gamma|z|^2-\eta_t z$, see Barrieu and El
Karoui~\cite{carmona2008}.
%By letting $\gamma \to 0$, we recover the previous example.
%By letting $\gamma \to \infty$, we have
%\begin{equation*}
%\Pi_t(H) = \inf_{\Q}\left\{\E^{\Q}[H|\mathcal{F}_t] :  \Q \sim \P, \Q=\P \text{
%	on } \mathcal{F}_t\right\},
%\end{equation*}
%which is the essential infimum value of $H$ under the conditional
%probability given $\mathcal{F}_t$. 
{\item A driver function $g$ is called positively homogeneous if $g(t,\lambda z)=\lambda g(t,z)$ for all $\lambda>0$. Examples of positively homogeneous drivers can be found for instance in \cite{KLL15,LSSS20} corresponding to so called coherent risk measures, see \cite{Artzner99}. For example, if the market maker is not sure about the reference model and considers a robust expectation over a set of alternative probability measures inducing a drift in the Brownian Motion which is at most allowed to deviate by $\delta$ in the Euclidean norm,  one obtains $g(t,z)=\delta|z|-\eta_t  z$.  }
\item
We remark that there is a close connection between $g$-expectations and convex risk measures, which (modulo a change of sign) are cash-invariant, monotone and convex, see e.g., Barrieu and El
Karoui~\cite{carmona2008}. %or Delbaen et al. \cite{delbaen2010representation}. 
%A convex risk measure satisfies the following axioms. 
%For any $H_1, H_2 \in \mathcal{D}_T$:
%\begin{enumerate} 
%	\item Normalization: $\rho_{\tau}(0) = 0$,
%	\item Cash-Invariance: $\rho_\tau(H_1+H_2) = \rho_\tau(H_1) + H_2$ if $H_2 \in \mathcal{D}_\tau$,
%	\item Convexity: $\rho_\tau(\lambda H_1 + (1-\lambda)H_2) \leq 0 $
%	for all $\lambda \in [0,1]$ if 
%	$\rho_\tau(H_1) \leq 0$ and $\rho_\tau(H_2) \leq 0$,
%	\item Time-Consistency: $\rho_\tau(H_1) \geq \rho_\tau(H_2)$ if there exists 
%	$\sigma \geq \tau$ such that
%	$\rho_\sigma(H_1) \geq \rho_\sigma(H_2)$.
%\end{enumerate}
 A $g$-expectation is a convex risk measure and satisfies time consistency if and only if $-g$ is convex and $g(t,0)=0$, see Jiang \cite{jiang2008convexity}. In particular, $-\Pi(-\cdot)$ is a convex risk measure.
It is worth noting that 
under additional compactness or domination assumptions, {also} every time-consistent convex
risk measure corresponds to a $g$-expectation. 
%in the sense that there exists $g$ such that
%$\Pi$ satisfies (\ref{bsdeX}).
For these and other related results, see %Barrieu and El
%Karoui~\cite{carmona2008}, 
Coquet et al.~\cite{CHMP} 
%Briand and Hu~\cite{BH1,BH2}, 
%\item In the theory of no-arbitrage pricing, attempts have been made to narrow the no-arbitrage bounds by restricting the set of pricing kernels considered. One of these approaches is the good-deal bounds ansatz introduced in Cochrane and Sa\'{a}-Requejo~\cite{CS00} which corresponds to excluding in the no-arbitrage bounds pricing kernels which induce a too high Sharpe ratio. Upper good-deal bounds can be modelled via $g$-expectations with positively homogeneous drivers. We refer to \cite{Fukasawa2017} for further discussions. %The good-deal bound is then %...
%
 %Specifically,suppose that we have a $d-$dimensional Brownian motion $W$ generating the economic noise and that the dynamics of some basic risky assets are given by
%\begin{equation*} 
%\frac{dX^i_t}{X_t^i}=\mu^i dt+\sigma^i dW_t
%,\,\,\,\,\,\,\,\,i=1,\ldots,k. \end{equation*}
%We further suppose that the interest rate of the bond is zero. Let $A:=(\sigma^1,\ldots,\sigma^k)$ and $b:=-\mu^\intercal=-(\mu^1,\ldots,\mu^k)^\intercal$.
%Let $P_B(0)$ be the projection of $0$ onto the
%set $B:=\{x : Ax=b\}$ in the Euclidean $|\cdot|$ norm, and define $P_{\text{Kernel}(A)}(z)$ accordingly as the projection of $z$ in the $|\cdot|$ norm onto the space given by the kernel of the matrix $A$.
%One can prove (see Kr\"{a}tschmer et al~\cite{KLL15}) that the good-deal bound price is given by a $g$-expectation following (\ref{bsdeX}) with driver function 
%\begin{align*}
%g(t,z)=\sqrt{\Lambda-|P_B(0)|^2}\Big|P_{\text{Kernel}(A)}(z)
%\Big|+zP_B(0).
%\end{align*}
\end{enumerate}
\end{example}

\subsection{The pricing rule and the trading with permanent market impact}\label{tradingimpact}

{Suppose that the Market Maker has initially an endownment whose cash-flow at time $T$ is represented by a bounded random variable $H_\mathrm{M}$. {Additionally,} for each quantity of the risky assets $S=(S^1,\cdots,S^n)$ with $d\geq n\geq 1$ {the Market Maker} quotes  a price.} {Below, we} assume that $S$ is either square integrable if ($\textbf{HL}$) holds or bounded if alternatively ($\textbf{Hg}$) holds. If the Market Maker at time $t \in [0, T]$ is holding $z^i$ units of the security $S^i$ in question besides $H_{\mathrm{M}}$, then her utility is measured  as $\Pi_t(H_{\mathrm{M}} + \sum_{i=1}^n z^iS^i)=\Pi_t(H_{\mathrm{M}} + zS)$.
According to the utility indifference principle,
the Market Maker quotes {at time $t$} a selling price for $y=(y^1,\ldots,y^n)$ units of the security by
\begin{equation}\label{pzy}
\begin{split}
P_t(z,y) :=& \essinf\left\{ p \in {\cD_t} :
\Pi_t(H_{\mathrm{M}}+zS-yS + p) \geq \Pi_t(H_{\mathrm{M}}+zS) \right\} \\
=& \Pi_t(H_{\mathrm{M}} + zS) - \Pi_t(H_{\mathrm{M}} + (z-y)S).
\end{split}
\end{equation}
For the above equality
we have used the cash invariance property. Note that in the risk-neutral case \eqref{rn} in the previous example, we have
%\begin{equation*}
$P_t(z,y)=\sum_{i=1}^n y^i \,\E^\Q[S^i\mid\mathcal{F}_t] = y \,\E^\Q[S\mid\mathcal{F}_t],$
%\end{equation*}
where we used vector notation.
%In general, 
%the price depends on the inventory $z$ of the securities, 
%which describes permanent market impact.
  %Our model does not allow any price manipulation in the sense
%that an immediate round-trip cost is always $0$, that is,
%\begin{equation*}
%P_t(z,y) + P_t(z-y,-y) = 0.
%\end{equation*}
%For all $t$ and $z$, $P_t(z,y)$ is a convex function of $y$ with $P_t(z,0) = 0$ by the third axiom of $\Pi$ (concavity).
%This implies in particular that
%\begin{equation*}
%-P_t(z,-y) \leq P_t(z,y) 
%\end{equation*}
%for any $y$ and $z$, which means that the selling price for an amount is higher than or equal to the buying price for the same amount.
%This represents bid-ask spread that is a measure of market liquidity.

Let
$\Theta_0$ be the set of simple $n$-dimensional left-continuous processes $\theta$
with $\theta_{0} = 0$.
The Large Trader is allowed to take any element $\theta \in \Theta_0$ as her trading strategy.
The price for the $y$ units of the security at time $t$ is
$P_t(-\theta_t,y)$. This is because the Market Maker holds $-\theta_t$ units of the
security due to the preceding trades with the Large Trader. 
Then the profit and loss (P\&L) at time $T$  associated with $\theta \in \Theta_0$  (i.e., corresponding to the self-financing strategy $\theta$) of the Large Trader
is given by
\begin{equation}\label{I}
\mathcal{I}(\theta) := \theta_TS - \sum_{0 \leq t < T} P_t(-\theta_t, \Delta \theta_t). 
\end{equation}
Due to Proposition \ref{propextend} below, $\mathcal{I}(\theta)$ has the form of a nonlinear
stochastic integral studied in Kunita~\cite{kunita1997stochastic}.
Note that in the risk-neutral case \eqref{rn}, $\mathcal{I}(\theta) = \theta_TS_T - \sum_{0 \leq t < T}  \Delta \theta_t S_t= \int_0^T \theta_t \mathrm{d}S_t$ by integration-by-parts, where $S_t = \E^\Q[S\mid\mathcal{F}_t].$ For any $y \in \mathbb{R}^{1\times n}$, using the BSDE representation \eqref{bsdeX}, we set $\Pi^y = \Pi(H_\mathrm{M} -yS)$, $\cZ^y = Z(H_\mathrm{M}-yS)$, {i.e., $(\Pi^y, \cZ^y)$ is the solution of the BSDE \eqref{bsdeX} with terminal condition $H_\mathrm{M}-yS$}. {Note that for each fixed $y,\cZ^y$ is only defined $\d  \P \times \d t$ a.s.}
To extend trading to general predictable trading strategies we need the following assumption. 

%\textcolor{blue}{
%\textbf{Assumption (A)} 
%There exists a $\mathcal{P} \otimes \mathbb{B}(\mathbb{R}^n)$- measurable, mapping $Z:\Omega\times [0,T]\times \mathbb{R}^n \mapsto \mathbb{R}^d $ such that  $\cZ^y_t(\omega)=Z(\omega,t,y)$ $\d \mathbf{P}\times \d t$ almost surely for each fixed $y\in\mathbb{R}^n$. Furthermore, $Z$ can (and will in the sequel) be chosen such that the mapping $y\mapsto Z(\omega,t,y)$ is continuous $\d \P\times \d t$ a.s., and ~ $||\sup_t |Z(t,\omega,y)|\hspace{0.05cm}||_\infty\leq K(1+|y|)$ for some $K>0$.
%}\\
\vspace{2mm}
{
\textbf{Assumption (B)}: One of the following assumptions hold:
\begin{itemize}
	\item[\textbf{(B1)}] $g$ does not depend on $\omega$ and is three times continuously differentiable with respect to $z$ with continuous derivatives.  
	Define the diffusion
	process $\cR=(\cR^1,\ldots,\cR^d)$ to be the solution to the following SDE:  
	$\cR_u^{t,r} = r,\ \ 0\leq u\leq t,$ and
	\begin{eqnarray}\label{R}
		\d \cR_u^{t,r} &=& \mu(u,\cR_u^{t,r})\,\d s+\sigma\,\d W_u,\ \ t\leq u\leq
		T,\label{SDE}\end{eqnarray} where $\mu:[0,T]\times \mathbb{R}^d\rightarrow \mathbb{R}^d$ is \tn{three times} continuously differentiable %with respect to its second argument 
		with bounded derivatives, and $\sigma$ is in $\mathbb{R}^{d\times d}$. {$\mathcal{R}$ could for instance reflect the returns or dividends of $d$ risky assets.} %% We remark that $\nabla_pP_s^{t,p}$ and $(\nabla_pP_s^{t,p})^{-1}$ are bounded and bounded away from zero.
	We assume in addition that $S$ and $H_\mathrm{M}$ are of the form 
	$S=s(\cR_T) \ \text{  and  } \
	H_\mathrm{M}=h^{\mathrm{M}} (\cR_T),$ with $\cR_0=(r^1_0,\ldots,r^d_0)$, and $s:\R^d\to \R^n$ and $h^{\mathrm{M}}:\R^d\to \R$ are functions in $\mathcal{C}^3$ with bounded \tn{first} derivatives, \tn{and the second and third derivatives grow only polynomially}. Additionally assume that one of the following conditions hold:
	\begin{itemize}
		\item[i)] %\tn{$\mu_{RR}$ is bounded}, and 
		$d=1$ or $\sigma$ is diagonal, $g$ is of the form $g(t,z)=\sum_{i=1}^ng^i(t,z^i)+g^{n+1}(t,z^{n+1},\ldots,z^d)$, and $	H_\mathrm{M}=\sum_{i=1}^n h^{\mathrm{M},i} (\cR_T^{i})+H_{\mathrm{M}}^{\perp}$ where $H_{\mathrm{M}}^{\perp}=h^{\perp}(W^{n+1},\ldots,W^d)$ is assumed to be measurable with respect to the filtration generated by $(W^{n+1},\ldots,W^d)$;
		\item[ii)] or $g$ is linear of quadratic;
		\item[iii)] or the solution of the BSDE \eqref{diff} in the appendix has a unique solution which is jointly bounded for each $y$;
		\item[iv)] or $\cZ^y_t$ has a modification such that $y\mapsto\cZ^y_t(\omega)$ is continuous on $y$ for fixed $(t,\omega)\ \d \mathbf{P} \times \d t$ a.s; % \tn{and $Z(	H_\mathrm{M})$ is bounded}.
	\end{itemize}
	\item[\textbf{(B2)}] $g$ is positively homogeneous in $z$, $n=1\leq d$, $H_\mathrm{M}=0$. \tn{Furthermore,  $Z(\pm S)$ is either bounded or has continuous paths.}   %and $Z(\pm S)$ is uniformly bounded\footnote{\textcolor{blue}{e mappis is for instance the case, if $S=s(\mathcal{R}_T)$ and $s'$ is bounded, with $\mathcal{R}$ as in \eqref{R}.}}.
\end{itemize}}
{
The following proposition is shown in the appendix by analyzing certain uncoupled multidimensional FBSDEs. 
%While it is well known that such BSDEs may not admit a global solution (see the counterexamples provided in \cite{FG11,F14}), in our case the driver of the $i$-th component only depends on the $i$-th component of the stochastic integrand. Under our assumptions on the terminal condition, a bounded global solution with a bounded stochastic integrand then exists, see \cite{CK14,CK15,L15,L17} for similar arguments. Different multidimensional quadratic FBSDEs were also analyzed in \cite{KP16a,KP16b}, studying for instance equilibria for related price impact models or questions of stability.
}
%\begin{remark}
%
%In the case of a one dimensional risky asset, it is possible to, instead of (H1), assume that $g$ is positively homogenous, {which %$g$ being positively homogenous 
%is actually equivalent modulo a compactness assumption to $\Pi$ being a coherent risk measure, see for instance \cite{CHMP}. Since then the mathematical analysis is rather different complicating the exposition, we will not consider this case in this paper.} % and instead allways assume that (H1) holds.} 
%\end{remark}
%We suppose that $g$. Furthermore, the first order derivatives of $h^{\mathrm{M},i}$ and $s^i$ %, either a) $h^\mathrm{M}_{rr}$ and $s_{rr}$ or b) 
%	$h^\mathrm{M}_r$ and $s_r$, 
%exist and are bounded for $i=1,\ldots,n$. 
%We denote these assumptions by (H1). 
%\thaicomment{here it would be clearer to mention that (HL) or (Hg) are fulfilled for all $g^i$?}
\begin{proposition}\label{propmarkov}
{{Under Assumption \textbf{(B)}}, there exists a $\mathcal{P} \otimes {\mathcal{B}}	(\mathbb{R}^n)$- measurable, mapping $Z:\Omega\times [0,T]\times \mathbb{R}^n \mapsto \mathbb{R}^d $ such that  $\cZ^y_t(\omega)=Z(\omega,t,y)$ $\d \mathbf{P}\times \d t$ almost surely for each fixed $y\in\mathbb{R}^n$. Furthermore, $Z$ can (and will in the sequel) be chosen such that the mapping $y\mapsto Z(\omega,t,y)$ is continuous $\d \P\times \d t$ a.s. {Finally}, under \textbf{(B1)}} ~ $||\sup_t |Z(t,\omega,y)|\hspace{0.05cm}||_\infty\leq K(1+|y|)$ for some $K>0$, {while under \textbf{(B2)}, $\cZ_t^y=\tn{\cZ(\omega,t,y)}=|y|Z_t(-sgn(y)S)$, where $sgn(y)=1$ if $y>0$, $sgn(y)=-1$ if $y<0$, and $sgn(0)=\pm 1$.} 
%In particular, the continuity condition of Proposition \ref{propextend} is satisfied and trading may be extended to continuous time. % and the differential quotient from Definiton \ref{de:der} converges not only in $L^2(\d \mathbf{P}\times \d t)$ but also $\d \mathbf{P}\times \d t$ a.s. 
\end{proposition}

	Denote the $\mathcal{P}\otimes\mathcal{B}(\R^n)$-measurable 
	image space of $Z$ as
	$${\text{Im}(Z)}=\{ Z(t,\omega,y)\vert t\in[0,T], \omega \in \Omega,y\in \mathbb{R}^n\}.$$
	%(\text{Im}(Z(t,\omega,\cdot)))_{t\in[0,T],\omega\in\Omega}$. % 
 Set $	\text{Im}(Z(t,\omega,\cdot)):=\text{Im}_t(Z(\omega)):=\big\{Z(t,\omega,y)\big|y\in\mathbb{R}^n\big\}.$
%\end{equation*}

The next proposition shows how to extend trading from simple strategies to continuous time. 
\begin{proposition}
	\label{propextend}
	%Suppose that $y\mapsto \cZ(t,\omega,y)$ is continuous $\d \P\times \d t$ a.s. Then for any progressively measurable square-integrable process $\theta$ and  $\theta^{m} \in \Theta_0$, with $\theta^{m} \rightarrow \theta$ in $L^2(\d \P \times \d t)$ as $m\to\infty$ and $|\cZ^{\theta^{m}}|^2$ being uniformly integrable,
{{Suppose that Assumption \textbf{(B)} holds} and let $\mathcal{I}(\theta^m)$ be the Market Maker's profit-loss defined in \eqref{I} for simple trading strategies $\theta^m \in \Theta_0$. If $\theta\in L^2(\d \P \times \d t)$ and $\theta^m \to \theta$ in $L^2(\d \P \times \d t)$, then we have}
	\begin{equation*}
	\lim_{{m}\to\infty} \mathcal{I}(\theta^{{m}}) = \mathcal{I}(\theta)
	:= H_\mathrm{M} - \Pi_0(H_\mathrm{M}) - \int_0^T g(t,\cZ^\theta_t)\mathrm{d}t
	+ \int_0^T \cZ^\theta_t\mathrm{d}W_t,
	\end{equation*}
	where $\cZ^\theta_t(\omega): = Z(\omega, t, \theta_t(\omega))$ and the limit holds {in $L^1$ under \textbf{(B1)}, and in probability under \textbf{(B2)}.}
\end{proposition}
% \noindent\textbf{Proof of Proposition \ref{propextend}.}
\proof For the sake of exposition let us recall the proof of Lemma 1 in \cite{Fukasawa2017} showing the integral representation above for a fixed admissible strategy $\theta^m\in\Theta_0$. Denote the discontinuity points of $\theta^m \in \Theta_0$ by
$
0 \leq {t}_1 < {t}_2 < \cdots.$
Let $l$ be the number of the discontinuity points, ${t}_0 = 0$ and
${t}_k = T$ for  $k \geq l+1$.
By definition, using that $\Pi_T(H_\mathrm{M}-\theta^m_TS) = H_\mathrm{M}-\theta^m_TS$ and
$\theta^m_0 = 0$
\begin{equation*}
	\begin{split}
		\mathcal{I}(\theta^m) = &
		\theta^m_TS - \sum_{0 \leq t < T}\bigl(\Pi_t(H_\mathrm{M} -\theta^m_tS ) - \Pi_t(H_\mathrm{M}-\theta^m_{t+}S)\bigr) \\
		=& \theta^m_TS - \sum_{j=1}^l
		\bigl(\Pi_{{t_j}}(H_\mathrm{M} -\theta^m_{{t_j}}S ) - \Pi_{{t_j}}(H_\mathrm{M}-\theta^m_{{t_{j+1}}}S)\bigr) \\
		=& H_\mathrm{M} - \Pi_0(H_\mathrm{M}) - 
		\sum_{j=0}^l \bigl(\Pi_{{t_{j+1}}}
		(H_\mathrm{M} -\theta^m_{{t_{j+1}}}S ) - \Pi_{{t_j}}(H_\mathrm{M}-\theta^m_{{t_{j+1}}}S)\bigr).
	\end{split}
\end{equation*}
Again, by definition,
$\Pi_{{t_{j+1}}}(H_\mathrm{M}-yS) - \Pi_{{t_j}}(H_\mathrm{M}-yS)
= \int_{{t_j}}^{{t_{j+1}}} g(s,Z^y_s)\mathrm{d}s
- \int_{{t_j}}^{{t_{j+1}}}Z^y_s\mathrm{d}W_s$.
Since $\theta^m$ is a simple left-continuous process,
$\theta^m_{{t_{j+1}}}$ is $\mathcal{F}_{{t_j}}$ measurable and so,
we can substitute $y = \theta^m_{{t_{j+1}}}$ to obtain
\begin{align}
	\label{star2}
	\mathcal{I}(\theta^m)
	&=  H_\mathrm{M} - \Pi_0(H_\mathrm{M}) - 
	\sum_{j=0}^l \left(\int_{{t_j}}^{{t_{j+1}}} g(s,Z^{\theta^m}_s)\mathrm{d}s
	- \int_{{t_j}}^{{t_{j+1}}}Z^{\theta^m}_s\mathrm{d}W_s \right)\nonumber \\
	&	= H_\mathrm{M} - \Pi_0(H_\mathrm{M}) - \int_0^T g(t,\cZ^{\theta^m}_t)\mathrm{d}t
	+ \int_0^T \cZ^{\theta^m}_t\mathrm{d}W_t.
\end{align}
%	which implies the result. 
%\end{equation*}
Now, by the continuity of $\cZ$ in $y$ {shown in Proposition \ref{propmarkov}} we have that $\cZ^{\theta^m} \rightarrow \cZ^{\theta}$  in measure with respect to $\d \P \times \d t$. {Note that under \textbf{(B2)} we can use a localization argument to assume without loss of generality that $Z(\pm S)$ is bounded. Hence, Proposition \ref{propmarkov} and the uniform integrability of $|\theta^m|^2 $ give that $|\cZ^{\theta^m}|^2$ is uniformly integrable \tn{w.r.t. $L^1(\d \P \times \d t)$}, under both {\textbf{(B1)} or \textbf{(B2)}}}. Thus, \tn{passing to the limit in \eqref{star2}}, convergence actually holds in $L^2(\d \P \times \d t)$ if $g$ is Lipschitz, and in $L^1(\d \P \times \d t)$ if $g$ grows at most quadratically. 
%Passing to the limit in (\ref{star2}) yields the proposition. 
\endproof

\vspace{2mm}
\begin{remark}\label{pl} Using similar arguments, the profit and loss until time $t$ measured in terms of its liquidation value is given by
\begin{equation*}
\mathcal{I}_t(\theta):=\Pi_t(H_\mathrm{M}) - \Pi_0(H_\mathrm{M}) - \int_0^t g(t,\cZ^\theta_t)\mathrm{d}t+ \int_0^t \cZ^\theta_t\mathrm{d}W_t.
\end{equation*}
\end{remark}

\begin{remark}\label{HM}
$(\Pi_t(H_\mathrm{M}),Z_t(H_\mathrm{M}))$ solves the BSDE
\begin{equation*}
H_{\mathrm{M}} = \Pi_t(H_{\mathrm{M}}) + \int_t^T g\bigl(s,Z_s(H_{\mathrm{M}})\bigr)\mathrm{d}s - \int_t^T Z_s(H_{\mathrm{M}}) \mathrm{d}W_s.
\end{equation*}
By definition $\mathcal{Z}_t^0=Z_t(H_\mathrm{M})$ and in particular $\mathcal{I}_t(0)=0$. Finally we remark that if $H_\mathrm{M}=0$ (i.e. if the Market Maker does not have an endowment), we clearly have $(\Pi_t(H_\mathrm{M}),Z_t(H_\mathrm{M}))=(0,0)$ since $g(t,0)=0$. 
\end{remark}
{
	\begin{proposition} \label{boundedHM}
		Under Assumption \textbf{(B)}, $Z(H_\mathrm{M}) $ is bounded.
\end{proposition}
\noindent Throughout the rest of this paper we will always assume that Assumption \textbf{(B)} holds.
}

\subsection{Expected utility maximization} \label{sec:Tech}
Motivated by Proposition \ref{propextend} we define the set of admissible strategies as
\begin{equation*}
\Theta := 
\left\{ \theta: \Omega\times [0,T] \to \mathbb{R}^d
\text{ progressively measurable with }
\E\bigg[\int_0^T |\cZ^\theta_t|^2 \mathrm{d}t\bigg] < \infty 
\right\}.
\end{equation*}
We furthermore say that $\cZ\in L^2(\d \P \times \d t)$ is \textit{admissible} if there exists $\theta\in \Theta$ such that $\cZ=\cZ^\theta$.
%Let $\Theta^+\subset\Theta$ be the set of admissible strategies \red{with $\theta\geq0$}. %In other words, $\theta\in \Theta^+$ if we only allow $\theta$ to take values in $\mathbb{R}_0^+$ . 
As can be seen from Proposition \ref{propextend} and the definition of $\Theta$, $Z(t,\omega,y)$ plays a crucial role for our analysis. %Below we will provide two sets of assumptions, (H1) and (H2), where the assumptions and hence the conclusion of Proposition \ref{propextend} hold, respectively.

%\section{A forward-backward SDE system and utility maximization} \label{sec:FBSDE}
%We observe that for any terminal position $H_L$, 
%\begin{equation}
%H_L=-\Pi_0(-H_L)-\int_0^T g(t,\cZ^{\theta}_t)\d t +\int_0^T \cZ^{\theta}_t \d W_t.
%\label{eq:}
%\end{equation}
Recall that $\cZ^\theta_t(\omega) = Z(\omega, t, \theta_t(\omega))$.  %a terminal portfolio value $X_T^\theta=x_0+\cI_t(\theta)$ can be hedged perfectly if and only if the initial endowment/capital is $x_0=-\Pi_0(-X_T^\theta)$. 
By Proposition \ref{propextend} and Remark \ref{pl} the portfolio value at time $t$ is given by $X_t^\theta=x_0+\cI_t(\theta)$ with
%\begin{equation}
%X_t^\theta=-\Pi_t[-X_T^\theta]=-\Pi_0(-X_T^\theta)-\int_0^t g(s,\cZ^{\theta}_s)\d s +\int_0^t \cZ^{\theta}_s \d W_s.
%\label{eq:}
%\end{equation}
%where ${\theta}_t(\omega)=\wt{\cZ}(\omega,t,Z_t(-X_T^\theta)(\omega))$. 
%Hence, 
%\begin{equation}
%X_t^\theta=x_0+\cI_t(\theta)=x_0-\int_0^t g(s,\cZ^{\theta}_s)\d s +\int_0^t \cZ^{\theta}_s \d W_s,
%\label{eq:}
%\end{equation} 
\begin{equation}
\cI_t(\theta)=-\int_0^t g(s,\cZ^{\theta}_s)\d s +\int_0^t \cZ^{\theta}_s \d W_s+\Pi_t(H_\mathrm{M})-\Pi_0(H_\mathrm{M}),
\label{eq:stern}
\end{equation} 
%\thaicomment{Here we do not assume (H1) so the above representation may not be ok}
which can be considered as the gain/loss in the time interval $[0,t]$. Note that by definition of $\cI_t(\theta)$ we have that 
{$(-(X_t-\Pi_t(H_{\mathrm{M}})+\Pi_0(H_{\mathrm{M}})),\cZ^{\theta}_t)$ is a solution of a BSDE with driver $g$} and terminal condition $-(X_T^\theta-H_{\mathrm{M}}+\Pi_0(H_{\mathrm{M}}))$ (since it satisfies (\ref{bsdeX2})). In particular, setting $t=0$ we have, $x_0=X^\theta_0=-\Pi_0(-(X_T^\theta-H_{\mathrm{M}}+\Pi_0(H_{\mathrm{M}})))$.
Furthermore, in case that $H_\mathrm{M}=0$ we have that $\Pi_t(H_\mathrm{M})=\Pi_0(H_\mathrm{M})=0$ and therefore $x_0=-\Pi_0(-X_T^\theta)$. 

%\thaicomment{ the notation $\cI_t^\theta$ shall be made more consistent. Can you more specifically point to inconsistencies?}
%\subsection{Expected utility maximization}
Below, we study the following utility maximization problem for the Large Trader
\begin{equation}
\sup_{\theta \in \Theta}\E[U(X_T^{\theta}+H_L)],
\label{eq:EU.1}
\end{equation}
where $\theta_t$ is the number of risky assets held at time $t$, and $X_0^{\theta}+H_L=x_0+H_L$ is the initial endowment with $H_L\in L^{\infty}(\mathcal{F}_T)$. In (\ref{eq:EU.1}) we use the convention that $\E[U(X_T^{\theta})]=-\infty$ if for the negative part we have $\E[U^{-}(X_T^{\theta})]=\E[U(X_T^{\theta}){\bf 1}_{U(X_T^{\theta})<0}]=-\infty$. We assume that the utility function $U$ is a strictly increasing, strictly concave and three times differentiable function.
% satisfying the Inada's condition $\lim_{x\to -\infty}U'(x)=\infty$ and  $\lim_{x\to +\infty} U'(x)=0$ and has reasonable asymptotic elasticity ({\bf where is it used?}). 
%Assume that 
%\begin{equation}
%\frac{\d\Q}{\d \P}=\cE_T(-\eta)=\exp\bigg\{-\frac{1}{2}\int_0^T \eta^2_t-\int_0^T \eta_t \d W_t\bigg\}:=\xi_T,
%\label{eq:}
%\end{equation}
%where $\eta$ is adapted process satisfying the Novikov's integrability condition $\E[\frac{1}{2}e^{\int_0^T\eta_t^2 \d t}]<\infty$. For simplicity we assume that the MM has zero inventory i.e. $H_{\mathrm{M}}=0$.  We observe that for each strategy $\theta$, the gain/loss process can be expressed as
%\begin{equation}
%\cI_T(\theta)=-\int_0^T \wt{g}(t,\cZ^{\theta}_t)\d t +\int_0^T \cZ^{\theta}_t \d W_t, \quad \wt{g}(t,\cZ^{\theta}_t)={g}(t,\cZ^{\theta}_t)-\eta_t \cZ^{\theta}_t.
%\label{eq:}
%\end{equation} 
%It is clear that $\wt{g}_z(t,z)={g}_z(t,z)-\eta_t$. 
 Note that problem \eqref{eq:EU.1} can be restated as
$
\sup_{\theta \in \Theta}\E[U(x_0+\cI_T(\theta)+H_L)].
$
The rest of the paper is dedicated to analyze problem (\ref{eq:EU.1}). 
\section{An FBSDE approach} \label{sec:FBSDE}
{After introducing the notion of a subgradient for convex functions on possibly non-convex domains, we show in this section that an optimal solution of the  utility maximization problem \eqref{eq:EU.1} necessarily {solves} a coupled FBSDE. Conversely, when suitable integrability conditions are met, the solution to a coupled FBSDE corresponds to an optimal strategy. Subsequently, we give {some examples and conclude this section with} existence results of the associated FBSDEs.}
\subsection{Background: Subgradients for convex functions on convex or non-convex domains}

%On the other hand, under appropriate integrability conditions the solution of such a coupled FBSDE also yields an optimal strategy. The section then discusses numerical approximations and finally gives existence results of the associated FBSDEs. }
%Let us first consider the case where condition (H1) is fulfilled and show that then the conclusions of Proposition \ref{Th:FBSDE1} hold.
{
	\noindent For a convex function $g$, the subgradient of $g$ {at $z\in\bbr^{1\times d}$} is defined as %\mitjacomment{Throughout the paper we use $\mathbb{R}^d=\mathbb{R}^{d\times 1}$.} 
	\begin{equation*}
		\nabla_c g(t,z) := \{y \in \mathbb{R}^{d}\vert g(t,\tilde{z}) - g(t,z) \geq (\tilde{z}-z)y \text{ for all } \tilde{z} \in \mathbb{R}^{1\times d} \}.
	\end{equation*} 
	If $g$ is differentiable in $z$, then the subgradient has only one element, the derivative of $g$, henceforth, denoted by $g_z$. It follows directly from the definition of a subgradient that a convex function $g$ is minimized in a $z$ if and only if $0\in \nabla_c g(t,z).$ % which coincides with the gradient of $g$ with respect to $z$.  
For a set $\mathcal{M} \subset \mathbb{R}^d$ we define $a+b\mathcal{M}:=\{a+bm | m \in \mathcal{M}\}$ for $a,b \in \mathbb{R}$. We denote $\nabla_c f:=-\nabla_c (-f)$.}\\
%Sometimes we will additionally discuss the case that $g$ is positively homogeneous (in $z$) meaning that $g(t, \lambda z) = \lambda g(t,z)$ for all $\lambda > 0$. In this case it follows from {\bf (HL)} that $g$ must be Lipschitz continuous (in $z$).
%Next, we will define the utility dynamic evaluation method $\Pi_t$ of the market maker. 

Let $g$ be given as above. {For $t\in[0,T]$, fix $A\in \mathcal{B}(\bbr^d)$.  
We denote the restriction of the function $g$ to $A$ as}
{
\begin{equation}\label{def:gbar}
g_A(t,z)=\begin{cases}
g(t,z) & \text{if }z\in A, \\
\infty & \text{else}.
\end{cases}
\end{equation}}
 {In the later sections, $A$ will either be equal to $\text{Im}_t(Z)$ or to $cl(\text{Im}_t(Z))$,} where we denote by $ cl(\text{Im}_t(Z))$ the pointwise closure of the set $\text{Im}_t(Z(\omega))$ for each $(t,\omega)$. {To ease notation we will suppress the possible dependence of $A$ on $t$}.\\
 
 {Note that $g_A$ and $g$ coincide on the set $A$}. {In particular, $g_A$ is convex on its domain.\footnote{{Meaning that for $z_1,z_2\in \text{dom}(g_A(t,\cdot))$ we have $g_A(t,\lambda z_1+(1-\lambda)z_2)\leq \lambda g_A(t,z_1)+(1-\lambda)g_A(t,z_2)$ \disc{for all $\lambda\in[0,1]$ such that $\lambda z_1+(1-\lambda)z_2\in \text{dom}(g_A(t,\cdot))$}. However, the domain itself may not be convex.}}}
 %and for $i=1,\ldots,n$ we have $cl(\text{Im}(Z(t,\cdot)))=[a^i_t,b^i_t]$, while for for $i=n+1,\ldots,d$ we have $cl(\text{Im}(Z(t,\cdot)))=[Z^i_t(H_{\mathrm{M}}),Z^i_t(H_{\mathrm{M}})]$. Define ${g}_A=\sum\limits_{i=1}^n {g}_A^i +\sum\limits_{i=n+1}^d {g}_A^i.$
 {The reason for sometimes considering $A=cl(\text{Im}_t(Z))$ (instead of $A=\text{Im}_t(Z)$) is the following:} we can transform the problem \eqref{eq:EU.1} of maximizing over all strategies, into the problem of maximizing over all ``admissible'' $\cZ^\theta$ {with values in $\text{Im}(Z)$}. However, if the set ${\text{Im}(Z)}$ is not closed, the maximum of the latter problem may only be attained in an unadmissible $\cZ$ in the boundary. In this case, the portfolio optimal problem \eqref{eq:EU.1} does not have a solution, see Theorem \ref{Th:inverse} and Remark \ref{remark31} below. Hence, the boundary points of ${\text{Im}(Z)}$ play an important role in the analysis.\\
 
 {
 For $z\in A=\text{dom}(g_A(t,\cdot))$\footnote{{So that $g_A(t,z)<\infty$.}} define for $h\in \R^d$ $$D{g}'_A(t,z,h):=\limsup_{\lambda\downarrow 0} \frac{{g}_A(t,z+\lambda h)-{g}_A(t,z)}{\lambda}.$$}
  {
 Define further the (generalized) subgradient of the (not necessarily convex) function ${g}_A$ as $$\nabla {g}_A(t,z):=\{\zeta\in\R^d|\zeta h\leq D{g}'_A(t,z,h) \text{ for every } h\in\R^d\}.$$
By Theorem 23.2 in \cite{rockafellar1970convex}, $\nabla g_A=\nabla_c g_A$ if $A$ (and therefore $g_A$) is convex. Hence $\nabla$ is a true generalization of $\nabla_c$. The following summarizes some useful properties of $D{g}'_A$ and $\nabla g_A$.}
 {
 \begin{proposition}
 	\label{propnotes}
 \begin{enumerate}
 	\item[(a1)] If $D{g}'_A(t,z,h)<\infty$ for an $h\in\R^d$ and $g(t,\cdot)$ is differentiable in $z$ then $D{g}'_A(t,z,h)=g_z(t,z)h$. In particular, under \textbf{(B1)}, $D{g}'_A(t,z,h)<\infty$ implies that $D{g}'_A(t,z,h)=g_z(t,z)h$.
 	\item[(a2)] $D{g}'_A(t,z,h)< \infty$ if and only if $z+\lambda h\in A$ for all $\lambda\in [0,\varepsilon^{h,z}]$ for an $\varepsilon^{h,z}>0$. In this case $$D{g}'_A(t,z,h)=\lim_{\lambda\downarrow 0} \frac{g(t,z+\lambda h)-g(t,z)}{\lambda}=Dg'(t,z,h)<\infty.$$
 	\item[(a3)] If $B\in\mathcal{F}_t$, $X,Y\in L^0(\mathcal{F}_t)$, and $Dg'(t,X,Y)<\infty$ then $D{g}'_A(t,X,Y\I_{B})=\I_{B}D{g}'_A(t,X,Y)$.
 	\item[(a4)] $D{g}'_A(t,z,\tilde{\lambda}h)=\tilde{\lambda}D{g}'_A(t,z,h)$ for all $\tilde{\lambda}>0$ and $z\in A$.
 	\item[(b)] If $A_1\subset A_2$, then $\nabla g_{A_2}(t,z)\subseteq \nabla g_{A_1}(t,z)$ for all $z\in A_1$. Thus, $\nabla g_c=\nabla g_{\bbr^d}\subseteq \nabla g_A$ for all $A\in \mathcal{B}(\bbr^d)$ and $z\in A$. In particular, $\nabla  {g}_A(t,z)\neq \phi$.
 	\item[(c)] The set $\nabla {g}_A(t,z)$ is convex and weakly closed for $z\in A$.
 	\item[(d)] If ${g}_A(t,\cdot)$ is differentiable in $z$ we have $\nabla {g}_A(t,z)=\{ {g}_{A,z}(t,z)\}=\{g_z(t,z)\}$. If ${g}_A$ is convex, (but not necessarly differentiable) $\nabla  {g}_A$ agrees with the usual definition of a subgradient for convex functions, i.e., $\nabla_c {g}_A=\nabla {g}_A$.
 	\item[(e)] If $z$ is a local minimum, we have $0\leq D{g}'_A(t,z,h)$ for all $h \in \R^d$. This again is equivalent to $0\in\nabla {g}_A(t,z)$.
 	\item[(f)] $\nabla(\beta z+\lambda  {g}_A(t,z))=\beta+\lambda\nabla {g}_A(t,z)$ for $\lambda>0,\beta\in\R^d$. More general, if $f$ is differentiable we have $\nabla (f(z)+\lambda {g}_A(t,z))=f_z(z)+\lambda\nabla {g}_A(t,z)$ for all $z\in A$.
 	\item[(g)] Suppose that $A_1$ and $A_2$, locally agree in $z$ in the sense that $\tn{B_\epsilon (z)\cap A_1=B_\epsilon (z)\cap A_2}$ for an $\epsilon$-neighborhood around $z$. Then $\nabla g_{A_1}(t,z)=\nabla g_{A_2}(t,z)$. In particular, $\nabla g_A(t,z)=\nabla g_{cl(A)}(t,z)$, for all $z\in A$, if $A$ is closed, open, or a multidimensional Cartesian interval.
 \end{enumerate}
 \end{proposition}
}

\subsection{Necessary and sufficient conditions} 

The following theorem gives necessary optimality conditions.

\begin{theorem}\label{Th:FBSDE1new}
Suppose that $\theta^*$ is an optimal strategy of Problem  \eqref{eq:EU.1} with %%with $\partial \cZ_y(t,\theta^*_t)$ non-zero $\d \P\times \d t$ a.e. 
$\E[\vert U(X_T^{\theta^*}+H_L)\vert]<\infty$ and $\E [\vert U'(X_T^{\theta^*}+H_L)\vert^{1+\tilde{\epsilon}}]<\infty$ with $\tilde{\epsilon}>0$. Then there exists a continuous adapted process $\zeta$ with $\zeta_T=H_L$ such that $U'(X^{\theta^*}+\zeta)$ is a martingale process. Setting $M_t^\zeta:=d \langle\zeta, W\rangle_t/ \d t$\footnote{{Note that this derivative may only be defined $\d t$-a.e.}} { and $A=\text{Im}_t(Z)$}, we have
\begin{equation}\label{eq:rem}
0\in U''(X^{\theta^*}_t+\zeta_t)(\cZ^{\theta^*}_t-Z_t(H_\mathrm{M})+M^{\zeta}_t)^\intercal-U'(X_t^{\theta^*}+\zeta_t) \nabla {{g}_A}(t,\cZ^{\theta^{*}}_t) ,\quad  \d \P\times \d t \hspace{0,1cm} \mbox{a.s.}
\end{equation}
\end{theorem}
\disc{
\begin{remark}
	Note that equation \eqref{eq:rem} actually gives the necessary optimality condition, in the sense that for an optimal $\theta^*$, \eqref{eq:rem} holds. If $\cZ^{\theta^{*}}_t\in \text{int}(A),$ \eqref{eq:rem} becomes 
	$$0= U''(X^{\theta^*}_t+\zeta_t)(\cZ^{\theta^*}_t-Z_t(H_\mathrm{M})+M^{\zeta}_t)^\intercal-U'(X_t^{\theta^*}+\zeta_t)g_z(t,\cZ^{\theta^{*}}_t) .$$
	On the other hand, in the general case \eqref{eq:rem} means in terms of the directional derivative that
	$$0\leq  -U''(X^{\theta^*}_t+\zeta_t)(\cZ^{\theta^*}_t-Z_t(H_\mathrm{M})+M^{\zeta}_t)^\intercal\, h+U'(X_t^{\theta^*}+\zeta_t) Dg(t,\cZ^{\theta^{*}}_t,h), $$
	 for all $h\in\bbr^d$. % (which starting at $\cZ^{\theta^{*}}_t$ ``lie'' in $A$, in the sense that $\cZ^{\theta^{*}}_t(\omega)+\lambda(\omega) h(\omega)\in A$ for all $\lambda(\omega)\in [0,\varepsilon^{h,z}(\omega)]$ for an $\varepsilon^{h,z}(\omega)>0$). 
\end{remark}
}
{Fix $t\in[0,T]$ and $A\in \mathcal{B}(\bbr^d)$. To characterize the necessary conditions for an optimum in terms of a stochastic evolution of an FBSDE we need to define a new functional, say $\mathcal{H}$, which is set-valued. Specifically, for a triple of adapted processes $(X,\zeta,M)$ let $\cH_A(t,X_t(\omega),\zeta_t(\omega),M_t(\omega))$ be the $\mathcal{B}(\bbr^d)$-measurable set of all $h\in A $ that solve the equation }
{ \begin{equation}\label{eq:bijective1}
	0 \in-U'(X_t(\omega)+\zeta_t(\omega)) \nabla {g}_A(t,h) +U''(X_t(\omega)+\zeta_t(\omega))\big(h{-Z_t(H_\mathrm{M})(\omega)}+M_t(\omega)\big).
\end{equation}}
%For a triple of adapted processes $(X,\zeta,M)$ (which solves the FBSDE system \eqref{eq:FBSDE} below), we define 
%\textcolor{blue}{
%$$\cH(t,X_t,\zeta_t,M_t):=(\cH^1(t,X_t,\zeta_t,M_t),\ldots,\cH^d(t,X_t,\zeta_t,M_t)):[0,T]\times \Omega \times \bbr \times \bbr \times \bbr^n \mapsto  \mathcal{B}(cl (\text{Im}(Z)))$$ }by
% and $\hat{U}^i_{t,X,\zeta,M^i}(\cH^i) := -U^\prime(X+\zeta)g_z^i(t,\cH^i) + U^{\prime\prime}(X+ \zeta)(\cH^i -Z^i_t(H_\mathrm{M})+ M^i)$, see Proposition \ref{multidimintervals} for the definition of $a^i$ and $b^i$.
	In other words, $\cH_{\tc{A}}$ is the solution of \eqref{eq:bijective1} involving the gradient of $g_A$. We write this as \begin{align}\label{eq:bijective}
		0 \in&-U'(X_t(\omega)+\zeta_t(\omega)) \nabla {g}_A(t,\cH_A(t,X_t(\omega),\zeta_t(\omega),M_t(\omega)))\nonumber \\ &+U''(X_t(\omega)+\zeta_t(\omega))\big(\cH_A(t,X_t(\omega),\zeta_t(\omega),M_t(\omega)){-Z_t(H_\mathrm{M})(\omega)}+M_t(\omega)\big).
\end{align}
 Note that by definition, $\cH_{\tc{A}}$ only takes values in ${A}$.

	{For $A=\text{Im}_t(Z)$ or $A=cl(\text{Im}_t(Z))$, we may by a measurable selection theorem assume that the image of $(t,\omega)\mapsto\cH_A(t,X_t(\omega),\zeta_t (\omega),M_t^\zeta(\omega))$ is a $\mathcal{P} \otimes {\mathcal{B}}	(\mathbb{R}^d)$-measurable set. } Below we show that the optimal strategy can then be characterized by a solution {(in the sense of the stochastic differential inclusion)} of a fully-coupled {multi-valued} forward-backward system. {Following Aumann (1965), we identify in the sequel the integral of a set-valued function with the set of all integrals of its progressively measurable and integrable (measurable) selections and consider stochastic differential inclusions, see \cite{kree1982diffusion,aubin1984differential,kisielewicz2013stochastic,ararat2023set}. Specifically, if $B$ is a progressively-measurable set, we define $\int f(B)  \d s +\int g(B) \d W_s$ as the set of all progressively-measurable processes say $Y$, such that there exists a progressively-measurable processes $b$ taking values in $B$, with $Y=\int f(b_s) \d s +\int g(b_s) \d W_s$.
}

{Consider the coupled FBSDE
\begin{align}
	\begin{cases} 
		X_t&{\in} \,x_0-\int_0^t g(s,\cH_A(s,X_s,\zeta_s,M_s))\d s +\int_0^t \cH_A(s,X_s,\zeta_s,M_s) \d W_s+\Pi_t(H_{\mathrm{ M}})-\Pi_0(H_{\mathrm{ M}}),\\[3mm]
		\zeta_t&{\in} \, H_L-\int_t^T M_s\d W_s+\int_t^Tg(s, Z_s(H_\mathrm{M})) -g(s,\cH_A(s,X_s,\zeta_s,M_s)) \d s\\[3mm]
		& \hspace{1cm} +\int_t^T\frac{1}{2}\frac{U^{(3)}}{U''}(X_s+\zeta_s) \vert \cH_A(s,X_s,\zeta_s,M_s) -Z_s(H_\mathrm{M})+M_s\vert^2 \d s,
		%\\[3mm]
		% & \hspace{1cm}+g(s, Z_s(H_\mathrm{M})) -g(s,\cH(s,X_s,\zeta_s,M_s))\big\}  \d s.
	\end{cases}
	\label{eq:FBSDE}
\end{align}
where $(X,\zeta,M)$ is called the ``solution'' of this FBSDE.
}

\begin{theorem}\label{Th:FB} %{Let $A=\text{Im}_t(Z)$}. 
\tn{ Let $A$ be any progressively measurable set for which \eqref{eq:rem} holds.} Under the assumptions of Theorem \ref{Th:FBSDE1new}, there exists a triple $(X,\zeta,M)$ solving the FBSDE \eqref{eq:FBSDE}. The optimal strategy for Problem \eqref{eq:EU.1} is characterized by 
\begin{equation}
\cZ^{\theta^{*}}_t{\in}\cH_{{A}}(t,X_t,\zeta_t,M_t),
\label{eq:Opt.2}
\end{equation}
%where  
In particular, \eqref{eq:Opt.2} holds for $A=\text{Im}_t(Z)$. %the function $\cH_{{A}}$ {is defined by \eqref{eq:bijective}.} 
%and $(X,\zeta,M)$ is a triple of adapted processes which solves the following FBSDE

%Under (H2) and the assumptions of Theorem \ref{Th:FBSDE1_32}, the optimal strategy for Problem \eqref{eq:EU.1} is characterized by 
%\begin{equation}
%\cZ^{\theta^*}_t=\vert\theta^*(t,X_t,\zeta_t,M_t)\vert Z_t(-\mathrm{sgn}(\theta^*(t,X_t,\zeta_t,M_t))S),
%\label{eq:Opt.2copy}
%\end{equation}
%where $(X,\zeta,M)$ is a triple of adapted processes which solves the following FBSDE
%\begin{align}
%\begin{cases} 
%X_t&=x_0-\int_0^t g(s,\theta^*(s,X_s,\zeta_s,M_s)Z_s(-\mathrm{sgn}(\theta^*(s,X_s,\zeta_s,M_s))S))\d s \\
%&\hspace{1cm}+\int_0^t \theta^*(s,X_s,\zeta_s,M_s)Z_s(-\mathrm{sgn}(\theta^*(s,X_s,\zeta_s,M_s))S) \d W_s,\\[3mm]
%\zeta_t&=H_L-\int_t^T M_s\d W_s -\int_t^Tg(s,|\theta^*(t,X_s,\zeta_s,M_s)|Z_s(-\mathrm{sgn}(\theta^*(s,X_s,\zeta_s,M_s))S))  \d s\\[3mm]
%& \hspace{1cm} +\int_t^T\frac{1}{2}\frac{U^{(3)}}{U''}(X_s+\zeta_s) \Big| |\theta^*(s,X_s,\zeta_s,M_s)|Z_s(-\mathrm{sgn}(\theta^*(s,X_s,\zeta_s,M_s))S)+M_s \Big|^2\d s,
%\end{cases}
% \label{eq:FBSDEcopy}
%\end{align}
%with the function $\theta^*(\ldots)$ satisfying (\ref{lhsp})-(\ref{rhsp}).
\end{theorem}
%Due to the monotonicity and concavity of $U$, $\cH$ is determined uniquely. Moreover, it

%\textcolor{blue}{Note that in the proof it is shown that $U'(X_t+\xi_t)$ is a local martingale corresponding to $\mathbb{E}[U'(X_T+H_L)\mid\mathcal{F}_t]$, i.e., ...}
%$-\frac{U^{(3)}}{U''}$ is also called prudence in the decision theoretic literature. Note that the utility function of the large investor influences the FBSDE and the optimal solution only through the levels of the underlying prudence it induces. \textcolor{blue}{ Specifically, the marginal utility of the investor in \eqref{eq:FBSDE} is penalized by the sum of two terms. The first term, is the difference between the $g$-penalty of the market maker and the $g$-penalty of the wealth process. The second term, corresponds to the prudence of the investor multiplied by the local variance %(i.e. the local variability or local volatility)
% of the difference between the conditional marginal optimal utility of the investor EXPLAIN, $U'(X_t+\zeta_t)$, and the local variance of the market maker. } 
{The ratio $-\frac{U^{(3)}}{U''}$ is also called prudence in the decision theoretic literature. Note that the utility function of the large investor influences the FBSDE and the optimal solution only through the levels of the underlying prudence it induces. Furthermore, since by \eqref{eq:FBSDE} $U'(X_t+\zeta_t)$ is a martingale, we have
\begin{equation}
X_t+\zeta_t=(U')^{-1}(R_t)=(U')^{-1}(\mathbb{E}[U'(X_T+H_L)\mid \mathcal{F}_t]).
\label{eq:prudence}
\end{equation}
Consequently, $X+\zeta$ is the certainty equivalent of the marginal utility. In other words the expected marginal optimal terminal wealth, corresponds to gaining additional \$$ \zeta_t$ from time $t$ on. On the other hand $M$ gives the local variation of this amount due to noise. It is of course well known from (static) duality theory and the BSDE literature (see the references in the introduction) that the marginal utility plays a key role for analyzing optimal solutions of portfolio selection problems. It is furthermore worthwhile to note that by Theorem \ref{Th:FB} the marginal utility certainty equivalent, $X+\zeta$, consists of the sum of three parts a) a general martingale term due to underlying nose, b) a penalty term involving the function $g$ due to the market maker being risk averse, %(this term might be positive if the trade reduces the risk of the market maker)
c) a penalty term which is the squared quadratic variation of the expected marginal terminal utility, $X+\zeta$, multiplied with the prudence of the large investor (induced by the utility function $U$). In particular, locally the penalty due to the market maker's risk aversion, and the one due to the prudence of the investor multiplied with the local variability can be additively separated. A higher prudence results in a higher penalty. %Finally, let us remark that if the penalty due to the risk aversion of the market maker given by the integral of $g$ is non-negative and the investor is prudent in the sense that $-\frac{U^{(3)}}{U''}>0$, then $X+\zeta$ is a supermartingale meaning that time will in expectation reduce the expected future marginal utility.
}

%\thaicomment{We should remove the discussion about the CRRA utility which is not defined on $\bbr$?}
\vspace{2mm}%\mitjacomment{Yes}

\vspace{2mm}

%%{\em
%% \noindent{\bf (CZ)}: 

Below we study the other direction of Theorem \ref{Th:FB}. {For this purpose, we consider the following assumption: \\
\textbf{Assumption (H0)}: One of the following conditions holds: Either
\begin{itemize}
	\item[i)]  ${cl(\text{Im}(Z))}$ is convex \disc{(for our fixed $H_\mathrm{M}$)}; or
	\item[ii)]  the market is complete \disc{(for our fixed $H_\mathrm{M}$)}; or
	\item[iii)] 
	%Define the diffusion process $\cR=(\cR^1,\ldots,\cR^n)$ to be the solution to the following SDE:  
	%$\cR_u^{i,t,r^i} = r^i,\ \ 0\leq u\leq t,$ and
	%we have the following separability condition,
 $d=1$.\footnote{{Or $\sigma$ is diagonal, $g$ is of the form $g(t,z)=\sum_{i=1}^ng^i(t,z^i)+g^{n+1}(t,z^{n+1},\ldots,z^d)$, and $	H_\mathrm{M}=\sum_{i=1}^n h^{\mathrm{M},i} (\cR_T^{i})+H_{\mathrm{M}}^{\perp}$ where $H_{\mathrm{M}}^{\perp}=h^{\perp}(W^{n+1},\ldots,W^d)$ is assumed to be measurable with respect to the filtration generated by $(W^{n+1},\ldots,W^d)$.}}
\end{itemize}
}
{
	\begin{proposition}\label{h2}
		If \textbf{(H0)} holds, \eqref{eq:bijective} has a unique solution for $A=cl(\text{Im}_t(Z))$, and we can identify $\mathcal{H}_{cl(\text{Im}_t(Z))}(t,X,\zeta,M)$ as an $\bbr^d$-valued function which for fixed $(t,X,\zeta,M)$ is given  by the unique solution of \eqref{eq:bijective}. In particular, $\cH_{cl(\text{Im}_t(Z))}$ is not multi-valued (i.e., set-valued) but single-valued. 
\end{proposition}}

 Below we show that an optimal strategy can be obtained from the solution of an FBSDE system. To begin, let $\psi_1(x):=\frac{U'}{U''}(x)\leq 0.
$ $-1/\psi_1$ is also called the risk aversion of the investor.
\begin{theorem}\label{Th:inverse}
{Suppose that \textbf{(B1)-(H0)} hold} and that $\psi_1$ is bounded. Let $(X,\zeta,M)$ be a triple of adapted processes which for {$A=cl(\text{Im}_t(Z))$} solves the FBSDE \eqref{eq:FBSDE} (with ``$=$'' instead of ``$\in$''), and satisfies 
$$
\E [U'(X_T+H_L)^2]<\infty,\quad \E [\vert U(X_T+H_L)\vert]<\infty,\quad \E[\int_0^T |M_t|^2 dt]<\infty.
$$
Furthermore, assume that either (a) $\sup_t U'(X_t+\zeta_t)$ is integrable or (b) $\frac{1}{\psi_1}$ is bounded. %and in case that $d>n$ additionally $\psi_1$ is bounded away from zero. %and that Assumption \ref{Le:Condition3} holds.
%Assume furthermore %that $\cH(t,X_t,\zeta_t,M_t) \in Image(\mathcal{Z}(t,\omega,\cdot))$, and 
Then, the solution of the problem 
\begin{equation}\label{eq: sternstern}
\sup_{\mathcal{Z}\: \text{\small takes values in } cl(\text{Im}(Z)), \cZ \in \mathcal{L}^2(\d \P \times \d s)}\E[U(X_T^{\mathcal{Z}}+H_L)]
\end{equation}
is given by $(\mathcal{Z}_t^*)_t=((\cH_{{cl(\text{Im}_t(Z))}}(t,X_t,\zeta_t,M_t))_t\in {cl(\text{Im}_t(Z))}$. Furthermore, \tn{if $g$ is Lipschitz or strictly convex, the optimal solution of \eqref{eq: sternstern} is unique.} The optimal portfolio choice Problem (\ref{eq:EU.1}) has a solution $\theta^*$ if and only if there exists a version of  $\mathcal{Z}^*$ such that $\mathcal{Z}^*\in\text{Im}(Z)$. In particular, if \eqref{eq:EU.1} has a solution, $\theta^*_t$ is given by (\ref{eq:Htilde}) below. {Finally, if $\text{Im}_t(Z)$ is open, closed, a multi dimensional Cartesian interval, or, more general, agrees with $cl(\text{Im}_t(Z))$ locally for all $z\in\text{Im}_t(Z)$, we have $\cH_{\text{Im}_t(Z)}(t,X_t,\zeta_t,M_t)=\cH_{cl(\text{Im}_t(Z))}(t,X_t,\zeta_t,M_t)\tn{\cap \text{Im}_t(Z)}$. In particular, \tn{in this case, if $\mathcal{Z}^*\in {\text{Im}_t(Z)}$}, $(X,\zeta, M)$ also solve the FBSDE \eqref{eq:FBSDE} with $A=\text{Im}_t(Z)$.} 
\end{theorem}% (\thaicomment{ below, should it be $\nabla g$ in stead of $g_z$ ?} \mitjacomment{Yes!} 
%\thaicomment{Mitja, I do not see how $\partial\cZ_y$ is dropped here?} \textcolor{blue}{OK?}
{
\begin{corollary}\label{coro37}
Under \textbf{(B1)-(H0)}, there can be at most one solution $(X,\zeta,M)$ to the FBSDE \eqref{eq:FBSDE} with $A=cl(\text{Im}_t(Z))$ satisfying \tn{$X_T\in\mathcal{D}_T$}, $\E[U'(X_T+H_L)^2]<\infty$, \tn{$\E [\vert U(X_T+H_L)\vert]<\infty$} and $\E[\int_{0}^{T}|M_t|^2 \d t]<\infty$ (with "$\in$" in \eqref{eq:FBSDE} replaced by "$=$").
\end{corollary}
}

\begin{remark}\label{remark31}
Without the ``cl'', Problem \eqref{eq: sternstern} corresponds to the optimal portfolio problem \eqref{eq:EU.1}. Now under the conditions of Theorem \ref{Th:inverse} above we always have therefore $\mathcal{Z}^*\in {cl(\text{Im}(Z))}$. On the other hand, $\mathcal{Z}^*$ might be in the boundary of $ {cl(\text{Im}(Z))}$, but not in ${\text{Im}(Z)}$ itself. In this case an admissible optimal strategy $\theta^*$ does not exist. %However, we can find a sequence $\mathcal{Z}^m\in\text{Im}(\mathcal{Z})$ converging to $\mathcal{Z}^*$ in $L^2({\bf\d P}\times\d s)$ with corresponding strategies $\theta^m$. By (\ref{eq:proof}) $\theta^m$ is then the sequence of strategies converging to the supremum in (\ref{eq:EU.1}).
\end{remark}
{
\begin{remark}
	If an optimal solution exists, then it follows from the strict comparison result for BSDEs and \eqref{diff} in the appendix, that under assumption \textbf{(B1)} with $d=1$ and $s'\stackrel{(<)}{>}0$, the optimal strategy $\theta^*$ is unique.
\end{remark}
%\begin{remark}
%Note that if $g$ is positively homogeneous, then by Proposition \ref{propmarkov}, $ \text{Im}(Z)$ is closed so that an optimal solution $\theta^*$ exists under the conditions of Theorem \ref{Th:inverse}.
%\end{remark}
}
Using a measurable selection theorem, see e.g. \cite{aumann1967measurable}, there exists a $\mathcal{P}\otimes\cB(\bbr^d)$-measurable function
$
\wt{\cZ}:\Omega\times[0,T]\times \bbr^d\mapsto \bbr^d
$
such that 
%\begin{equation}
%\cZ(\omega,t,Y_t(\omega))=Z(H_{\mathrm{M}}-Y_t(\omega)S)\quad \mbox{for all}\quad (\omega,t,y)\in \Omega_0\times[0,T]\times \bbr 
%\label{eq:Z1}
%\end{equation}
%and
$
\cZ(\omega,t,
\wt{\cZ}(\omega,t,z))=z$ for $z \in {\text{Im}(Z)}.$ %\quad i=1,\ldots,n.%\dot{c}))
%}
{Now suppose that $(X,\zeta,M)$ solves the FBSDE \eqref{eq:FBSDE} for $A=cl(\text{Im}_t(Z))$.} By the definition of \tc{$\cH_A$} we have $\cH_{\tc{A}}(t,X_t,\zeta_t,M_t) \in  cl(\text{Im}(Z))$. If $\cH_{\tc{A}}(t,X_t,\zeta_t,M_t)$ actually belongs to ${\text{Im}(Z)}$, an optimal strategy $\theta^*$ is therefore given by
\begin{equation}
\theta^{*}_t=\wt{\cZ}(t,\cH_{\tc{A}}(t,{X_t},\zeta_t,M_t)).
\label{eq:Htilde}
\end{equation}
%We further remark that %%the assumption of $\mathcal{Z}$ beeing onto is actually satisfied
%in complete markets $\mathcal{Z}$ is onto so that in this case, we always have that $\cH \in Image(\mathcal{Z})$. 
%through (\ref{eq:Htilde})
%See also Example \ref{completemarket}.\\
%\begin{equation}
%\theta^*_t=\wt{\cZ}(t,\cH(t,X_t,\zeta_t,M_t)),
%\label{eq:Opt.2}
%\end{equation}

%\subsection{Discussion}\label{section4.3}
If one wants to compute the optimal solution numerically in the general case, a possible way suggested by our results is to: a) in a first step solve the coupled FBSDEs \eqref{eq:FBSDE}, to obtain the optimal $\cZ^*$ given by equation \eqref{eq:Opt.2}. Then step b): once the optimal $\cZ^*$ is found, invert the parameterized family of decoupled FBSDEs (or in the Markovian case, equivalent semi-linear PDEs) in order to find the corresponding $\theta^*$ such that $\cZ^*=\mathcal{Z}^{\theta^*}.$ For both steps there are by now numerical algorithms available, see \cite{delarue2006forward}, \cite{zhao2014new}, \cite{huijskens2016efficient}, and \cite{germain2022numerical}, or the earlier work \cite{douglas1996numerical}. {Note that} step b) in low dimensions can actually be implemented by computing  $(\Pi^y, \cZ^y)$) for a fixed $y$, for values of $y$ on a grid. Three further remarks are in order: 
\begin{itemize}
	\item[i)] As shown in Section \ref{se:Regu}, in the case of a complete financial market step a) can often be avoided. A complete financial market means that for any suitable integrable $\mathcal{F}_T$-measurable random variable $H$ there exists $(a,\theta)\in \mathbb{R}\times \Theta$ such that
	$H=a+\mathcal{I}_T(\theta),$
	with $\mathcal{I}_T(\theta)$ corresponding to the P\&L of the strategy $\theta$.
	Sufficient conditions when such a completeness condition holds can be found in \cite{Fukasawa2017}. 
	\item[ii)] The optimal $\cZ^*$ in our setting becomes even explicit (see Proposition \ref{propcomplete} below) if the Market Maker {(but not necessarily the Large Trader)} uses an exponential utility function. It is worthwhile noting that in a setting without price impact the optimal strategy also typically becomes only explicit in special cases. % only when the utility function is power. 
	\item[iii)] {As shown by Proposition \ref{Le:Condition3next} below, in the case of a positively homogeneous driver, $\theta^*$ can be expressed directly through the solution of the FBSDE by equations \eqref{lhsp}-\eqref{rhsp}.}
		%can be avoided (since these equations give the form of $\theta^*$ explicitly in terms of $\cZ^*$), reducing overall the numerical complexity to a similar level as in the case without price impact.}
\end{itemize}

\subsection{Examples}

The next proposition gives cases where it is optimal to invest everything in the riskless asset.
\begin{proposition}\label{optimalL} Suppose that $H_L=0$ and $0\in \nabla {g}_A (t,Z_t(H_{\mathrm{M}}))$ {for $A=cl(\text{Im}_t(Z))$}. %\thaicomment{Assume that the Large Trader evaluates his utility by $\Q$, the same preference measure which the Market Maker uses.} 
	Then the optimal terminal wealth is given by $X_T^*=x_0$. This means that it is optimal for the Large Trader to invest nothing {in the risky assets}, i.e., $\theta^*=0$. In addition, the triple $(X_t^*=x_0, \zeta_t^*=0,M_t^*=0)$ is then a solution of the FBSDE system \eqref{eq:FBSDE}. 
\end{proposition}
{The next proposition shows that {if $d=1$} the image space of $\cZ^y$ can be expressed as a non-empty Cartesian {random} interval.}
\begin{proposition}\label{multidimintervals} 
{Assume that \tn{$d=1$.}} Then $\text{Im}(Z)$ is a non-empty (random) interval, $ \mathbb{I}=(\mathbb{I}_t(\omega))_{t,w}$, of the form $\mathbb{I}_t (\omega):= [a_t(\omega),b_t(\omega)]$
%\times\ldots\times [a_t^n(\omega),b_t^n(\omega)]\times \{Z_t^{n+1}(H_{\mathrm{M}})(\omega)\}\times\ldots\times \{Z_t^d(H_{\mathrm{M}})(\omega)\}$
 $\d\mathbf{P}\times\d t$ a.s. where $a\leq b$ are progressively-measurable processes possibly taking the values $\pm \infty$, and the interval may also be open or half-open. %and we write with a slight abuse of notation $\{Z_t^{i}(H_\mathrm{M})(\omega)\}=[Z_t^i(H_\mathrm{M}),Z_t^i(H_\mathrm{M})].$ 
	Therefore, $${\text{Im}(Z)}=\mathbb{I}:=(\mathbb{I}_t(\omega))_{t,\omega},\quad \mbox{ and we write}\quad  
	\text{Im}(Z(t,\omega,\cdot))=\mathbb{I}_t(\omega).
	$$
\end{proposition}
%\red{Alternatively to (H1) we also consider below the following assumption.}
%
%\textbf{Assumption (H2)} (Positive Homogeneity):\label{ex:2}
%\label{ph}
%\normalfont	We assume that $n=1$, $H_\mathrm{M}=0$,
% $g$ is positively homogeneous, meaning that $g(t,\lambda z)=\lambda g(t,z)$, and $Z_t(\pm S) \neq 0:\,\, \d\mathbf{P}\times\d t$ a.s.
%We denote these conditions as  Assumption (H2). We remark that under (H2), it can be checked that
%\begin{equation}\label{eq:H2cond}
%\cZ^y_t=\vert y\vert Z_t\big(-\mathrm{sgn}(y)S\big), 
%\end{equation}
%with $\mathrm{sgn}(y)=1$ if $y>0$, $\mathrm{sgn}(y)=-1$ if $y<0$, see Appendix \ref{se:App.A}. 
%%In particular, the continuity condition of Proposition \ref{propextend} is satisfied and trading may be extended to continuous time.
%	Therefore, the $\mathcal{P}\otimes\mathcal{B}(\R^d)$-measurable 
%	image space of $\theta\mapsto \cZ^\theta$ under (H2) is given by \begin{equation*}
	%	\text{Im}(\mathcal{Z})=\bigg\{|\theta| Z(-\mathrm{sgn}(\theta)S)\bigg|\theta\in \Theta \bigg\},\quad \mbox{and}\quad \text{Im}(\mathcal{Z}_t):=	\text{Im}(\mathcal{Z}_t)=\bigg\{\vert y \vert Z_t(-\mathrm{sgn}(y)S)\bigg|y\in \mathbb{R}\bigg\}.
	%	\end{equation*} 
{
\begin{proposition}\label{Le:Condition3} Suppose that \tn{$d=1$.} 
	%There exist %%$\Omega_{\cH}$ with $\P(\Omega_\cH)=1$ and unique $\mathcal{P}\otimes \cB(\bbr^{3})$-measurable functions, say 
%\textbf{(H0)(iii)} holds,	$\mathcal{H}^i(t,X,\zeta,M^i):= a^i_t\vee(\hat{U}^i)^{-1}_{t,X,\zeta,M^i}(0)\wedge b^i_t \ \text{for} \ i=1,\ldots,n,\
%	\cH^i(t,X,\zeta,M^i):[0,T]\times \Omega \times \bbr \times \bbr \times \bbr \mapsto  \textcolor{blue}{cl (\text{Im}(Z)^i)} $, defined in \eqref{eq:bijective} above for $i=1,\ldots,n$ uniquely solve the equations 
	Define \[ \hat{U}_{t,X,\zeta,M}(\cH) := -U^\prime(X+\zeta)g_z(t,\cH) + U^{\prime\prime}(X+ \zeta)(\cH -Z_t(H_\mathrm{M})+ M). \] 
	Then, %If under \textbf{(H0)(iii)} an optimal solution $\theta^*$ exists and the assumptions of Theorem \ref{Th:FBSDE1new} hold, we have
	\begin{equation}\label{Ztheta}
		\cZ_t^*= \cH(t,X_t,\zeta_t,M_t),
		%:=(\cH^1(t,X_t^{\theta^*},\zeta_t,M^1_t),\ldots,\cH^n(t,X_t^{\theta^*},\zeta_t,M^n_t),Z^{n+1}_t(H_{\mathrm{M}}),\ldots,Z^d_t(H_{\mathrm{M}}) ), %\: \text{for } i=1,\ldots,n,
	\end{equation}% Furthermore, $\cH$ grows at most linearly in $M$.
	with %$M=(M^1,\ldots,M^n,M^{n+1},\ldots,M^d)$ and 
	$\mathcal{H}(t,X,\zeta,M):= a_t\vee(\hat{U})^{-1}_{t,X,\zeta,M}(0)\wedge b_t$, where $(X,\zeta,M)$ is the solution of the FBSDE \eqref{eq:FBSDE}  with $A=cl(\text{Im}_t(Z))$, and $\cZ^*$ solves the optimization problem \eqref{eq: sternstern}.
\end{proposition}
%\proof
%Assumption \ref{Le:Condition3} is satisfied under the assumptions of Theorem \ref{Th:FBSDE1} if addidtionally $d=1$, $\partial \cZ_y \neq 0$, $\psi_1:=U'/U''$ is bounded and $g$ is continuously differentiable in $z$. Indeed, in this case we can divide both sides in (\ref{eq:H}) by $\partial \cZ_y$, and note that
}
\begin{example}
	Under \tn{$d=1$}, the functional $\cH$ defined above can be computed {even more} explicitly for special cases. For example:
	\begin{itemize}
		\item 
		If $g(t,z)=c_t z+d_t$ for some deterministic functions $c,d$ and $U$ is a CARA function, meaning that $U(x) = -e^{-x\gamma}$, with a risk aversion coefficient $\gamma \in (0,\infty)$, then we have $\cH(t,X,\zeta,M)=a_t\vee (-{c}_t/\gamma-M_t)\wedge b_t$, which is independent of $(X,\zeta)$.
		%When $U$ is a CRRA utility, $U(x) = \frac{x^{1-\gamma}}{1-\gamma}$ for $\gamma>0$ with $\gamma \neq 1$, we obtain $\cH(t,X,\zeta,M)=-(X+\zeta)a_t/\gamma-M_t$. 
		\item Similarly, if $g(t,z)=c_t \vert z \vert ^2/2$ and $U$ is a CARA utility function we get $\cH(t,X,\zeta,M)=a_t\vee (-\gamma M/(\gamma+ {c}_t))\wedge b_t$, which is also independent of $(X,\zeta)$.% and $\cH(t,X,\zeta,M)=-M/(\gamma a_t (X+\zeta)+1)$ for a CRRA utility function.
	\end{itemize}
\end{example}
{
Next, we turn our attention to the case \textbf{(B2)} with $g$ being positively homogeneous and $H_\mathrm{M}=0$. For $M \in \mathbb{R}^{1\times d} $, we define \begin{align}\label{lhsp} \frac{-Z_t(-S)M^\intercal+\frac{U^\prime}{U^{\prime\prime}}(X+\zeta) g(t,Z_t(-S))}{\vert Z_t(-S)\vert^2}=:\theta^*(t,X,\zeta,M) \end{align}
if the left-hand side (LHS) of (\ref{lhsp}) is (strictly) positive, or alternatively
\begin{align}\label{rhsp}-\frac{-Z_t(S)M^\intercal+\frac{U^\prime}{U^{\prime\prime}}(X+\zeta)g(t,Z_t(S))}{\vert Z_t(S)\vert^2}=:\theta^*(t,X,\zeta,M)  \end{align}
if the LHS of (\ref{rhsp}) is (strictly) negative. % and else a convex combination of both left-hand sides in (\ref{lhsp})-(\ref{rhsp}) adding up to zero.
Else we choose $\theta^*(t,X,\zeta,M)=0.$ {The following proposition characterizes optimality under \textbf{(B2)}(i.e., under $g$ being positively homogeneous and $H_\mathrm{M}=0$).} %In the sequel we will to ease notation often drop the ``*'' in $\theta$.
\begin{proposition}\label{Le:Condition3next} If $g$ is positively homogeneous and $H_\mathrm{M}=0$, an optimal strategy must satisfy $\theta^*:=\theta^*(t,X,\zeta,M)$ with $\theta^*(t,X,\zeta,M)$ satisfying \eqref{lhsp}-\eqref{rhsp} and $(X,\zeta,M)$ being the solution of the FBSDE \eqref{eq:FBSDE} with $A=cl(\text{Im}_t(Z))=\text{Im}_t(Z)$. Furthermore, 
	 %Furthermore, $\cH$ grows at most linearly in $M$.
	%and the assumptions of Theorem \ref{Th:FB}, if an optimal solution exists, then 
	$$\cZ_t^{\theta^*}=\theta^*(t,X_t,\zeta_t,M_t)Z_t(-S).$$ 
\end{proposition}
}

\vspace{2mm}
%Due to the monotonicity and concavity of $U$, $\cH$ is determined uniquely. Moreover, it
\subsection{Existence results for coupled FBSDEs}\label{se:existence}
In this section we show that the FBSDE \eqref{eq:FBSDE} admits a solution under appropriate assumptions on $g$ and the utility function $U$. 
%\thaicomment{Are we assuming $g$ is differentiable in this section?} \mitjacomment{I think we may not have to and can work with subgradients to include the case that $g$ is positively homogeneous.}
%\subsection{Existence results for coupled FBSDEs}
%\subsection{The general case}
%Let us more generally show under which conditions the FBSDE in (\ref{eq:FBSDE}) and (\ref{eq:FBSDEcopy}) has a solution.
%We first note that as the coupled FBSDEs are quadratic, general existence results to our best knowledge are not available in the literature. The proof of the following theorem strongly relies on the derived connections between (\ref{eq:FBSDE}) \& the optimal control problem.
\begin{theorem}\label{th:generalcase}
Suppose that {\textbf{(H0)} and} one of the following conditions holds:
\begin{itemize}
	\item[(i)] There exists $K,\tilde{\epsilon}>0$ such that $|U'(x)|^{1+\tilde{\epsilon}}\leq K(1+|x|+|U(x)|)$ for all $x$, and $g$ grows at least quadratically, meaning that there exists $K_1,K_2>0$ such that $$g(t,z)\geq -K_1+K_2|z|^2.$$ 
\item[(ii)] {$H_L=0$, and for $\alpha\geq 0$}, $g(t,z)=\alpha |z|^2+l(t,z)$ is convex in $z$ and $l$ is Lipschitz in $z$ (also uniformly in $t$ and $\omega$) and continuously differentiable in $z$ with $l(t,0)=0$. There exists $K,K_1,K_2>0$ such that $|U'(x)|^{1+\tilde{\epsilon}}\leq K(1+|x|^2+|U(x)|)$ and $U(x)\leq K_1-K_2|x|^2$ both for all $x$. Furthermore, there exists an $w_0\in \mathbb{R}$ such that for all $x\leq w_0,$ $U$ is strictly increasing, and $U(x)$ is {constant} for all $x\geq w_0$.  % meaning that we only consider non-negative strategies.
\item[(iii)] $U$ is exponential, i.e. $U(x)=a-be^{-{\gamma x}}$ for $a\in\R$ and $b,\gamma>0$.
\end{itemize}
Then there exists a solution to the FBSDE \eqref{eq:FBSDE} {(with ``$=$'' instead of ``$\in$'') for $A=cl(\text{Im}_t(Z))$}.
\end{theorem}
{
\begin{remark}
	For the uniqueness of the FBSDE \eqref{eq:FBSDE}, see Corollary \ref{coro37}.
\end{remark}}
\begin{remark}
It is shown in the appendix that actually under the conditions of (i) or (ii) in Theorem \ref{th:generalcase}, the maximization problem \eqref{eq: sternstern} above has a solution.
\end{remark}
%\red{DISCUSS RELATED RESULTS IN THE LITERATURE.}

\begin{remark}
	 Although throughout the paper we require $U$ besides being concave to be strictly increasing and three times continuously differentiable on $\mathbb{R}$, for Theorem \ref{th:generalcase}(ii) we actually only need these conditions on $(-\infty,w_0]$ as we will see in the proof that the optimal solution will only take values in this interval. The classical examples for a utility satisfying (ii) are  quadratic utility functions.
\end{remark}

\section{Connection with BSPDEs }\label{se:BSPDE}
In this section we characterise the value function of our expected utility maximization problem \eqref{eq:EU.1} by a BSPDE which results from a direct application of the It\^o-Ventzel formula for regular families of semimartingales. \tn{For similar results without price impact, see e.g. \cite{mania2017}.} % We remark that SPDEs have been also studied in the context of progressively forward utility.
 Setting $H_\mathrm{M}= 0$ we first introduce
\begin{equation*}
\cI_\zs{s,t}(\theta):=-\int_s^t {g}(u,\cZ^{\theta}_u)\d u +\int_s^t \cZ^{\theta}_u \d W_u,\,\quad 0\le s\le t\le T,
\label{eq:}
\end{equation*}
which represents the total gain/loss of the strategy $\theta$ in $[s,t]$. For any $t\in[0,T]$ and $x\in\bbr$ we define 
\begin{equation}
V(t,x):={\esssup}_\zs{\theta \in \Theta, \theta_s, s\in[t,T]} \E\big[U(x+\cI_\zs{t,T}(\theta)+H_L)\mid\cF_t\big].
\label{eq:Vx}
\end{equation}

The following condition is assumed throughout this section.

\vspace{2mm}
{\em
\noindent{\bf (CV)} For any $t\in[0,T]$ and $x\in\bbr$, the supremum in \eqref{eq:Vx} is attained, i.e., there exists an admissible strategy $\theta^*(x)_s, s\in[t,T]$ such that $V(t,x)= \E\big[U(x+\cI_\zs{t,T}(\theta^*(x))+H_L)\mid\cF_t\big]$.  \tn{Furthermore, $V$ is a random field, i.e. a mapping from $\Omega\times [0,T]\times \bbr$ to $\bbr$.} 
}
\vspace{2mm}

%\noindent{\bf Condition C5}: $ g(t,z)$ is convex with respect to $z$ with $g(t,0)=0$ and $\sup_{t\in[0,T]} \vert g_z(t,z)\vert\le K(1+\vert z\vert)$ for some $K$.

%Let ${\Theta}$ be the set of all admissible strategies. %We observe that under the Markovian setting in Proposition  \ref{theorem2} ${\Theta}$ is convex. Indeed, 
%
%need the following condition.
%
%\vspace{2mm}

%{\em
%\noindent{\bf (C$_\cZ$)} The set $\cA$ defined by 
%$
%\cA(t,\omega):=\mbox{Image}(\cZ(t,\omega,\cdot)$ is convex and closed $\d \P\times \d t$ a.s.  %%The mapping $y \rightarrow \mathcal{Z}_t^y(\omega)$ is continuous $\d P \times \d t$ a.s., and the set of admissible strategies ${\Theta}$ is convex, i.e., for any $\theta^1,\theta^2\in\Theta$ and $\lambda \in[0,1]$, $\lambda \theta^1 +(1-\lambda)\theta^2\in\Theta$.
%}

%{\em
%\noindent{\bf (CH)} $H_{\mathrm{M}}=0$. %%The mapping $y \rightarrow \mathcal{Z}_t^y(\omega)$ is continuous $\d P \times \d t$ a.s., and the set of admissible strategies ${\Theta}$ is convex, i.e., for any $\theta^1,\theta^2\in\Theta$ and $\lambda \in[0,1]$, $\lambda \theta^1 +(1-\lambda)\theta^2\in\Theta$.
%}
%
%\thaicomment{Mitja, I find it is hard to understand the assumption $\theta^*(x) \geq 0$ for $x \geq 0$, or $\theta^*(x) \leq 0$, for all $x\geq 0$? From my understanding it sounds like we are assuming only buying or selling ?}

 \vspace{2mm}
%Note that Condition {\bf (CH)} is fulfilled in Case \ref{ph} and Example \ref{completemarket}. 
Sufficient conditions for \noindent{\bf (CV)} to hold are {for instance} given in Theorem \ref{th:generalcase} (guaranteeing existence of an optimal $\cZ^*$ solving Problem \eqref{eq: sternstern}) and in Theorem \ref{Th:inverse}. %Of course, we additionally assume that $g$ satisfies either Condition ${\bf (HL)}$ or ${\bf (Hg)}$ specified in Section \ref{sec:mod}. %Products of vectors or rows with each other will be understood as scalar products.

%Indeed, under some mild conditions on $\psi,\mu,\sigma$ (see \cite{Fukasawa2017}, Section 4), it has been shown that $Z^\theta$ is monotone in $\theta$. 

\vspace{2mm}

We recall that for any $x\in\bbr$, the process $V(t,x),t\in[0,T]$ is a supermartingale admitting an RCLL modification (see e.g. Theorem 9 in \cite{protter2005stochastic}). Its Galtchouk-Kunita-Watanabe (GKW) decomposition is given by $V(t,x)=V(0,x)-A(t,x)+\int_0^t \alpha(s,x) \d W_s,$ where $A(t,x)$ is an increasing process, and $\alpha$ is a progressively measurable and square integrable process. As in \cite{mania2010,mania2017}, we define a regular family of semimartingales as follows:
\begin{definition}[regular family of semimartingales]\label{Def:1}
The process $V(t,x):\Omega\times [0,T] \times \mathbb{R}\to\mathbb{R}$ is a regular family of semimartingales if
\begin{itemize}
	\item [(a)] $V(t,x)$ is twice continuously differentiable with respect to $x$ for any $t\in[0,T]$.
	\item [(b)] For any $x\in\bbr$, $V(t,x),t\in[0,T]$ is a special semimartingale with progressively measurable finite variation part $A(t,x)$ which admits the representation $A(t,x)=\int_0^t b(s,x) \d s$, where $b(s,x)$ is progressively measurable, i.e.,
	\begin{equation}
V(t,x)=V(0,x)-\int_0^t {b}(s,x) \d s+\int_0^t \alpha (s,x) \d W_s.
\label{eq:ValueBSPDE}
\end{equation}
%where $L(t,x)$ is a local martingale orthogonal with the Brownian motion $W$ for all $x$.  
	\item [(c)] For any $x\in\bbr$, the derivative process $V_x(t,x)$ is a special semimartingale with decomposition
$
V_x(t,x)=V_x(0,x)-\int_0^t {b}_x(s,x) \d s+\int_0^t \alpha_x (s,x) \d W_s,
$
where $\alpha_x$ and $b_x$ denote the derivative of $\alpha$ and $b$ with respect to $x$ respectively. %\red{and denotes the derivative of $b$ with respect to $x$}.
\end{itemize}
\end{definition}
We in addition assume that the following condition holds for the coefficients of the regular family of semimartingales $V(t,x),t\in[0,T]$:

\vspace{2mm}
{\em
\noindent{\bf (CR)} The functions $b(t,x),\alpha(t,x)$ and $\alpha_x(t,x)$ in Definition \ref{Def:1} are continuous with respect to $x$
and satisfy, for any constant $c>0$,
$
\E\bigg[\int_0^T  \max_\zs{\vert x\vert \le c}(\vert b(t,x)\vert, {\vert b_x(t,x)\vert},\vert\alpha(t,x)\vert^2 ,\vert\alpha_x(t,x)\vert ^2)\d s\bigg]<\infty.
$
}
\vspace{1mm}
\\
{Condition ${\bf (CR)}$ is standard in the analysis of stochastic flows and typically satisfied in examples, see for instance Section \ref{se:Regu}.}
Below, $\cV^{1,2}$ denotes the class of all regular families {of} semimartingales $V$ defined by Definition \ref{Def:1} whose coefficients $b$, and $\alpha$ satisfy Condition ${\bf (CR)}$. Recall also that a process $V$ belongs to the class $D$ if the family of processes $V_{\tau}{\bf 1}_\zs{ \tau \le T}$ for all stopping times $\tau$ is uniformly integrable.

\vspace{2mm}
The following proposition shows a dynamic programming principle.
\begin{proposition}\label{Le:supermart} %\tn{Suppose that $V\in \cV^{1,2}$} 
Let $\theta$ be admissible and $s\in[0,T]$. Then, the process $\{V(t,x+\cI_\zs{s,t}(\theta)),t\ge s\}$ is a supermartingale for all $x\in\bbr$. Furthermore, 
\begin{equation}
V(s,x)=\esssup_\zs{\theta \in \Theta,\theta_u,u\in[s,T]} \E \bigg[V(t,x+\cI_\zs{s,t}(\theta))\mid \cF_s\bigg]
\label{eq:super}
\end{equation}
and a strategy $\theta^*$ is optimal if and only if $V(t,x+\cI_\zs{s,t}(\theta^*))$ is a martingale process for every $s$.
\end{proposition}

\begin{proposition}\label{Le:VlamdaV} \tn{Under \textbf{(H0)} the value function $V(t,x)$ is concave with respect to $x$. Moreover,} {under \textbf{(HL)-(H0)}} the value function $V(t,x)$ is strictly concave with respect to $x$ meaning that {for all $\lambda\in(0,1)$,  $V(t,\lambda x_1+(1-\lambda)x_2)> \lambda V(t,x_1)+ (1+\lambda) V(t,x_2)$ a.s. if $x_1\neq x_2$}. 
	
\end{proposition}

The next lemma prepares the ground for showing that under appropriate smoothness conditions, the value function solves a BSPDE.

\begin{lemma}\label{Le:abc} \tn{Suppose that $V\in \cV^{1,2}$ and that \textbf{(H0)} holds. Furthermore, assume either \textbf{(HL)} or that $g$ is coercive and strictly convex.} Then, there exists a progressively measurable process denoted by ${\upsilon}(t,x)$ such that the supremum of \begin{equation}
\cL^V(t,x):=\esssup_{\cZ\in{cl(\text{Im}_t(Z))}}\bigg(-g(t,\cZ)V_x(t,x)+\frac{1}{2}\vert \cZ\vert ^2 V_{xx}(t,x)+\cZ\alpha_x(t,x)\bigg),
\label{eq:cL}
\end{equation}
is attained {at $\cZ_t^*(x)=\upsilon(t,x)$}. In particular, ${\upsilon}(t,x)$ satisfies
%$z\mapsto-{g}(t,z)a+\frac{1}{2}\vert z\vert^2 b+zc$ is concave in $z$ and attains the supremium at $\wh{\upsilon}(t,a,b,c)$ characterised by the first order condition %
the first order condition $ 0\in \cU^V(t,\upsilon(t,x),x)$, where $\cU^V(t,z,x):=-\nabla {g}_A^\intercal(t,z)V_x(t,x)+ z V_{xx}(t,x)+\alpha_x(t,x)$ {with $A=cl(\text{Im}_t(Z))$}. {Moreover, $\upsilon$ is unique under \textbf{(H0)}.}

%\thaicomment{Questions: do we have $Z_t(-S)=Z(S) $ a.s. ? If yes, it simplifies a lot :)}\mitjacomment{No.}
%-\nabla g(t,{\upsilon}(t,x))a+ \wh{\upsilon}(t,a,b,c) b+c\ni 0. 
%\end{equation}
 %which is denoted by $\upsilon(t,a,b,c)$ and
%$$\upsilon(t,a,b,c)\in\argmax_{z\in \bbr}(-{g}(t,z)a+\frac{1}{2}z^2 b+zc).$$
\end{lemma}

We can now present the two main theorems of this section.
%We show in the following that the value function can be characterized as  solution of a BSPDE. % where $g$ satisfies {\bf (Hg)} or {\bf (HL)}.%where Condition $C 5$ is replaced by Condition $C 6$ and $\gamma_t$ is a positive deterministic function of $t$.

 \begin{theorem}\label{Th:SBPDEss}
Suppose that there exists a modification of the value function, also denoted by $V$, which is in $\cV^{1,2}$. \tn{Furthermore, assume either \textbf{(HL)} or that $g$ is coercive and strictly convex.} Then $V$ is a solution of the BSPDE
\begin{align}
V(t,x)=V(0,x)+\int_0^t \alpha(s,x) \d W_s-\int_0^t \cL^V(s,x) \d s, \label{eq:VLV}
\end{align}
where $V(T,x)=U(x+H_L)$ and the operator $\cL$ is defined by \eqref{eq:cL}.
%\begin{equation}
%\Lambda(s,x):=(-g(s,\upsilon(s,x))+\upsilon(s,x)g_z(s,\upsilon(s,x)))V_x(s,x)-\frac{1}{2}\upsilon^2(s,x)V_{xx}(s,x),
%\label{eq:}
%\end{equation}
%where $\upsilon(s,x)=z^*(s,V_x(s,x),V_{xx}(s,x),\eta_s V_x(s,x)+\alpha_x(s,x))$ is the root of equation \eqref{eq:z**} given in Lemma \ref{Le:abc}.
%\begin{equation}
%-g_z(z)V_x+ z V_{xx}+\eta V_x+\alpha_x=0
%\end{equation}
Moreover, {under \textbf{(H0)}} a strategy $\theta^* \in \Theta$ with $V(t,X_t^{\theta^*})$ belonging to class $D$ is optimal if and only if 
$\mathcal{Z}^{\theta^*}_t=\upsilon(t,X^{\theta^*}_t)$, i.e.
\begin{equation}
0 \in \cU^V(t,\upsilon(t,X^{\theta^*}_t),X^{\theta^*}_t):=-\nabla{g}_A{^\intercal}(t,\upsilon(t,X^{\theta^*}_t))V_x(t,X^{\theta^*}_t)+ \upsilon(t,X_t^{\theta^*}) V_{xx}(t,X^{\theta^*}_t)+\alpha_x(t,X^{\theta^*}_t),
\label{eq:Z*gen}
\end{equation}
{with $A=cl(\text{Im}_t(Z))$}.
%%when $g$ is continuously differentiable or
%and under (H2) $\mathcal{Z}^{\theta^*}_t=\wh{\theta}(t, X^{\theta^*}_t) Z_t(-S).$ %(see \eqref{eq:vv}) 
%\thaicomment{(in other words,
%\begin{equation}
%-Y_tV_x(t,X^{\theta^*}_t)+ \mathcal{Z}^{\theta^*}_t V_{xx}(t,X^{\theta^*}_t)+\alpha_x(t,X^{\theta^*}_t)=0,
%\label{eq:Z*gennew}
%\end{equation}
 %for some $Y_t\in\nabla g_z(t,\mathcal{Z}^{\theta^*}_t)$.}
%and $\mathcal{Z}^{\theta^*}_t$ is square integrable. 
The optimal wealth process $X_t^{\theta^*}$ is then characterized by the forward SDE
\begin{equation}
X_t^{\theta^*}=x_0-\int_0^tg(s,\mathcal{Z}^{\theta^*}_s) \d s+\int_0^t \mathcal{Z}^{\theta^*}_s \d W_s.
\label{eq:Xoptbgen}
\end{equation}
%such that the process $V(t,X_t^{\theta^*})$ belongs to class $D$.
\end{theorem}

\vspace{2mm}

We have seen that the value function of an optimal strategy can be characterized by a BSPDE \eqref{eq:VLV}-\eqref{eq:Xoptbgen}. Differentiating this BSDPE (assuming all derivatives below exist) we obtain
\begin{equation}
\begin{cases}
V_x(t,x)&=V_x(0,x)+\int_0^t \alpha_x(s,x) \d W_s- \int_0^t \cL^V_x(s,x) \d s,\quad V_x(T,x)=U'(x), \label{eq:VLVx}\\[2mm]
X_t^{\theta^*}&=x_0-\int_0^tg(s,\mathcal{Z}^{\theta^*}_s) \d s+\int_0^t \mathcal{Z}^{\theta^*}_s \d W_s.
\end{cases}
\end{equation}
%Below we will show that the BSPDE \eqref{eq:VLVx} coresponds in our setting to $g=0$. We remark that \cite{mania2017} obtains a similar result for the case without price impact, which in our setting corresponds to $g=0$.

The following theorem gives the connection between the FBSDEs \eqref{eq:FBSDE} and the BSPDE \eqref{eq:VLV}.

\begin{theorem}\label{Th:6.2}
Assume {\textbf{(HL)-(H0)}}, all the conditions of Theorem \ref{Th:SBPDEss}, that $(V_x(t,x), \alpha_x(t,x),\cL^V_x(t,x),X_t^{\theta^*})$ is a solution of the BSPDE \eqref{eq:VLVx} and that $V_x(t,x)$ is a regular family of semimartingales \tn{and $V(t,X_t^{\theta^*})$ is of class D}. Let $\upsilon(t,X^{\theta^*}_t)$ be the unique adapted maximizer process in \eqref{eq:cL} and taking only values in the interior of ${\text{Im}_t(Z)}$.
%, i.e. $\nabla g(t,\mathcal{Z}^{\theta^*}_t)$ such that \eqref{eq:Z*gen} holds, i.e., 
%$$-Y_tV_x(t,X^{\theta^*}_t)+ \mathcal{Z}^{\theta^*}_t V_{xx}(t,X^{\theta^*}_t)+\alpha_x(t,X^{\theta^*}_t)=0.$$
Then the triple $(X_t^{\theta^*},\zeta_t,M_t)$  defined by
\begin{equation*}
\zeta_t=I(V_x(t,X_t^{\theta^*}))-X_t^{\theta^*},\quad M_t=\frac{\upsilon(t,X^{\theta^*}_t)V_{xx}(t,X_t^{\theta^*})+\alpha_x(t,X_t^{\theta^*})}{U''(X_t^{\theta^*}+\zeta_t)}-\upsilon(t,X^{\theta^*}_t),
\label{eq:}
\end{equation*}
is a solution of the FBSDE \eqref{eq:FBSDE} {with $A=cl(\text{Im}_t(Z))$}. 
%Furthermore, the optimality conditions \eqref{eq:opcond} or \eqref{eq:Opt.2copy} hold under (H1) or (H2) respectively.% or \eqref{eq:derivatives}-\eqref{eq:starstar}                     holds. On the other hand,Opt.2copy}  holds under (H21) the triple $(X_t^{\theta^*},\zeta_t,M_t)$ satisfies the FBSDE \eqref{eq:FBSDEcopy} and the optimality condition \eqref{eq:Opt.2copy} holdsor (H2) respectively.
\end{theorem}

%\begin{remark}
%In case that the market is complete $\mbox{Image}(\mathcal{Z}^{\theta}_t)=\mathbb{R}^n$ so that $\upsilon(t,X^{\theta^*}_t)$ is always in its interior.   
%\end{remark}

\section{Regularity of the value function} \label{se:Regu}
In this section we provide some results on the regularity (see Definition \ref{Def:1}) of the dynamic value function assuming that $H_L=0$. These properties together with the results from the last section imply that the value function satisfies a corresponding BSPDE. We first consider the case of CARA utility functions. %In particular, assume that the Large Trader's utility is given by $U(x)=-\frac{1}{\gamma}e^{-\gamma x}$. 
%From \eqref{eq:Vx}, the value function can be written as $V(t,x)=U(x)V_t$, where $V_t$ satisfies a BSDE with terminal condition $V_T=1$. 
\begin{proposition} \label{propcara}
Assume $U(x)=-\frac{1}{\gamma}e^{-\gamma x}$ and that the essential supremum in \eqref{eq:Vx} is attained. Then, the value function in \eqref{eq:Vx} satisfies conditions (a)-(c) in Definition \ref{Def:1} and can be written as $V(t,x)=U(x)V_t$, where $V_t$ satisfies the BSDE.
%existence of a solution is ensured through the proof. If $cl(\text{Im}(Z))$ is bounded the driver is Lipschitz.
$V_t=V_0+\int_0^t \alpha_s \d W_s-\int_0^t \cL^V_s \d s,$ %\label{eq:VLVExponential}
with terminal condition $V_T=1$, where the operator $\cL_t^V$ is defined by  
%\begin{align*}
%\cL^V_t&:=\esssup_{\cZ\in\mbox{Im}(\cZ_t^\theta )}\bigg(\gamma g(t,\cZ)V_t+\frac{1}{2}\gamma^2\vert \cZ\vert ^2 V_t+\mitjacomment{-\gamma}\cZ\alpha_t\bigg).
%%\label{eq:cL2}
%\end{align*}
%\thaicomment{Mitja, since $U(x)<0$ I think the above condition should be
$
\cL^V_t:=\essinf_{\cZ\in cl({\text{Im}_t(Z)})}\bigg(\gamma g(t,\cZ)V_t+\frac{1}{2}\gamma^2\vert \cZ\vert ^2 V_t-\gamma\cZ\alpha_t\bigg), 
%\label{eq:cL2}
$  \tn{ where $\essinf$ should be understood as the largest progressively measurable process being ``$\le$" $\d \P \times \d t$ a.s. }
%}\mitjacomment{I agree}
%\mitjacomment{Should it be $+\gamma\cZ\alpha_t$?}
Moreover, the value function satisfies the BSPDE \eqref{eq:VLV} with $\alpha (t,x)=U(x)\alpha_t $ and $\cL^V_t$ defined above.
%In case that $\mathrm{sgn}(\theta_t)=-1$ taking the derivative in \eqref{eq:3star} with respect to $\theta_t$ gives
%\begin{equation} \label{eq:starone2}
%0 = g(t,Z_t(S)V_x(t,x) + \wh{\theta}(t,x) \vert Z_t(S)\vert^2 V_{xx}(t,x) + Z_t(S) \alpha_x(t,x),
%\end{equation} 
%provided $g(t,\cZ_t(-S))V_x(t,X)-\cZ_t(-S)\alpha_x(t,X)\geq 0$. 
\end{proposition}
{
\begin{remark}
If the essential supremum in \eqref{eq:Vx} is attained, the value function itself solves the BSDE above. If $\text{Im}(Z)$ is bounded, the above BSDE is solvable as well, since then the right-hand side is uniformly Lipschitz in $V$ and $\alpha$, so that existence follows from classical BSDE results.
\end{remark}}
Next let us consider the case of a quadratic $g$ and $H_{\mathbf{M}}=H_L=0$. We show that the optimal investment strategy and its FBSDE characterization can be explicitly determined in settings where the market is complete, i.e., for any $X\in \mathcal{D}_T$ there exists a strategy $\theta\in \Theta$ such that $X=-\Pi_0(-X) +\mathcal{I}_T(\theta).$ See \cite{Fukasawa2017} for sufficient conditions for market completeness. Furthermore, in this case $V$ can be shown to be a regular family of semimartingales. To this end, we suppose that the utility $U$ in addition satisfies the Inada condition $\lim_{x\to -\infty}U'(x)=\infty\quad \mbox{and}\quad \lim_{x\to +\infty} U'(x)=0.$% and has reasonable asymptotic elasticity ({\bf where is it used?})
We assume further that the Market Maker evaluates the market risks in terms of an exponential utility certainty equivalent principle under an equivalent measure $\Q\sim\P$. More precisely, for any $X\in \cD_T$, we assume that the evaluation $\Pi$ of the market maker is given by $\Pi_t(X)=-\frac{1}{\gamma}\log\bigg[\E^\Q[e^{-\gamma X}\vert \cF_t]\bigg],$
 %\begin{equation*}
%e^{-\gamma\Pi_t(X)}=\E^\Q[e^{-\gamma X}\vert \cF_t],
%\label{eq:}
%\end{equation*}
where the constant $\gamma>0$ is the risk aversion and 
\begin{equation*}
\frac{\d\Q}{\d \P}=\exp\bigg\{-\frac{1}{2}\int_0^T \vert\eta_t \vert^2 \d t-\int_0^T \eta_t \d W_t\bigg\}:=\xi_T,
\label{eq:}
\end{equation*}
where $\eta$ is a deterministic and bounded process. % satisfying the Novikov's integrability condition $\E[e^{\frac{1}{2}\int_0^T\eta_t^2 \d t}]<\infty$. 
By Girsanov's Theorem we have that $W_t^{\Q}=W_t+\int_0^t \eta_s\d s$ is a standard Brownian motion under $\Q$. Using It\^o's lemma (see e.g. \cite{Fukasawa2017}) we observe that for any $X\in \cD_T$, 	$X=\Pi_t(X)+\int_t^T g(s,Z_s(X))\d s-\int_t^T Z_s(X) \d W_s,$
	%\begin{equation*}
	%X=\Pi_t(X)+\int_t^T g(s,Z_s(X))\d s-\int_t^T Z_s(X) \d W_s,
	%\label{eq:}
	%\end{equation*}
where $g(t,z)=\frac{1}{2}\gamma \vert z\vert^2-\eta_tz$, see also Example \ref{example0}.2. 
%For each strategy $\theta$, the gain/loss process can be expressed as
%\begin{equation}
%\cI_T(\theta)=-\int_0^T {g}(t,\cZ^{\theta}_t)\d t +\int_0^T \cZ^{\theta}_t \d W_t.
%\label{eq:}
%\end{equation} 
Next, we show that the expected utility maximization \eqref{eq:EU.1} in this special case can be solved by a martingale approach. By our assumptions, %\ref{Theo:hedge.1}
 any terminal value $X_T \in \mathcal{D}_T$ can be hedged perfectly with an initial endowment $x_0=-\Pi_0(-X_T)$ which for our $g$ by Example \ref{example0}.2 in Section 3 is equivalent to $e^{\gamma x_0}= \E^\Q[e^{\gamma X_T}]=\E[e^{\gamma X_T}\xi_T].$

Therefore, the problem of utility maximization \eqref{eq:EU.1} is equivalent to the following static optimization problem 
\begin{equation*}\label{eq:EUcompl}
\max_{X}\E[U(X)], \quad \mbox{s.t.}\quad \E[e^{\gamma X}\xi_T]\le e^{\gamma x_0}.
\end{equation*}
Using that by the concavity of $U$, $\lim_{x\to+\infty}U'(x)e^{-\gamma x}=0\quad \mbox{and}\quad\lim_{x\to -\infty}U'(x)e^{-\gamma x}=+\infty,$
we can conclude that
%\begin{lemma}\label{Le:f}
%Assume that %$-\frac{U''(x)}{U'(x)}>\gamma $ for all $x\in\bbr$ and 
%\begin{equation}
%\lim_{x\to+\infty}U'(x)e^{-\gamma x}=0;\quad \lim_{x\to -\infty}U'(x)e^{-\gamma x}=+\infty.
%\label{eq:}
%\end{equation}
$U'(x)e^{-\gamma x}{\gamma ^{-1}}$ is a decreasing function on $\bbr\to\bbr_+$ whose inverse we denote by $f$.  The next proposition gives an explicit solution of the optimal portfolio in case of a complete market. {The proof of the first part can be done using Lagrangian techniques while the second part can be seen using It\^o calculus. For similar results without FBSDEs, see \cite{S07,S14,anthropelos2018}.
%\end{lemma}
%
%\noindent{\bf Condition 4}: For any $\lambda>0$,
%\begin{equation}
%\E[e^{\gamma f(\lambda \xi_T)}\xi_T]<\infty.
%\label{eq:}
%\end{equation}
\begin{proposition}
	\label{propcomplete}
Assume that for any $\lambda>0$, $\E[e^{\gamma f(\lambda \xi_T)}\xi_T]<\infty.$ The optimal terminal wealth of Problem \eqref{eq:EU.1} is then given by $X^*_T:= f(\lambda \xi_T)$,
%\label{eq:} 
where $\lambda$ is determined such that the budget constraint $\E[e^{\gamma X^*_T}\xi_T]=e^{\gamma x_0}$ is met. The optimal strategy can be characterized as the strategy $\theta^*$ that perfectly replicates the terminal optimal wealth $X^*_T$, i.e.,
$
X^*_t=x_0-\frac{\gamma}{2}\int_0^t \vert\cZ^{\theta^*}_s\vert^2\d s +\int_0^t \cZ^{\theta^*}_s \d W_s^\Q,\quad t\in[0,T],
$
where
\begin{equation}
\cZ^{\theta^*}_t:=\frac{1}{\gamma}\bigg(\frac{\beta^*_t}{R^*_t}+\eta_t\bigg)
\label{eq:theta*}
\end{equation}
with $\beta^*_t$ being the progressively measurable process resulting from the martingale representation
\begin{equation}
R_t^*:=\E[U'(X^*_T)\vert \cF_t]=U'(X^*_T)-\int_t^T\beta^*_s\d W_s.
\label{eq:RT*}
\end{equation}
Furthermore, define $
\zeta^*_t:=I(R_t^*)-\frac{1}{\gamma}\log\bigg(\frac{R_t^*}{\gamma\lambda \xi_t}\bigg),$
and $
M_t^*:=\frac{\beta^*_t}{U''(X_t^*+\zeta^*_t)}-\frac{1}{\gamma}\bigg(\frac{\beta^*_t}{R^*_t}+\eta_t\bigg).%.\cH(t,X_t^*,\zeta_t^*).
$
Then the triple $(X_t^*,\zeta_t^*,M_t^*)$ solves the FBSDE \eqref{eq:FBSDE}. Finally, if additionally $U(x)=-e^{-\gamma_A x}$, for some constant $\gamma_A>0$ then
$\cZ^{\theta^*}_t=\frac{\eta_t}{\gamma+\gamma_A}.$
\end{proposition}
\begin{remark}
\tn{	Since in the above proposition $X^*_T= f(\lambda \xi_T)$ and $\eta$ in the definition of $\xi_T$ is deterministic, $(X_t^*)$ is a Markov process and $V(t,x)$ is deterministic.}
\end{remark}
For convenience, we have included the proof in the online version of the paper. In the sequel, let $\wt{U}(y)$ be the "exponential"-conjugate of $U$, which is defined by
\begin{equation}
\wt{U}(y):=\sup_{x}(U(x)-y e^{\gamma x}),\quad y>0.
\label{eq:Utilde}
\end{equation}

%As mentioned in Section \ref{se:complete},
%\begin{lemma}\label{Le:f}
%Assume that %$-\frac{U''(x)}{U'(x)}>\gamma $ for all $x\in\bbr$ and 
%\begin{equation}
%\lim_{x\to+\infty}U'(x)e^{-\gamma x}=0;\quad \lim_{x\to -\infty}U'(x)e^{-\gamma x}=+\infty.
%\label{eq:}
%\end{equation}
By the definition of $f$, we have $U'(f(y))e^{-\gamma f(y)}{\gamma ^{-1}}=y$, for $y>0$. It can be verified that $f'(y)<0$ and that $\wt{U}(y)=U(f(y))-y e^{\gamma f(y)}$ is a continuously differentiable convex function and we have the following conjugate relation
$
U(x)=\inf_{y>0}(\wt{U}(y)+y e^{\gamma x}),
$
and
$\wt{U}'(y)=-e^{\gamma f(y)}$. Let $$V(x):=\sup_{\cZ\in L^2(\d \P \times \d s)} \E\bigg[U\bigg(x+\int_0^T (\frac{1}{2}\gamma \vert \cZ_t\vert^2-\eta_t \cZ_t)  \d t -\int_0^T \cZ_t \d W_t\bigg)\bigg]$$ and the dual value function is given 
$
\wt{V}(y):=\E[\wt{U}(y\xi_T)]$ for $ y>0$.	
The following result gives {the optimal wealth explicitly through a duality approach} with the linear terms used for instance in \cite{kramkov1999asymptotic,schachermayer2001optimal} being replaced with an exponential function.
\begin{proposition}\label{Le:duality} Assume that for any $\lambda>0$, $\E[e^{\gamma f(\lambda \xi_T)}\xi_T]<\infty.$ The following statements hold:
\begin{enumerate}
	\item [(i)] The value functions $V(x)$ and $\wt{V}(y)$ are exponential-conjugate, i.e
	$$
	\wt{V}(y)=\sup_{x}(V(x)-y e^{\gamma x}),\quad {V}(x)=\inf_{y>0}(\wt{V}(y)+y e^{\gamma x}).
	$$
	They are continuously differentiable concave (resp. convex) function defined and finite valued on $\bbr$ (resp. $\bbr_+$) and satisfy
$$
\lim_{x\to \infty} V'(x)=0,\quad \lim_{x\to -\infty} V'(x)=\infty, \quad \lim_{y\to 0} \wt{V}'(y)=-\infty,\quad \lim_{y\to \infty} \wt{V}'(y)=0.
$$
\item [(ii)] An optimal investment strategy $\theta^*$ exists, is unique and the following duality holds
\begin{equation*}
\frac{1}{\gamma}U'(X_T^{\theta^*}(x))e^{-\gamma X_T^{\theta^*}(x)}=y \xi_T,\quad \mbox{or equivalently,}\quad e^{\gamma X_T^{\theta^*}(x)}=-\wt{U}'(y \xi_T),
\label{eq:duality}
\end{equation*}
where $X_T^{\theta^*}(x)=x{-}\int_0^T (\frac{1}{2}\gamma \vert Z_t^{\theta^*}\vert^2-\eta_t Z_t^{\theta^*})  \d t {+}\int_0^T Z_t^{\theta^*}  \d W_t$ and $y=\frac{1}{\gamma}V'(x)e^{-\gamma x}$.
\item [(iii)] For any $x\in\bbr$ and $ y=\frac{1}{\gamma}V'(x)e^{-\gamma x}$, the process $e^{\gamma X_t^{\theta^*}(x)}=\E^{\Q}[-\wt{U}(y\xi_T)\vert \cF_t]=\E^{\Q}[e^{\gamma X_T^{\theta^*}(x)}\vert \cF_t]$ is a $\Q$-martingale.  
\end{enumerate}
\end{proposition} 
Let
$
R_1(x):=-\frac{U''(x)}{U'(x)}$ and $R_2(x):=-\frac{U^{(3)}	(x)}{U''(x)}.$ {$R_1$ and $R_2$ are sometimes also called risk aversion and prudence respectively.} We now provide sufficient conditions on the utility function $U$ which gives the smoothness needed for the existence of a solution of the BSPDE \eqref{eq:VLV}.

\begin{proposition}\label{Th:quadg}
 Assume that $R_1(x)$ and  $R_2(x)$are bounded and bounded away from zero, $zf(z)$ is bounded, and $\E[U(f(\lambda \xi_T)]<\infty,$ and $\E[e^{\gamma f(\lambda \xi_T)}\xi_T]<\infty,$ for any $\lambda>0$. Assume further that $|F(x,y)|:=|U'(f( \lambda(x)\exp\{-\eta y-\frac{\eta^2 T}{2} \}))|$, 
%$	\lambda'(x)\exp\{-\eta y-\frac{\eta^2 T}{2} \}f'(\lambda(x) \exp\{-\eta y-\frac{\eta^2 T}{2}\})U''(f(\lambda(x) \exp\{-\eta y-\frac{\eta^2 T}{2} \}))$
	%and $f'( y)U''(f(y))$ 
	$|\frac{\partial F}{\partial y}|$, $|\frac{\partial F}{\partial x \partial y}|$ and $|\frac{\partial F}{\partial x^2 \partial y}|$
	are bounded by $K(x)\exp(a|y|^2)$ with $0<a<\frac{1}{2T}$ and $K$ a continuous function. Then there exists a modification of the value function in \eqref{eq:Vx} satisfying conditions (a)-(c) in Definition \ref{Def:1}. % are fulfilled and the value function in \eqref{eq:Vx} satisfies the following BSPDE 
%\end{lemma}
%
%\noindent{\bf Condition 4}: For any $\lambda>0$,
%\begin{equation}
%\E[e^{\gamma f(\lambda \xi_T)}\xi_T]<\infty.
%\label{eq:}
%\end{equation}
%\begin{align}\label{eq:Vform}
%V(t,x)=V(0,x)+\int_0^t \alpha(s,x) \d W_s+ \frac{1}{2}\int_0^t\frac{\vert\eta V_x(s,x)+\alpha_x(s,x)\vert^2}{V_{xx}(s,x)-\gamma_s V_x(s,x)} \d s.
%\end{align}
\end{proposition}

\section{Conclusion}\label{sec: conclusion}
We considered a continuous-time setting with permanent endogeneous price impact induced by a change in the inventory of the market maker. We showed that trading in such a setting corresponds to non-linear stochastic integrals and arises naturally as the limit of discrete time trading. In general incomplete markets with endowments of the investor and the market maker, we then characterized optimal solutions in terms of coupled FBSDEs and BSPDEs, thereby highlighting the interplay between these approaches. Finally, we gave new existence results for the arising FBSDEs and BSPDEs and furnished examples in scenarios where the driver function { is positively homogeneous}, the market is complete, and/or the utility function adopts an exponential form. Future works in this direction are to consider {questions pertaining to arbitrage, a numerical analysis,} different probabilistic settings or other decision theoretic preferences for the large investor or the market maker.

\vspace{2mm}

%\footnotesize
\bibliography{arxivresub_MOR}
\bibliographystyle{plain}

\normalsize
\begin{appendix}

%\proof See e.g. \cite{kunita1986lectures,kunita1997stochastic}. \endproof
\section{BMO's}\label{sec:appBMO}
Consider the set of progressively measurable processes $\phi$ satisfying $\E[\int_0^T\vert \phi_t\vert^2 \d t]<\infty$. The natural extension of the space of bounded processes is the space of BMO processes defined by
$$
BMO(\P)=\bigg\{\phi\in  \,\vert\,\exists\, C\quad\mbox{s.t.}\quad \forall t,\, \E\bigg[\int_t^T\vert \phi_s\vert^2 \d s\bigg|\cF_t\bigg]\le C\bigg\}.
$$ 
The BMO norm of a process $\phi\in BMO(\P)$ is defined as the smallest constant $C$ in the above definition. The stochastic integral process $\int_0^\cdot \phi_s\d W_s $ is called a $BMO($\P$)$ martingale if $\phi\in BMO(\P)$. We remark that for $\phi\in BMO(\P)$, the corresponding Dol\'eans-Dade exponential is a Radon-Nikodym derivative giving rise to a measure, say $\mathbf{Q}$, for which the Girsanov transformation is well-defined and for which the $BMO$ norms are equivalent to the ones under $\mathbf{P}$; see \cite{carmona2008,kazamaki2006continuous}.

\section{Proofs of Section \ref{sec:mod}}\label{se:proofpropmarkov}
{To ease notation, we will sometimes interpret derivatives like $g_z$ as vectors.}\\
\textbf{Proof of Proposition \ref{propmarkov}. }
{Under \textbf{(B2)} (positive homogeneity) note that for $y\geq0$,
	\begin{align*}
		y\Pi_t(-S) = -yS-\int_t^T g(s,yZ_s(-S)) \d s + \int_t^T yZ_s(-S) \d W_s .
	\end{align*}
	Hence, $(\Pi_t(-yS),Z_t(-yS))=(y\Pi_t(-S),yZ_t(-S))$. On the other hand for $y<0$,
	\begin{align*}
		& -y\Pi_t(S) = -yS-\int_t^T g(s,-yZ_s(S)) \d s + \int_t^T -yZ_s(S) \d W_s  
	\end{align*}
	so that also for $y<0$, $(\Pi_t(-yS),Z_t(-yS))=(-y\Pi_t(S),-yZ_t(S))$.}
	{
	This shows the proposition with $Z(t,\omega,y):=|y|Z_t(-\sign(y)S)$.}
	
{
	 Now assume that \textbf{(B1)} holds. By \cite{stroock1982lectures,pardoux1992backward} the random field $\mathcal{R}_u^{t,r}$ has a version which is a.s. jointly continuous in $(t,u,r)$ together with its partial derivatives with respect to $r$ of order one and two. From the definition of $\mathcal{R}$ it follows by taking derivatives with respect to $r$ in (\ref{SDE}) that $\nabla_r\cR^{t,r}, \nabla_{rr}\cR^{t,r}$ and $(\nabla_r\cR^{t,r})^{-1}$ are bounded %(by $e^{|u_r||_\infty T+||\mu_{rr}||_\infty T}$) 
	and bounded away from zero %(by $e^{-||\mu_r||_\infty T-||\mu_{rr}||_\infty T}$). 
	By Theorem 3.A.2 in \cite{L15}, $|Z|$ is bounded by $C:=||h_r-ys_r||_\infty ||\nabla_r R_T||_\infty$ and therefore indeed grows only linearly in $y$.}
	
	 {It remains to show that there exists a continuous version of $Z^y$ under any of the assumptions \textbf{(B1)}i)-iv). Now this is immediate under \textbf{(B1)}iv), and follows under \textbf{(B1)}ii) from Lemma 5.5 in \cite{BK13} (using a logarithm transformation). For the other cases,
	 } define a new smooth driver function coinciding with $g$ on $[0,T]\times[-C,C]^{\textcolor{blue}{d}}$ { and being zero on $[0,T]\times (\bbr^d\setminus[-(C+1),C+1]^d)$}. Since the $Z$-part of the (old) BSDE is bounded by $C$, the solution for the (old) BSDE is also a solution for the (new) BSDE with the {new driver. For the sake of simplicity, denote the new driver function again by $g$.}
	 
	Now it is well known that $Z^{t,r}$ can be characterized by taking formally the derivative (of the modified drivers $g$) with respect to $r$ in
{ the BSDE
	\begin{align*} \Pi_{\tilde{s}}^{t,r}\bigg(H_\mathrm{M}- yS\bigg)&=\big(h^{\mathrm{M}}(\cR_T^{t,r})-ys
		(\cR_T^{t,r})\big)\notag \\
		&\hspace{0,4cm}-\int_{\tilde{s}}^Tg(u,Z_u^{t,r}(H_\mathrm{M}-y S))\,\d u +\int_{\tilde{s}}^T Z_u^{t,r}(H_\mathrm{M}-y S)\,\d W_u,\quad {\tilde{s}}\in [t,T].
	\end{align*}}

	In particular, \tn{by Corollary 2.11 in \cite{pardoux1992backward}} (which we can apply since the first three derivatives of $g$ are bounded), \tn{the gradient is the unique solution of the BSDE}
		
	\begin{align*}
	   \nabla \Pi_{\tilde{s}}^{t,r,y}&=(h^{\mathrm{M}}_r(\cR_T^{t,r})-ys_r(\cR_T^{t,r}) )\nabla_r
	   \cR_T^{t,r}\nonumber \\
		&\hspace{0,4cm}-\int_{\tilde{s}}^T g_z(u,Z_u^{t,r}(H_\mathrm{M}-y S))\nabla Z^{t,r,y}_u\,\d u+\int_{\tilde{s}}^T
		\nabla Z^{t,r,y}_u\,\d W_u. \end{align*}	
		
		By \tn{Lemma 2.5} in \cite{pardoux1992backward} we have 
	\begin{equation}
	~\label{Zdiff}
	Z^{t,r}_u(h^{\mathrm{M}}(\mathcal{R}_T^{t,r})-ys(\mathcal{R}_T^{t,r}))={\nabla \Pi^{t,r,y}_u(\nabla_r\cR_{u}^{t,r})^{-1}\sigma},
	\end{equation}				
see also \cite{delbaen2011backward}. 
	Next, denote by $(F^{t,r,y},V^{t,r,y})$ the solution to the BSDE
{
	\begin{align}\label{diff}
		F^{t,r,y}_{\tilde{s}}& = (h^{\mathrm{M}}_r(\cR_T^{t,r})-ys_r(\cR_T^{t,r}) )\nabla_r
		\cR_T^{t,r}\nonumber \\
		&\hspace{0,4cm}-\int_{\tilde{s}}^T g_z(u,-F_u^{t,r,y}(\nabla_r\cR_u^{t,r})^{-1}\sigma)V^{t,r,y}_u\,\d u+\int_{\tilde{s}}^T
		V^{t,r,y}_u\,\d W_u. \end{align}}% where $\int_{\tilde{s}}^T V^i_u\,\d W_u$ 
	%is defined by $\sum_{1\leq j\leq
	%	d}\int_{\tilde{s}}^T V_u^{i,j}\,\d W_u^i,$ with $V^{i,j}$ denoting the $j$-th entry
	%of the $1\times d$ dimensional process $V^i$ for $i=1,\ldots,n$. 
{%where $F^{t,r,y}_{\tilde{s}}$ takes values in $\bbr^{ d}$, and $V^{t,r,y}_u$ takes values in $\bbr^{d\times d}$. 
%Furthermore, $y\in \bbr^{1\times n}, \text{ and }s_r(\cR_T^{t,r})$ takes values in $\bbr^{d\times n}$.}
	Note that the index $y$ in $F$ and $V$ simply refers to the dependence of these processes in $y$ through the terminal condition. %\textcolor{blue}{ %Denote $$C'=||h^{\mathrm{M}}_{rr}-ys_{rr}||_\infty ||\nabla_r
	%	\cR_T^{t,r}||_\infty ||(\nabla_r
	%	\cR_T^{t,r})^{-1}||_\infty ||\sigma||_{\infty} + ||ys_{rr}||_\infty||\nabla_{rr}\cR_T^{t,r}||_\infty||(\nabla_r
	%	\cR_T^{t,r})^{-1}||_\infty ||\sigma||_{\infty}.$$
%	Define $C':=|A|e^{BT}$,\text{ with }\\ $B:=2dC\sup_u||(\nabla_r R^{t,r}_u)^{-1} \sigma||_\infty\sup_{u,z}|g_{zz}(u,z)|$, and $A:=||h^{\mathrm{M}}_{rr}-ys_{rr}||_\infty ||\nabla_r
%	\cR_T^{t,r}||_\infty + ||h^{\mathrm{M}}_{r}-ys_r||_\infty||\nabla_{rr}\cR_T^{t,r}||_\infty$.	
%	Next, define a smooth function $\rho:\bbr^{d\times d} \to \bbr^{d\times d}$ with norm bounded by $2(C\vee C')$ as $\rho^{ij}(x^{ij})=x^{ij}$ if $|x^{ij}|\leq C \vee C',$ and $\rho^{ij}(x^{ij})=0$ if $|x^{ij}|\geq 1+(C \vee C')$ for $i,j=1,\ldots, d$. 
{Now if \textbf{(B1)}iii) holds, $(F,V)=(\nabla \Pi, \nabla Z)$ are bounded and we may therefore assume without loss of generality that the driver of \eqref{diff} is Lipschitz in $(F,V)$.}

{
Now, the existence of a version of $F$ continuous in $y$ follows by Proposition 2.4 in \cite{elkaroui97} (the proof relying on Kolmogorov's criteria) and our assumptions which imply that the first two derivatives of $h_\mathrm{M}$ and $s$ are uniformly bounded.}
Finally, if we assume \textbf{(B1)}i), then it may be seen that in case that $d>1$, the BSDEs are dimensionwise separable, so that we may assume without loss of generality that $d=1$. We remark further that $(F^{t,r,y},V^{t,r,y})$ is again the solution of a BSDE with driver function satisfying the quadratic growth conditions (H1), (H2) and (H3) in \cite{kobylanski2000}. This can be seen as follows: The mapping $f_0(u,\omega,F,v):=v g_z(u,F(\nabla \mathcal{R}_u^{t,r})^{-1}(\omega)\sigma)$ satisfies for all $(F,v)$ a.s. $|f_0|\leq K_1(1+|v|)$, a.s. $$\bigg |\frac{\partial f_0}{\partial v}\bigg |=|g_z(u,F(\nabla \mathcal{R}_u^{t,r})^{-1}\sigma)| \leq K_2 \text{ a.s. and}$$ 
	$$\bigg |\frac{\partial f_0}{\partial F}\bigg |=|v||(\nabla \mathcal{R}_T^{t,r})^{-1}\sigma g_{zz}(u,F(\nabla \mathcal{R}_u^{t,r})^{-1}\sigma) |\leq K_3|v|\leq l_\epsilon+\epsilon|v|^2 \text{ a.s.}$$ for any $\epsilon>0$, and $l_\epsilon$ appropriately chosen, with $K_1, K_2, K_3>0$. Hence, the assumptions (H1)-(H3) in \cite{kobylanski2000} hold. Thus, by Theorem 2.3 and Theorem 2.6 in \cite{kobylanski2000} a unique solution to \eqref{diff} exists which is bounded by $M:=||h^{\mathrm{M}}_r-ys_r||_\infty ||\nabla_r
	\cR_T^{t,r}||_\infty$, and therefore also satisfies a uniform growth condition in $y$. By uniqueness $(F,V)=(\nabla \Pi,\nabla Z)$. 
	
\disc{ To show now continuity we will use results on stochastic flow properties of solutions of Markovian quadratic BSDEs. While $\nabla_r \mathcal{R}$ may not be Markov, $( \mathcal{R}_u^{t,r},(\nabla \mathcal{R}_u^{t,r})^{-1})$ is again Markov, and as we will see immediately determines $F$ in \eqref{diff} through a function, say $w$. To derive this rigorously, we will in the sequel with a slight abuse of notation write , $F^{t,r,\tilde{r},y}$ and $(\nabla \mathcal{R}^{t,r,\tilde{r}})^{-1}$ to condition on $( \mathcal{R}_t^{t,r},(\nabla \mathcal{R}_t^{t,r})^{-1})=(r,\tilde{r})$.  Note that since
	$$d (\nabla \mathcal{R}_u^{t,r,\tilde{r}})^{-1}=-(\nabla \mathcal{R}_u^{t,r,\tilde{r}})^{-1}\mu_R(u,\mathcal{R}^{t,r}_u)\d u, $$
	clearly, $(\nabla \mathcal{R}_u^{t,r,\tilde{r}})^{-1}$ is bounded, and therefore $( \mathcal{R}_u^{t,r},(\nabla \mathcal{R}_u^{t,r})^{-1})$ is a Markov-process satisfying a multi-dimensional SDE with Lipschitz-continuous coefficients (as $\mu_R$ and $\mu_{RR}$ are both assumed to be bounded). Furthermore, the terminal condition of \eqref{diff} is a continuous bounded function of $( \mathcal{R}_T^{t,r},(\nabla \mathcal{R}_T^{t,r})^{-1})$.
	Let $\rho:\mathbb{R}\to [0,1]$ be a smooth truncation function with bounded derivative such that $\rho(F)=1$ if $|F|\leq M$, and $\rho(F)=0$ if $|F|>M+1$. Set $$\tilde{f}_0(u,F,\tilde{r},v):=v g_z(u,F \tilde{r}\sigma)\rho(F),$$ 
	and note that since we have already established that $(F_t)$ is bounded by $M$, $\tilde{f}_0(u,F,(\nabla \mathcal{R}_T^{t,r})^{-1},v)$ can be seen as the driver of the BSDE \eqref{diff}.
	Clearly $$|\tilde{f}_0|\leq K_1(1+|v|),\quad a.s.,$$ $$\Big |\frac{\partial \tilde{f}_0}{\partial v}\Big |=|g_z(u,F \tilde{r}\sigma)\rho(F)| \leq K_2 \text{ a.s.,}$$ 
	$$\bigg |\frac{\partial \tilde{f}_0}{\partial F}\bigg |=\bigg|v\tilde{r} \sigma g_{zz}(u,F \tilde{r}\sigma)\rho(F) +v g_z(u,F \tilde{r}\sigma)\rho'(F)\bigg |\leq K_4|v|\leq \tilde{l}_\epsilon+\epsilon|v|^2 \text{ a.s.}$$ 
	$$\bigg |\frac{\partial \tilde{f}_0}{\partial \tilde{r}}\bigg |=\bigg|v F  \sigma g_{zz}(u,F \tilde{r}\sigma)\rho(F) \bigg |\leq K_5|v|\leq \tilde{l}_\epsilon+\epsilon|v|^2 \text{ a.s.}$$ for any $\epsilon>0$, and $\tilde{l}_\epsilon$ appropriately chosen, with $ K_4, K_5>0$. Hence, the assumptions (H4)-(H5) in \cite{kobylanski2000} hold, and by Theorem 3.8 in \cite{kobylanski2000}, we have that $F^{t,r,\tilde{r},y}_u=w^y(t,\mathcal{R}_u^{t,r},(\nabla \mathcal{R}_u^{t,r,\tilde{r}})^{-1})$ with $w^y:[0,T]\times \mathbb{R}\times \mathbb{R}\to \mathbb{R}$ being the continuous viscosity solution of a suitable parabolic PDE (for each fixed $y$). Finally, since for $(t',r',\tilde{r}',y')\to (t,r,\tilde{r},y)$ the corresponding terminal conditions converge in $L^\infty,$ and $\tilde{f}_0(u,F,(\nabla \mathcal{R}_u^{t',r',\tilde{r}'})^{-1},v)$ converges locally uniformly to $\tilde{f}_0(u,F,(\nabla \mathcal{R}_u^{t,r,\tilde{r}})^{-1},v)$, by Theorem 2.8 in \cite{kobylanski2000}, the flow $(t,r,\tilde{r},y)\to (F^{t,r,\tilde{r},y}_s)_{s\geq t}$ is a.s. continuous uniformly in $s$. In particular,  $(t,r,\tilde{r},y)\to F^{t,r,\tilde{r},y}_t=w^y(t,r,\tilde{r})$ is a jointly continuous function, and hence also $y\to F_t^{0,r_0,y}=w^y(t,\mathcal{R}^{0,r}_t,(\nabla \mathcal{R}^{0,r_0,\tilde{r}_0}_t)^{-1})$ is continuous (except possibly on the zero set where $\mathcal{R}^{0,r}_t,(\nabla \mathcal{R}^{0,r_0,\tilde{r}_0})^{-1})$ are not defined).}
	% are also the solution of \eqref{diff} in a Markovian framework with  Theorem 3.7 in \cite{kobylanski2000} a unique solution to \eqref{diff} exists which is bounded by $M:=||h^{\mathrm{M}}_r-ys_r||_\infty ||\nabla_r
	%\cR_T^{t,r}||_\infty$, gives us now that the mapping $y \mapsto F^{r,y}_t(\omega)$ is continuous for every $t$ and $\omega$.}

	{To conclude we can in each case \textbf{(B1)}i) and \textbf{(B1)}iii)} $\omega$-wise set $$Z(t,y):=-F_t^{0,r_0,y}(\nabla_r
	\cR_t^{0,r_0})^{-1}\sigma=Z^{0,r_0}_t(h^{\mathrm{M}}(\mathcal{R}_T^{0,r_0})-ys(\mathcal{R}_T^{0,r_0}))=Z^{0,r_0}_t\bigg(H_\mathrm{M}-\sum_{i=1}^n y^i S^i\bigg).$$  This finishes the proof of Proposition \ref{propmarkov}. % (see \eqref{Zdiff}).} 
	The last two equalities in both equations hold $L^2(\d\P\times \d t)$ a.s. {It also follows directly from the remarks before about $F$, that $Z(t,y)$ is continuous in $y$, and that a linear growth condition of $Z$ in $y$ holds.} %Finally, for the proof of Proposition \ref{multidimintervals}, note that {for each $i$ the image of the mapping $y^i\rightarrow Z^i(t,y^i)$ is an interval}{ (since the image of a continuous function mapping to $\mathbb{R}$ is an interval)},  which we can thus denote as $\text{Im}(Z)^i$ for $i=1,\ldots,n$ .
%	 Clearly also for $i=n+1,\ldots,d$, the image space of $Z^i(t,\cdot)$ must be the point $Z^{i,0,r}\big(H_\mathrm{M}^\perp)=Z^{i}\big(H_\mathrm{M})$ (as $H_\mathrm{M}^\perp$ is independent of $y$). 
	 \endproof

\noindent \textbf{Proof of Proposition \ref{boundedHM}.}
{Under \textbf{(B2)}, Remark \ref{HM} yields that $Z(H_\mathrm{M})=0$.} That $Z(H_\mathrm{M})$ is bounded {also under \textbf{(B1)}} follows with an analogous argument as made in the proof of Proposition \ref{propmarkov}.
\endproof
\section{Proofs of Section \ref{sec:FBSDE}}
{
\noindent \textbf{Proof of Proposition \ref{propnotes}.}
 \begin{enumerate}
 	\item[(a1)-(a4)] Clear, using the definition of $D{g}'_A$, and in (a2) that $g$ is convex and finite.
 	\item[(b)] The first part follows from $Dg'_{A_2}(t,z,h)\leq Dg'_{A_1}(t,z,h)$ for $z\in A_1$, if $A_1\subseteq A_2$. For the second part, note that by Theorem 23.2 in \cite{rockafellar1970convex}, $\nabla_c g(t,z)=\{\zeta|\zeta h\leq Dg'(t,z,h)\text{ for all } h\in\R^d\}$. Since $Dg(t,z,h)\leq D{g}'_A(t,z,h)$ for all $z\in A$, we get that $\nabla_c g(t,z) \subset \nabla {g}_A(t,z)$.
 	\item[(c)] Clear.
 	\item[(d)] If ${g}_A$ is differentiable, then it follows that $z\in \text{int}(\text{dom}({g}_A))$. Hence, $D{g}'_A(t,z,h)$ for all $h$ coincides with the usual directional derivative for convex functions (with $\lim$ instead of $\limsup$), and $D{g}'_A(t,z,h)=Dg'(t,z,h)=g_z(t,z)h$ for all $h\in \R^d$, see \cite{rockafellar1970convex}. Thus, $g_z(t,z)$ is the only element of the subgradient $\nabla g_A$. %and by c) must coincide with $g_z(t,z)$. 
 	The last part follows by Theorem 23.2 in \cite{rockafellar1970convex}.
% 	\item[(e)] Suppose $D{g}'_A\geq 0$. If $D{g}'_A=\infty$ it follows that there exists $\varepsilon>0$ such that ${g}_A(t, z+\lambda h)=\infty$, for all $\lambda\leq\varepsilon$ and therefore $z$ is a local minimum on $z+\lambda h$ with $\lambda\in[0,\varepsilon]$. If there exists a sequence $\lambda\downarrow 0$ such that ${g}_A(t,z+\lambda_n h)<\infty$ it follows that $0\leq D{g}'_A(t,z,h)=Dg'(t,z,h)=\inf_{\lambda>0}\frac{g(t,z+\lambda h)-g(t,z)}{\lambda}$, where the first equation holds by the definition of $D$. Hence, ${g}_A(t,z)=g(t,z)\leq\inf_{\lambda>0}g(t,z+\lambda h)\leq {g}_A(t,z+\lambda h)$ for all $h\in \bbr^d$, and $z$ is therefore a local minimum. On the other hand, clearly 
\item[(e)] %If $D{g}_A(t,z,h)=\infty$ the statement is clear. 
Clear.
%Fix $h\in\bbr^d$. If  $D{g}_A(t,z,h)\in cl(\text{Im}(Z))$ for $0\leq\lambda<\delta^h$ and a $\delta^h>0$. %Now if $z$ is a local minimum clearly $z$ is also a local minimum on the line $z+\lambda h$, with $\delta^h>\lambda>0$, and therefore $D{g}'_A(t,z,h)\geq 0$.
% 	\item[e)] Suppose $D{g}'_A\geq 0$. If $D{g}_A=\infty$ it follows that ${g}_A(t, z+\lambda y)=\infty$, for all $\lambda\leq\varepsilon$ and therefore $z$ is a local minimum on $\lambda\in[0,\varepsilon]$. If ${g}_A(t,z+\lambda_n y)<\infty$ for all $\lambda_n\downarrow 0$ it follows that $0\leq D{g}'_A(t,z,y)=g'(t,zy)=\inf_{\lambda>0}\frac{g(t,z+\lambda y)-g(t,z)}{\lambda}=\inf_{\lambda>0}\frac{{g}_A(t,z+\lambda y)-{g}_A(t,z)}{\lambda}$ such that $g(z)\leq\inf_{\lambda>0}g(t,z+\lambda y)$ and $z$ a local minimum. On the other hand, clearly if $-z$ is a local minimum, we have $D{g}'_A(t,z,y)\geq 0$ for every $y$.
 	\item[(f)] It is $D(z\beta+\lambda {g}_A)'(t,z,v)=v\beta+\lambda D{g}'_A(t,z,v)$. Therefore, $\zeta\in \nabla(z\beta+\lambda {g}_A(t,z))$ is equivalent to $\zeta v\leq \beta v + \lambda D{g}_A(t,z,v)$ for all $v$. This corresponds to $\frac{\zeta -\beta}{\lambda}v\leq D{g}'_A(t,z,v)$, which is equivalent to $\frac{\zeta -\beta}{\lambda}\in \nabla {g}_A(t,z)$. Hence, $\zeta\in \beta +\lambda \nabla {g}_A(t,z)$. The second part follows similarly.
 	\item[(g)] If for an $\epsilon$-environment around $z$, say $B_\epsilon(z)$, we have $\tn{B_\epsilon (z)\cap A_1=B_\epsilon (z)\cap A_2}$, the directional derivatives $Dg'_{A_1}(t,z,h)$ and $Dg'_{A_2}(t,z,h)$, by definition, agree for all $h$. %\tn{Since $z+\lambda h\in A_1$ for all $0\le\lambda<\epsilon$, which implies that  $z+\lambda h\in A_1 \in B_\epsilon(z)\cap A_1$. The other direction follows by symmetry of $A_1$ and $A_2$}. 
	Thus, $\nabla g_{A_1}(t,z)=\nabla g_{A_2}(t,z)$.\endproof
 \end{enumerate}}

\tn{In the proof of the next result, we extend the results in \cite{horst2014,santacroce2014,santacroce2023} to our setting with market price impact.}

\noindent\textbf{Proof of Theorem \ref{Th:FBSDE1new}.}
\normalfont %As in the case without price impact (see e.g. \cite{santacroce2014}), it can be observed that Assumptions 1)-3) can be dropped for a complete market as in Example \ref{completemarket}. %%Indeed, consider the perturbed strategy $\cZ^{\varepsilon}_t:=\cZ^{\theta^*}_t+\varepsilon h_t$ % for $h$ be any measurable process satisfying the integrability condition \eqref{eq:adm}.
Clearly,
\begin{equation}\label{eq:proof}
\E[U(X_T^{\theta^*}+H_L)]=\sup_{\theta \in \Theta}\E[U(X_T^{\theta}+H_L)]=\sup_{\mathcal{Z}\in {\text{Im}(Z)}}\E[U(X_T^{\mathcal{Z}}+H_L)]=\E[U(X_T^{\mathcal{Z}^{\theta^*}}+H_L)],
\end{equation}
where with a slight abuse of notation we write $$X_T^{\mathcal{Z}} := x_0-\int_0^T g(s,\mathcal{Z}_s)\d s + \int_0^T\mathcal{Z}_s\d W_s +H_\mathrm{M}-\Pi_0(H_\mathrm{M}).$$ %Note that if $\mathcal{Z}^m\rightarrow\mathcal{Z}$
%We have the following proposition.
Denote $\mathcal{Z}^*:=\mathcal{Z}^{\theta^*}\in {\text{Im}(Z)}$ by assumption %=\mathbb{I}_1 \times \ldots \times \mathbb{I}_n \times \{Z^{n+1}(H_\mathrm{M})\}\times \ldots \times \{Z^d(H_\mathrm{M}) \}$ 
 attaining the supremum in (\ref{eq:proof}). 
%According to \eqref{eq:stern} the wealth process is given by 
%$$
%X^{\cZ^{*}}_t=x_0-\int_0^t g(s,\cZ^{*}_s)\d s +\int_0^t \cZ^{*}_s \d W_s+\Pi_t(H_\mathrm{M})-\Pi_0(H_\mathrm{M}).
%$$
Define $R_t:=\E[U'(X_T^{\cZ^{*}}+H_L)\vert \cF_t]$ and $\zeta_t:=I(R_t)-X^{\cZ^{*}}_t$ with $I=(U')^{-1}$. Then $\zeta$ is progressively-measurable and $R$ is a martingale. By the (local) martingale representation theorem there exists a locally square integrable progressively-measurable process $\beta$ taking values in $\mathbb{R}^{1\times d}$ such that
\begin{equation}
\label{beta}
R_t=U'(X_T^{\cZ^{*}}+H_L)-\int_t^T \beta_s \d W_s. %~~ \{\theta^*=0\} = \{ \mathcal{Z}^{*}=0\}.
%% \label{eq:}
\end{equation}
Hence, $\d R_t=\beta_t \d W_t$ with terminal condition $R_T=U'(X_T^{\cZ^{*}}+H_L)$. By the assumption in the theorem and Doob's Maximal inequality $\sup\limits_s |R_s|$ is in $L^{1+\tilde{\epsilon}}$, and by the Burkholder-Davis Gundy (BDG) inequality $\E[(\int_0^T|\beta_s^2|ds)^{\frac{1+\tilde{\epsilon}}{2}}]<\infty$. Note that $I(R_t)=X^{\cZ^{*}}_t+\zeta_t$ and $\zeta_T=H_L$. Applying It\^{o}'s formula to $I(R_t)$ we have the following backward representation
\begin{align*}
X^{\cZ^{*}}_t+\zeta_t&=X^{\cZ^{*}}_T+H_L-\int_t^T  \frac{\d R_s}{U''(X^{\cZ^{*}}_s+\zeta_s)}+\frac{1}{2}\int_t^T\frac{U^{(3)}}{(U'')^3}(X^{\cZ^{*}}_s+\zeta_s) \d \langle R,R\rangle_s\\
&=X^{\cZ^{*}}_T+H_L-\int_t^T \frac{\beta_s\d W_s}{U''(X^{\cZ^{*}}_s+\zeta_s)}+\frac{1}{2}\int_t^T\frac{U^{(3)}}{(U'')^3}(X^{\cZ^{*}}_s+\zeta_s) \vert\beta_s\vert^2 \d s,
%% \label{eq:}
\end{align*}
where $U''$ and $U^{(3)}$ are the second and the third order derivatives of $U$, respectively. Hence, $\zeta_t$ is a solution of the following BSDE
\begin{align*}
\zeta_t=H_L&-\int_t^T \bigg(\frac{\beta_s}{U''(X^{\cZ^{*}}_s+\zeta_s)}-(\cZ^{*}_s-Z_s(H^M))\bigg)\d W_s\notag\\
&+\frac{1}{2}\int_t^T\bigg(\vert\beta_s\vert^2\frac{U^{(3)}}{(U'')^3}(X^{\cZ^{*}}_s+\zeta_s)  -2g(s,\cZ^{*}_s)+2g(s,Z_s(H^M))\bigg)\d s.
\label{eq:zeta1}
\end{align*}
By construction the marginal utility process $U'(X^{\cZ^{*}}_t+\zeta_t)=R_t$ is a martingale and by the definition of $M_t^\zeta$ above  
\begin{equation}
\beta_t=U''(X^{\cZ^{*}}_t+\zeta_t)(\cZ^{*}_t+M^{\zeta}_t-Z_t(H_\mathrm{M})).
\label{eq:beta}
\end{equation}
Plugging the last equation into the dynamics of $\zeta$ we get
\begin{align}
\zeta_t&=H_L-\int_t^T M^\zeta_s\d W_s+\frac{1}{2}\int_t^T\bigg(\vert\cZ^{*}_s-Z_s(H_\mathrm{M})+M^\zeta_s\vert^2\frac{U^{(3)}}{U''}(X^{\cZ^{*}}_s+\zeta_s)\notag\\
&\hspace{0,4cm}-2(g(t,\cZ^{*}_t)-g(t,Z_t(H_\mathrm{M})))\bigg)\d s.
\label{eq:zeta2}
\end{align}
{
Let $A=\text{Im}_t(Z)$. Fix a bounded $1\times d$-dimensional progressively-measurable process $\tilde{h}$ with $Dg'_A(t,\cZ^*_t,\tilde{h}_t)<\infty$ and $|\tilde{h}_t|\leq\min (\frac{\tilde{K}}{|\mathcal{Z}^*_t|},\tilde{K})$ for a $\tilde{K}>0$. }
 {Then $Dg'_A(t,\cZ^*_t,\tilde{h}_t)<\infty$ implies by Proposition \ref{propnotes}(a2) and a measurable selection theorem (see \cite{aumann1967measurable}) that there exits a progressively-measurable $\delta_{t}(\omega)>0$ such that $\cZ^*_t+\tilde{\varepsilon}\tilde{h}_t\in A=\text{Im}_t(Z)$ for all $\tilde{\varepsilon}\in [0,\delta_{t}(\omega)]$. Without loss of generality assume that $\delta_t<1$, and scale $\tilde{h}$ by setting $h_t:=\delta_t\tilde{h}_t$ so that $\cZ^\varepsilon:=\cZ^*+\varepsilon h \in \text{Im}(Z)$ for all $\varepsilon\in[0,1]$.}
 Define then \begin{align}
 X_T^{\mathcal{Z}^{\varepsilon}} :&= x_0-\int_0^T g(s,\mathcal{Z}_s^{*}+ \varepsilon h_s)\d s + \int_0^T(\mathcal{Z}_s^{*}+ \varepsilon h_s)\d W_s\nonumber +H_\mathrm{M}-\Pi_0(H_\mathrm{M}).
 \end{align}
 Set $\phi(\varepsilon):=U(X_T^{\mathcal{Z}^{\varepsilon}}+H_L)$. Since {$\epsilon\mapsto X_T^{\cZ^\epsilon}$} is a.s. concave in {$\epsilon\in[0,1]$} (see also \eqref{uniqness2} below), and $U$ is increasing and concave, $\frac{\phi(\varepsilon)-\phi(0)}{\varepsilon}$ is decreasing in $\varepsilon$.
 %$ $\phi(\varepsilon^l)$ is concave on its domain in the sense that $\frac{\phi(\varepsilon^l)-\phi(0)}{\varepsilon^l}$ is decreasing for $\varepsilon^l$ such that $\phi(\varepsilon^l)<\infty$. %Moreover,
%$\phi'(\varepsilon)=U'(X_T^{\mathcal{Z}^\varepsilon}+H_L)\bigg(-\sum_{j=1}^n\int_0^T g^j_z(t,\cZ^{*,j}_t +\varepsilon h^j_t)h^j_t \d t+\sum_{j=1}^n\int_0^T h^j_t \d W^j_t\bigg).
%$
%and
%$$
%\phi''(\varepsilon)=U''(X_T^{\theta^\varepsilon})\bigg(-\int_0^T g_z(t,\cZ^{\varepsilon}_t)h_t \d t+\int_0^T h_t \d W_t\bigg)^2
%-U'(X_T^{\theta^\varepsilon}) \int_0^T g_{zz}(t,\cZ^{\varepsilon}_t)h^2_t \d t<0
%$$
%where we identified $g_z^i$ with the left hand-sided (right hand-sided) derivative at the upper (lower) boundary points of the corresponding interval $\mathbb{I}^i_t$ of $\text{Im}(\mathcal{Z})$ in the $i$-th dimension. 
This implies that for $0<\varepsilon<1$ the function $\Phi(\varepsilon)$ defined by
$$
\Phi(\varepsilon):=\frac{\phi(\varepsilon)-\phi(0)}{\varepsilon}=\frac{U(X_T^{\cZ^{\varepsilon}}+H_L)-U(X_T^{\mathcal{Z}^*}+H_L)}{\varepsilon},
$$
 is decreasing in $\varepsilon$ and 
$
\lim_{\varepsilon\downarrow 0}\Phi(\varepsilon)=U'(X_T^{\mathcal{Z}^*}+H_L)\wh{\chi}^{\theta^*}_T(h)\quad \text{ a.s.},
$
where $\wh{\chi}^{\theta^*}_s(h):=-\int_0^s {Dg'(t,\cZ^{*}_t,h_t)} \d t+\int_0^s h_t \d W_t.$ {Let us show that all moments of $\wh{\chi}^{\theta^*}_T(h)$ exist.} {In case of \textbf{(B1)}, $Dg'(t,\cZ^{*}_t,h_t)=g_z(t,\cZ^{*}_t,h_t)h_t$. 
In case of \textbf{(B2)} by Lipschitz continuity of $g$ with Lipschitz constant, say $K>0$: $|Dg'(t,\cZ^*_t,h_t)|\leq K|h_t|$.
%	In case of \textbf{(B2)} by Lemma \ref{C2} for $h_t=\tilde{\lambda}_tZ_t(-\sign(\tilde{\theta}^*)S)$ for a progressively-measurable $\tilde{\lambda}_t$ we have $Dg'(t,\cZ^{*}_t,h_t)=\sign(\theta^*_t)\tilde{\lambda}_t g(t,Z_t(-\sign(\theta^*)S))$. Else $Dg'(t,\cZ^{*}_t,h_t)=\infty$ which is excluded by the choice of $h$.
Since $h_t$ by construction is bounded by $\min (\frac{\tilde{K}}{|\mathcal{Z}^*_t|},\tilde{K})$, by {\bf (Hg)} or {\bf (HL)}
$Dg'(t,\cZ^{*}_t,h_t)$ in each case is uniformly bounded.}
Thus, by the BDG inequality for semimartingles indeed all moments of $\wh{\chi}^{\theta^*}_T(h)$ exist. Since by assumption $\E[\vert U'(X_T^{\theta^*} +H_L)\vert^{1+\tilde{\epsilon}}]<\infty$, we obtain 
$ \E[\vert U'(X_T^{\mathcal{Z}^*}+H_L)\wh{\chi}^{\theta^*}_T(h)\vert]<\infty.$
%which implies that $\Phi(\varepsilon)$ increasingly tends to the limit $U'(X_T^{\theta^*})\wh{\chi}^{\theta^*}_T(h)$ as $\varepsilon\searrow0$. 
Therefore, $\E[\Phi(\varepsilon)]$ is well-defined as $\E[\Phi(\varepsilon)^+]\le \E[( U'(X_T^{\mathcal{Z}^*}+H_L)\wh{\chi}^{\theta^*}_T(h))^+]<\infty$ for all $1>\varepsilon >0$. Hence, by the monotone convergence theorem we conclude that
{
\begin{equation} \label{eq:mono}
0\geq\lim_{\varepsilon\downarrow 0}\Phi(\varepsilon)=\lim_{\varepsilon\downarrow 0}\E\bigg[\frac{U(X_T^{\mathcal{Z}^{\varepsilon}}+H_L)-U(X_T^{\mathcal{Z}^*}+H_L)}{\varepsilon}\bigg]= \E[U'(X_T^{\mathcal{Z}^*}+H_L)\wh{\chi}^{\theta^*}_T(h)],
\end{equation}
where the first inequality hold as $\cZ^\epsilon\in\text{Im}(Z)$ is admissible\footnote{Since $\cZ^\epsilon_t\in\text{Im}(Z)$ by the measurable selection theorem in \cite{aumann1967measurable}, there exists a $\mathcal{P}$-measurable mapping $\theta^\epsilon$ such that $\cZ_t^\epsilon(\omega)=\cZ(t,\theta^\epsilon_t(\omega))$.} and $\cZ^*$ is optimal.}

Next, define an increasing sequence of stopping times $\tau_k:=\inf\limits_{s\geq 0}\{\int_0^s \vert\beta_t \wh{\chi}_t^{\theta^*}(h)+ U'(X_t^{\cZ^{*}}+\zeta_t) h_t\vert^2 \d t\geq k\}\wedge T $ with $\P [\tau_k=T]\to 1$ as $k\to\infty$.
Applying It\^{o}'s lemma to the product $U'(X_t^{\cZ^{*}}+\zeta_t)\wh{\chi}_t^{\theta^*}(h)=R_t\wh{\chi}_t^{\theta^*}(h)$ we have 
\begin{align}
U'(X_{\tau_k}^{\cZ^{*}}+\zeta_{\tau_k})\wh{\chi}^{\theta^*}_{\tau_k}(h)
=&\int_0^{\tau_k}\bigg(h_t\beta_t-U'(X_t^{\cZ^{*}}+\zeta_t){Dg'(t,\cZ^{*}_t,h_t)}\bigg)\d t\notag\\
&{}+ \int_0^{\tau_k} \bigg(\beta_t \wh{\chi}_t^{\theta^*}(h)+ U'(X_t^{\cZ^{*}}+\zeta_t) h_t \bigg) \d W_t.
\label{eq:last}
\end{align}

%Using the BDG inequality for martingales and for submartingales, the Cauchy-Schwarz inequality and the square integrability of $U^\prime(X_t+\zeta_t), \chi$, $\partial \cZ_y$, and $R$ we can conclude that
%\begin{equation*}
% \E\bigg [\sup_{u\in[0,T]}\bigg\vert\int_0^u  \bigg(\beta_t \chi_t(h)+ U'(X_t^{\theta^*}+\zeta_t)h_t \partial \cZ_y(t,\theta^*_t)\bigg) \d W_t\bigg\vert\bigg]<\infty.
%\label{eq:InteU}
%\end{equation*}
%Hence, the stochastic integral in \eqref{eq:last} is a square integrable martingale. 
Taking expectations and the limit on both sides and noting that $\zeta_T=H_L$ leads to
\begin{align}
&\lim\limits_{k\to\infty}\E\bigg[\int_0^{\tau_k} \bigg(h_t\beta_t-U'(X_t^{\cZ^{*}}+\zeta_t){Dg'(t,\cZ^{*}_t,h_t)}\bigg)\d t\bigg] \nonumber\\
&=\lim\limits_{k\to\infty} \E\bigg[ U'(X_{\tau_k}^{\cZ^{*}}+\zeta_{\tau_k})\wh{\chi}^{\theta^*}_{\tau_k}(h)\bigg]
=\E\bigg[ U'(X_T^{\cZ^{*}}+H_L)\wh{\chi}_T^{\theta^*}(h)\bigg] \le 0 ,
\label{eq:equation1}
\end{align}
{by \eqref{eq:mono}}. %, for any bounded progressively measurable $h$ taking values ($\omega$-wise) in $\{0, \pm \delta_t(\omega)e_i \}$.
%Assume now that $g$ is continuously differentiable in $z$. Using the definition of $g'$ and replacing $h$ by $-h$ 
In the last equality we have used uniform integrability, since $\sup_s |R_s|\in L^{1+\tilde{\epsilon}}$ and %\tc{$\E [\vert U'(X_T^{\theta^*}+H_L)\vert^{1+\tilde{\epsilon}}]<\infty$ } 
the BDG inequality for semimartingales (see for instance Theorem 2, Chapter V in \cite{protter2005stochastic}), {implies that} all moments of $\sup_t |\wh{\chi}_t(h)|$ exist.
%Defining $C:=\{(t,\omega)|\mathcal{Z}^{*}_t(\omega)\in (a^i_t(\omega),b_t(\omega))\}$ and replacing $h$ by $h\mathbf{1}_{C}$ and $-h\mathbf{1}_{C}$, we get a reverse inequality in \eqref{eq:equation1} and can conclude that
%\begin{equation*}
%\lim\limits_{k\to\infton ...y} -\E\bigg[\int_0^{\tau_k} h_t^{i} \bigg(\beta_t^{\red{i}}-U'(X_t^{\cZ^*}+\zeta_t)g_z^{\red{i}}(t,\cZ_t^{*,i})\bigg)\d t\bigg] \leq 0.
%\end{equation*}
%% Since $g^\prime(t,Z_t^{\theta^*},h_t) = \sup_{y \in \nabla g(t,Z_t^{\theta^*})}yh_t$ (see  \cite{rockefeller1970convex})
%It follows from the last two inequalities that 
%\begin{align}
% 0  & = \lim\limits_{k\to\infty} \E\bigg[\int_0^{\tau_k} I_{C} h_t (\beta_t-U'(X_t^{\cZ^{*}}+\zeta_t)g_z(t,\cZ_t^{*}))\d t\bigg].
%\label{eq:Equality}
%\end{align}
%%for any progressively measurable $Y_t \in \nabla g(t,\cZ_t^{\theta^*})$.
%\thaicomment{ could you explain \eqref{eq:Equality}? It seems to me that the differentiability is essential in the above argument to obtain \eqref{eq:aI}.} \mitjacomment{You are in principle right. I added an additional argument below for the case that $g$ is positively homogeneous.}
%%where we used in the last equation the ($t,\omega$)-vice definition of $g^\prime$. 
%Repeating arguments as in \eqref{eq:last}-\eqref{eq:Equality} taking the derivative with respect to a perturbation of $\mathcal{Z}^*$ instead of $\theta^*$ (setting in the corresponding equations for fixed $i$ $\mathcal{Z}^{\theta^*}=\mathcal{Z}^*$,$\mathcal{Z}^{\theta^\varepsilon}
%=\mathcal{Z}^*+\varepsilon (0,\ldots,0,h^i,0,\ldots,0)$ and $\partial\mathcal{Z}_y^{*,i}= e_i$), 
Now let $B_t:=h_t\beta_t-U'(X_t^{\cZ^*}+\zeta_t){Dg'(t,\cZ^{*}_t,h_t)}$
 and {choose $h_t^{new}=h_t{\bf 1 }_{B_t>0}$. Note that by Proposition \ref{propnotes}(a3) $\big(h_t^{new}\beta_t-U'(X_t^{\cZ^*}+\zeta_t)Dg'(t,\cZ^{*}_t,h_t^{new})\big) I_{[0,\tau_k]}=\big(h_t\beta_t-U'(X_t^{\cZ^*}+\zeta_t)Dg'(t,\cZ^{*}_t,h_t)\big) {\bf 1 }_{B_t>0}I_{[0,\tau_k]}$ is increasing in $k$ so that by the monotone convergence theorem, the limit on the outer left-hand side in \eqref{eq:equation1} (with $h$ replaced by $h^{new}$) can be taken inside. Hence, arguing by contradiction, from \eqref{eq:equation1} we get that $\{(t,\omega): B_t(\omega)>0\}$ must be a zero set, and therefore $B_t\le 0$, $\d \mathbf{P}\times \d t$ almost everywhere. }
{
By the positive homogeneity of the directional derivative (see Proposition \ref{propnotes}(a4)), and the fact that by Proposition \eqref{propnotes}(a2), $Dg'_A(t,z,h)=Dg'(t,z,h)$ if $Dg'_A(t,z,h)<\infty$, it follows that $h_t\beta_t-U'(X_t^{\cZ^*}+\zeta_t)Dg'_A(t,\cZ^{*}_t,h_t)\leq 0$ $\d \mathbf{P}\times \d t$ almost everywhere for all $h$ (and not only for $h$ which were scaled)\footnote{{For $h$ with $Dg'_A(t,\cZ^*_t,h)=\infty$, this inequality is clear.}}. By Proposition \ref{propnotes}(e)-(f), this entails that $0\in \big(\beta_t-U'(X_t^{\cZ^*}+\zeta_t)\nabla g_A(t,\cZ^{*}_t)\big)$. Using \eqref{eq:beta} leads to the desired conclusion. \endproof} 

\noindent\textbf{Proof of Theorem \ref{Th:FB}.}
%we know that the optimal strategy is given by \eqref{eq:Htilde},
% $$\theta^*_t=\wt{\cZ}(t,\cH(t,X_t^{\theta^*},\zeta_t,M_t^\zeta)),$$ 
 $(\zeta,M^\zeta)$ is a solution of the BSDE \eqref{eq:zeta2} with driver 
$$M\to\vert\cZ^{\theta^*}_s-Z_s(H_\mathrm{M})+M\vert^2\frac{U^{(3)}}{U''}(X^{\theta^*}_s+\zeta_s)-2(g(s,\cZ^{\theta^*}_s)-g(s,Z_s(H_\mathrm{M}))).$$
%\textcolor{red}{?? You need here uniquiness??} \\
%Furthermore, by Theorem \ref{Th:FBSDE1new}, $\cZ^{\theta^*}$ satisfies (\ref{eq:rem}).% or (\ref{eq:derivativesnew})-(\ref{eq:starstar1new}).
%% For $\theta^* \neq 0$ 
{By Theorem \ref{Th:FBSDE1new},  %and Proposition \ref{Le:Condition3next} 
	\begin{equation}\label{Ztheta2}
				\cZ_t^{\theta^*} \in \cH_{\tn{A}}(t,X_t^{\theta^*},\zeta_t,M_t^\zeta). %\: \text{for } i=1,\ldots,n,
	\end{equation} holds.} % and (\ref{eq:Opt.2copy}).
%%By Lemma \ref{Le:Condition3}, we obtain $\cZ^{\theta^*}_t=\cH(t,X_t,\zeta_t,M_t)$.
%where $(X,\zeta,M)$ is a triple of adapted processes which solves the following FBSDE \eqref{eq:FBSDE}. Note that $\theta^*_t=\wt{\cZ}(t,\cH(t,X_t^{\theta^*},\zeta_t,M_t^\zeta)),$ which implies that $\cZ^{\theta^*}_t=\cH(t,X_t^{\theta^*},\zeta_t,M_t^\zeta)$. 
The proof is completed by plugging this identity back into the BSDE and into the dynamics of $X_t^{\theta^*}$. \qed \\
\noindent\textbf{Proof of Proposition \ref{h2}}\\
{
	If the market is complete, $\text{Im}(Z)=\bbr^d$, and therefore $\text{Im}(Z)$ is convex. On the other hand, if \textbf{(H0)(iii)} by Proposition \ref{multidimintervals} below,  $\text{Im}(Z)$ is convex as well. In particular, in each of these cases  $cl(\text{Im}(Z))$ is convex which is therefore the only case we need to consider.\\ }
	
	{
		Fix $(X,\zeta,M) \in \bbr \times \bbr \times \bbr^d$. For $h\in \bbr^d$ let \begin{equation}\label{eq: F}
			F(h):=-U'(X+\zeta){g}_A(t,h) +U''(X+\zeta)\big(\frac{1}{2}|h|^2{-Z_t(H_\mathrm{M})h}+Mh\big).
		\end{equation} {Note that for $h\notin A=cl(\text{Im}(Z))$ we have $F(h)=\infty$ (by the definition of ${g}_A$). Clearly, $F$ is strictly concave since $A$ is convex. Note further that $F$ is lower-semicontinuous (as $A$ is closed), and coercive on $\bbr^d$ (i.e., $F(h)\to -\infty$ if $|h|\to\infty$). In particular, $F$ has a unique maximum on its domain, $A=cl(\text{Im}(Z))$, say $h^*$. Since this \tn{maximum}, $h^*$, is unique, it is the only solution to
		\begin{equation}\label{nablaF}
			0\in \nabla_c F(h)=\nabla F(h)=-U'(X+\zeta)\nabla {g}_A(t,h) +U''(X+\zeta)\big(h{-Z_t(H_\mathrm{M})}+M\big),
		\end{equation}}
		%it follows that $0\in \nabla F$.
		%By Proposition \ref{propnotes}(e) and \ref{propnotes}(f) $h^*$ solves \eqref{eq: F} and Proposition \ref{h1} follows.
		where we used Proposition \ref{propnotes}(e)-(f). Hence, \eqref{eq:bijective} has exactly one solution. %The last statement of the proposition follows from Theorem \ref{Th:FBSDE1new}.
		\endproof\\}
		
%	{
%	Now if $cl(\text{Im}(Z))$ is convex, the function $F$ defined in \eqref{eq: F} is strictly concave, and therefore has a unique maximum. In particular, there is a unique $h^*\in \bbr^d$ such that $0\in \nabla_c F(h^*)= \nabla F(h^*)$ (see Proposition \ref{propnotes}(d)), and therefore $\mathcal{H}_A(t,X_t,\zeta_t,M_t)=\{h^*\}$. 
%	}\\

Next let us show a technical lemma for the proof of Theorem \ref{Th:inverse}. %\label{se:lineargrove}
{As in the proof of Proposition \ref{h2} above we may assume  under \textbf{(H0)} that $cl(\text{Im}(Z))$ is convex. In particular, ${g}_A(t,z)$ is convex with $A=cl(\text{Im}_t(Z))$, and therefore by Proposition \ref{propnotes}(d) $\nabla_c g_A$ (the usual subgradient of a convex function) and $\nabla g_A$ agree. Furthermore, by Proposition \ref{h2}, $\cH$ is single-valued, i.e., a proper function.}
\begin{lemma}\label{lineargrove}
	Suppose that \textbf{(H0)} holds and $\psi_1=U'/U''$ is bounded. Then the function $\cH_{{A}}(t,X,\zeta,M)$ defined in \eqref{eq:bijective} {with $A=cl(\text{Im}_t(Z))$} grows at most linearly in $M$, i.e., there exists $C>0$ such that	$\vert \cH_{{A}}(t,X,\zeta,M)\vert \leq C (1+ \vert M \vert).$
\end{lemma}
{
	\proof%If $g$ is positively homogeneous, the claim follows immediately from Proposition \ref{Le:Condition3next} and (\ref{lhsp})-(\ref{rhsp}). Let us consider case (H1). 
	%It is sufficient to check the first $n$ components of $\cH(t,X,\zeta,M)$ because  the last $n-d$ components of $\cH(t,X,\zeta,M)$ are equal to $(Z^{n+1}(H_\mathrm{M}),\ldots,Z^d(H_\mathrm{M}))$, which by Proposition  \ref{multidimintervals} are bounded.  Denote therefore  $\hat{\cH}:=(\cH^1(t,X,\zeta,M^1),\ldots,\cH^n(t,X,\zeta,M^n))$, $\hat{Z}(H_\mathrm{M}):=(Z^1(H_\mathrm{M}),\ldots,Z^n(H_\mathrm{M}))$ and $\hat{M}:=(M^1,\ldots,M^n)$. 
	%% that $d=1.$ 
	Let us denote $\cH:=\cH_A$. Note that $Z(H_\mathrm{M})=\cZ^0\in \text{Im}(Z)$. Next we get from \eqref{eq:bijective} that
	\begin{equation}\label{H}
	\cH = -M+Z(H_\mathrm{M})+ \psi_1 Y 
	\end{equation}
	for a $Y=(Y^1,\ldots,Y^d)\in\nabla {g}_A(t,\cH)=\nabla_c {g}_A(t,\cH)$. Choose $\bar{y}_t=(\bar{y}_t^1,\ldots, \bar{y}_t^d)\in \nabla_c {g}_A(t,Z(H_\mathrm{M}))$. Note that  $\bar{y}_t$  is uniformly bounded (because of Assumption \textbf{(Hg)} or \textbf{(HL)} and Proposition \ref{boundedHM}).  Using that $\psi_1\le 0$ and that by convexity $(\nabla_c {g}_A(t,\cH) - \nabla_c {g}_A(t,Z(H_\mathrm{M})))(\cH-Z(H_\mathrm{M}))^\intercal\geq 0$, multiplying both sides in \eqref{H} by $\cH$ yields that
	\begin{align*}
		\vert\cH\vert^2 &= - \cH (M-Z(H_\mathrm{M}))^\intercal  + \psi_1 \cH Y^\intercal   \\
		&= - \cH (M-Z(H_\mathrm{M}))^\intercal+\psi_1 Z(H_\mathrm{M})Y^\intercal  + \psi_1 (\cH-Z(H_\mathrm{M})) (Y-\bar{y})^\intercal + \psi_1 (\cH-Z(H_\mathrm{M})) \bar{y}^\intercal   \\
		& \leq \frac{\vert \cH \vert ^2}{2a^2} + \frac{a^2\vert M-Z(H_\mathrm{M})\vert^2}{2}+\psi_1 Z(H_\mathrm{M})Y^\intercal  + \psi_1 (\cH-Z(H_\mathrm{M})) \bar{y}^\intercal \\
		& \leq \frac{\vert \cH \vert^2}{2a^2} + a^2\big(\vert M \vert^2 +\vert Z(H_\mathrm{M})\vert^2\big) + \frac{\|\psi_1\|_\infty a^2}{2}\vert Z(H_\mathrm{M})\vert^2 +\|\psi_1\|_\infty\frac{|Y|^2}{2a^2} + \|\psi_1\|_\infty \vert\vert \bar{y} \vert\vert_\infty \bigg (\vert \cH \vert+|Z(H_\mathrm{M})|\bigg )\\
		&\leq \frac{\vert \cH \vert^2}{2a^2} +  a^2\big(\vert M \vert^2 +\vert Z(H_\mathrm{M})\vert^2\big) + \frac{\|\psi_1\|_\infty a^2}{2}\vert Z(H_\mathrm{M})\vert^2  \\
		&\hspace{0.5cm}+\|\psi_1\|_\infty\frac{K^2(1+|\cH|)^2}{2a^2}+\frac{\|\psi_1\|^2_\infty \vert |\bar{y} \vert|^2_\infty}{2}+\frac{1}{2}|\cH|^2 + \|\psi_1\|_\infty \vert |\bar{y} \vert|_\infty |Z(H_\mathrm{M})|,
	\end{align*}
	with $a>0$, where we have used Assumption \textbf{(Hg)} or \textbf{(HL)} in the last inequality and that $c|\cH|\leq \frac{c^2}{2}+\frac{1}{2}|\cH|^2$.} Since by Proposition \ref{boundedHM}, $Z(H_\mathrm{M})$ is bounded, we can choose a fixed and sufficiently large constant $a$ \tn{such that} the lemma follows.
\endproof

\noindent\textbf{Proof of Theorem \ref{Th:inverse}.}
{Let us denote $\cH:=\cH_A$ with $A=cl(\text{Im}_t(Z))$. As in the proof of Proposition \ref{h2} we may assume that {$cl(\text{Im}(Z))$} is convex. In particular, ${g}_A$ is convex.} %$W=(\hat{W},\tilde{W})^\intercal$ with $\hat{W}=(W^1,\ldots,W^n)^\intercal$ and $\tilde{W}=(W^{n+1},\ldots,W^d)^\intercal$. Recall that
%$\mathcal{Z}_t^*:=(\cH^1(t,X,\zeta,M^1_t),\ldots,\cH^n(t,X,\zeta,M^n_t),Z^{n+1}_t(H_\mathrm{M}),\ldots,Z^d_t(H_\mathrm{M}))$ and denote
%$\tilde{M}:=(M^{n+1},\ldots,M^{d}).$
 Since $\cH$ is single-valued and by Lemma \ref{lineargrove}, grows at most linearly in $M$, by assumption $\E [\int_0^T  \vert\cH(t,X_t,\zeta_t,M_t)\vert^2 \d t]<\infty$. Therefore, $\mathcal{Z}^*=\cH(t,X_t,\zeta_t,M_t) \text{ takes values in }{cl(\text{Im}(Z))}$ and is square-integrable.  %By Proposition \ref{propnotes}(b),  $\nabla {g}_A$ is non-empty (with ${g}_A$ defined in \eqref{def:gbar}). 
 %By \eqref{eq:bijective}, 
 For a $d$-dimensional progressively measurable process $Y$ with $Y_t \in\nabla {g}_A(t,\cH(t,X_t,\zeta_t,M_t))$ 
%in case of (H1) 
it holds that 
{
\begin{align}
\d U'(X_t+\zeta_t) =U''(X_t+\zeta_t) (\mathcal{H}(t,X,\zeta_t,M_t)+M_t-Z_t(H_\mathrm{M})) \d W_t=:U'(X_t+\zeta_t)Y_t \d W_t,
\label{eq:Ude}
\end{align}}
where we used \eqref{eq:FBSDE} in the first, and \eqref{eq:bijective} in the second equation (where $\mathcal{H}$ by Proposition \ref{h2} is single-valued, i.e., a proper function {such that ``$\in$''  in \eqref{eq:FBSDE} becomes ``$=$''}). Hence, $U'(X_t+\zeta_t)$ is always a local martingale.
%with $\cZ_t^{*}= \cH^i(t,X_t,\zeta_t,M_t^i)$. %in case of (H1), and $\cZ_t^{*}= \theta^*(t,X_t,\zeta_t,M_t)Z_t(-S)$ in case of (H2).  
%In case of (H2) we simply set ${g}_A=g$ to unify notation. 
%In case of (H1) this follows by Lemma \ref{lineargrove} while in case of (H2) this follows directly by (\ref{lhsp}) or(\ref{rhsp}).
%In case that $Y_t\in\nabla\hat{g}(t, \cH(t,X_t,\zeta_t,M_t))$, by Lemma \ref{le:BMO} in Appendix \ref{Ap:E}, $W_t^{\Q_0}$ is a standard Brownian motion under $\Q_0$. 
%On the other hand, in case that $g$ is positively homogeneous, \eqref{eq:Yt} holds, and $\cH(t,X_t,\zeta_t,M_t)=\cZ_t^{\theta^*}$
In case of assumption (a) in the theorem, $U'(X_t+\zeta_t)$ by assumption is actually  a positive martingale. That the same holds in case of (b) is seen as follows: By Proposition \ref{boundedHM}, $Z(H_\mathrm{M})$ is bounded. Thus, the couple  $(\zeta, M)$ may be viewed as a solution of a quadratic FBSDE, and by Proposition 2.1 in \cite{kobylanski2000} { (or by classical results on Lipschitz-continuous BSDEs)}, $\zeta$ is bounded.  Hence, by Proposition 3.13 in Barrieu and El Karoui \cite{carmona2008},  
%(see also Mania and Schweizer \cite{mania2005dynamic})
$M$ is a $BMO(\P)$ (see the Appendix for the definition of BMOs). Lemma \ref{lineargrove} entails then that $\cZ^*$ is $BMO(\P)$, and therefore since by assumption in case (b) $U''/U'=1/\psi_1$ is bounded, by \eqref{eq:bijective} $Y$ is $BMO(\P)$. It follows by Kazamaki's criterion then that $U'(X_t+\zeta_t)$ is a positive martingale.

%Therefore, by \eqref{eq:Ude},
%$\d U'(X_t+\zeta_t)=U'(X_t+\zeta_t) G_t\, \d W_t$ for a d-dimensional process $G=(Y,\ldots)$ and $U'(X_t+\zeta_t)$ is a positive martingale.
This discussion entails that
$U'(X_t+\zeta_t)/\E [U'(X_T+\zeta_T)]=\mathcal{E}\bigg(\int_0^t {Y_s} \d W_s\bigg),$
with $\mathcal{E}$ corresponding to the Dol\'eans-Dade exponential.
Furthermore,  by Girsanov's theorem $\d W^{\Q_0}_t=\d W_t-{Y_t} \d t$ is a standard Brownian motion under the measure $\Q_0$ defined by 
$
\frac{\d \Q_0}{\d \P}:=\frac{U'(X_T+\zeta_T)}{\E [U'(X_T+\zeta_T)]}.
$ 
%Now for the first $n$ components of the Brownian motion $W^{\Q_0}$ we get
%$\d W^{\Q_0,i}_t= \d W^i_t-Y^i_t \d t= \d W_t^i-Y^i_t \d t$ for $i=1,\ldots,n.$
\newline

 Let $\wh{X}$ be an arbitrary admissible portfolio corresponding to an ``attainable'' $\wh{Z}$ which satisfies $\E [\int_0^T  \vert\wh{Z}_t\vert^2 \d t]<\infty$ and \text{ takes values in }{$A=cl(\text{Im}(Z))$}. Due to the concavity of $U$, we obtain
 $$
 U(\wh{X}_T+H_L)-U(X_T+H_L)\le U'(X_T+H_L) (\wh{X}_T-X_T).
 $$ % with $(\wh{Z}^1,\ldots,\wh{Z}^n)\in \text{Im}(\cZ)$. 
%\newline
%thaicomment{Mitja, it seems that we need to take $\sum_1^d$ in the sums below? I was thinking that by (\ref{eq:H}) the dimensions $n+1,\ldots,d$ cancel}

%Denote by
%\begin{equation*}
%\hat{\cH}(t,X_t,\zeta_t,M_t)=
%\begin{cases}
%(\cH(t,X_t,\zeta_t,M_t),Z_t(H_\mathrm{M}))&\text{ if (H1)}\\
%\theta^*(t,X_t,\zeta_t,M_t)Z_t(-S)&\text{ if (H2)}.
%\end{cases}
%\end{equation*}
Note that $g$ and ${g}_A$ coincide on  {$cl(\text{Im}(Z))$}. We observe that
\begin{align*}%\label{eq:neualign}
\wh{X}_T-X_T&=\int_0^T (\wh{Z}_t-\cH(t,X_t,\zeta_t,M_t)) \d W_t-\int_0^T ({g}_{{A}}(t,\wh{Z}_t)-{g}_{{A}}(t,\cH(t,X_t,\zeta_t,M_t)))\d t\notag \\
&=\int_0^T (\wh{Z}_t-\cH(t,X_t,\zeta_t,M_t))\d W_t^{\Q_0}
\notag \\
& \hspace{0.2cm} -\int_0^T   \bigg({g}_{{A}}(t,\wh{Z}_t)-{g}_{{A}}(t,\cH(t,X_t,\zeta_t,M_t))-(\wh{Z}_t-\cH(t,X_t,\zeta_t,M_t))
Y_t\bigg)\d t.
\end{align*}
%By the definition of $\cH$ the components $n+1,\ldots,d$ were canceled.
Since $Y\in\nabla{g}_{{A}}(t, \cH(t,X_t,\zeta_t,M_t))=\nabla_c{g}_{{A}}(t, \cH(t,X_t,\zeta_t,M_t))$ the last integral is always non-negative due to the convexity of ${g}_{{A}}$. %\tn{By a similar argument as in \cite{horst2014}}
 Therefore,
\begin{align}\label{underH1}
\E[U(\wh{X}_T+H_L)-U(X_T+H_L)]&\le \E[ U'(X_T+H_L) (\wh{X}_T-X_T)] \notag \\
&=\E [U'(X_T+H_L)]\E^{\Q_0}[\wh{X}_T-X_T]\\
&\le\E [U'(X_T+H_L)]\E^{\Q_0}\bigg[\int_0^T (\wh{Z}_t-\cH(t,X_t,\zeta_t,M_t))\d W^{\Q_0}_t\bigg], \notag 
\end{align}
\tn{where the first inequality is strict if $\wh{X}_T\neq X_T$ while the second inequality is strict if $\wh{Z}_t\neq \cZ^*_t$ and $g$ is strictly convex, giving the uniqueness in this case.}
%Now, it follows that (noting that $\zeta_T=0$)
%\begin{align*}
%\E[U(\wh{X}_T)-U(X_T)]&\le \E[ U'(X_T) (\wh{X}_T-X_T)]=\frac{1}{{\E [U'(X_T)]}}\E^{\Q_0}[(\wh{X}_T-X_T)]\\
%&=\frac{1}{{\E [U'(X_T)]}}\E^{\Q_0}\bigg[\int_0^T (\wh{Z}_t-\cH(t,X,\zeta,M))\d W^Q_t\bigg].
%\end{align*}
%By integration by parts we have
%$$
%U'(X_t+\zeta_t) (\wh{X}_t-X_t)=\int_0^t (\wh{X}_t-X_t) dY_s+\int_0^t Y_s (\wh{Z}_s-Z_s) \d W_s +\int_0^t h(s,Z_s,\wh{Z}_s)\d s 
%$$
%where 
%$$
%h(t,Z_t,\wh{Z}_t):=Y_t(g(t,Z_t)-g(t,\wh{Z}_t)+U''(X_t+\zeta_t) (Z_t+M_t) (\wh{Z}-Z_t)
%$$
%From Condition 3 we can express
%$$
%h(t,Z_t,\wh{Z}_t)=Y_t\big( g(t,Z_t)-g(t,\wh{Z}_t)+g_z(t,Z)(\wh{Z}-Z_t)\big)\le 0 \quad a.s.
%$$
%as $Y$ is positive and $g$ is convex. 
By the BDG inequality and the Cauchy-Schwarz inequality using that $U'(X_T+H_L)$ is square-integrable (\tn{see also \cite{horst2014}}), the last expectation is zero, showing that $X$ is indeed optimal. The uniqueness \tn{in the case where $g$ is Lipschitz} is given by Proposition \ref{uniqueness} below. {The last part follows from Proposition \ref{propnotes}(g).}

{
	\begin{proposition}\label{uniqueness}
		Under the conditions of Theorem \ref{Th:inverse}, Problem \eqref{eq: sternstern} has \tn{at most one solution such that $X_T^{{\cZ}}\in \cD_T$. In particular, if $g$ is Lipschitz, the optimal solution of \eqref{eq: sternstern}  is unique. }
	\end{proposition}
	\proof %Consider from Theorem \ref{Th:inverse} 
	%\begin{equation*}\label{argmax}
	%	\max_{\mathcal{Z}\: \text{\small takes values in } {cl(\text{Im}(Z))}, \cZ \in \mathcal{L}^2(\d \P \times \d s)}\E[U(X_T^{\mathcal{Z}}+H_L)].
	%\end{equation*}
	Let us prove Proposition \ref{uniqueness} by contradiction. Assume there exists $\cZ_1,\cZ_2$ both attaining the maximum in \eqref{eq: sternstern} with $\cZ_1\neq \cZ_2$. Let us first show that then necessarily $X_T^{{\cZ}_1}\neq X_T^{{\cZ}_2}$. Assume by contradiction that $X_T^{{\cZ}_1}= X_T^{{\cZ}_2}$. Then by the uniqueness and continuity of solutions of BSDE, see for the quadratic case Corollary 6 in \cite{briand2008quadratic}, we have $X_t^{\cZ_1}=X_t^{\cZ_2}$ for all $t$ a.s.. 
	Thus, $\cZ_t^1=\frac{\d \langle X^{\cZ_1}, W \rangle_t}{\d t}=\frac{\d \langle X^{\cZ_2}, W \rangle_t}{\d t}=\cZ_t^2 \  \d \mathbf{P}\times \d t$ a.s. which is a contradiction.
	Therefore, indeed $X_T^{{\cZ}_1}\neq X_T^{{\cZ}_2}$.}\\
	
	{
	Now, define $\tilde{\cZ}:=\lambda\cZ_1+(1-\lambda)\cZ_2\in cl(\text{Im}(Z))$ as $cl(\text{Im}(Z))$ is assumed to be convex. Using the convexity of $g$ \begin{align}\label{uniqness2}
		X_T^{\tilde{\cZ}}&=-\int_0^T g(s,\lambda \cZ^1_s+(1-\lambda)\cZ^2_s)\d s +\int_0^T  (\lambda \cZ^1_s+(1-\lambda)\cZ^2_s) \d W_s+\Pi_T(H_\mathrm{M})-\Pi_0(H_\mathrm{M})
		\nonumber \\
		&\geq-\int_0^T \lambda g(s,\cZ^1_s)\d s +\int_0^T \lambda \cZ^1_s \d W_s+\lambda(\Pi_T(H_\mathrm{M})-\Pi_0(H_\mathrm{M}))
		\nonumber\\
		&\hspace{0.5cm}-\int_0^T (1-\lambda)g(s, \cZ^2_s)\d s +\int_0^T (1-\lambda) \cZ^2_s \d W_s+(1-\lambda)(\Pi_T(H_\mathrm{M})-\Pi_0(H_\mathrm{M}))
		\nonumber\\
		&=\lambda\bigg(-\int_0^T g(s,\cZ^1_s)\d s +\int_0^T  \cZ^1_s \d W_s+\Pi_T(H_\mathrm{M})-\Pi_0(H_\mathrm{M})\bigg)
		\nonumber\\
		&\hspace{0.5cm}+(1-\lambda)\bigg(-\int_0^T g(s, \cZ^2_s)\d s +\int_0^T  \cZ^2_s \d W_s+\Pi_T(H_\mathrm{M})-\Pi_0(H_\mathrm{M})\bigg) = \lambda X_T^{\cZ^1}+(1-\lambda)X_T^{\cZ^2}.
\end{align}}
{Hence, 
	\begin{align*}
		\mathbb{E}[U(X_T^{\tilde{\cZ}})]\geq \mathbb{E}[U(\lambda X^{\cZ_1}_T+(1-\lambda)X_T^{\cZ_2})]>\lambda\mathbb{E}[U(X_T^{\cZ_1})]+(1-\lambda)\mathbb{E}[U(X_T^{\cZ_2})]=\max_\cZ \mathbb{E}[U(X_T^\cZ)]
	\end{align*}
	which is again a contradiction. %Hence, $\cZ^*$ is unique, and therefore $X_T^{\cZ^*}$ is unique.
	\endproof
}
{\noindent\textbf{Proof of Corollary \ref{coro37}.}
Suppose that $(X^1,\zeta^1,M^1)$ and $(X^2,\zeta^2,M^2)$ are two solutions to the FBSDE \eqref{eq:FBSDE}. By the uniqueness of $\cZ^*$ shown in Proposition \ref{uniqueness}, we have $\cH_A(t,X_t^1,\zeta_t^1,M_t^1)=\mathcal{Z}_t^*=\cH_A(t,X_t^2,\zeta_t^2,M_t^2)$. By the first equation of \eqref{eq:FBSDE}, the two corresponding forward processes are \tn{then} the same, i.e., $X^1=X^2$. Consequently, by \eqref{eq:prudence}
$$
\zeta_t^1=(U')^{-1}(\mathbb{E}[U'(X_T^1+H_L)\mid \mathcal{F}_t])-X_t^1=(U')^{-1}(\mathbb{E}[U'(X_T^2+H_L)\mid \mathcal{F}_t])-X_t^2=\zeta_t^2.
$$
In particular, $M_t^1=\d \langle\zeta^1,W\rangle_t / \d t=\langle \zeta^2, W\rangle_t / \d t= M_t^2, \ \d \mathbf{P}\times \d t$ a.s., which concludes the proof.\endproof}
 
 \noindent\textbf{Proof of Proposition \ref{optimalL}.}
 %When $g$ is quadratic then the optimal terminal wealth is given by $X^*_T:= f(\lambda)=x_0$, from the budget constraint $\E[e^{\gamma X^*_T}]=e^{f(\lambda)}= e^{\gamma x_0}$. We deduce immediately that $X^*_t=\frac{1}{\gamma}\log(\E^{\Q}[e^{\gamma X_T}\vert \cF_t]) =x_0$, which means that $X_t^*=x_0$ a.e. for all $t\in[0,T]$. From the construction we obtain that it is optimal for the LT to invest nothing, i.e. $\cH(t,X_t^*,\zeta^*_t)=0$ and the triple $(X_t^*=x_0, \zeta_t^*=0,M_t^*=0)$ solves the FBSDE system \eqref{eq:FBSDE}. 
 It is straightforward to see that the triple $(X_t^*=x_0, \zeta_t^*=0,M_t^*=0)$ solves the FBSDE system \eqref{eq:FBSDE}, where $\cH\equiv Z(H_{\mathrm{M}})=\cZ^0$ and satisfies \eqref{eq:bijective}. Clearly, the resulting constant process $U'(x_0)$ is a bounded strictly positive martingale and $\cZ^{\theta^*}=\cZ^0$ is square-integrable. Consequently, Theorem \ref{Th:inverse} applies without needing the assumption that $\psi_1$ is bounded.\endproof
  
\noindent \textbf{Proof of Proposition \ref{multidimintervals}.}
 {
 	%Hence, the dimensions of the problem are separable and therefore $ Z(t,\omega,y)=Z^1(t,\omega,y^1),\ldots Z^n(t,\omega,y^n),Z^{n+1,t}\big(H_\mathrm{M}^\perp\big),\ldots,Z^{d,t}\big(H_\mathrm{M}^\perp\big)) $ with $Z$ defined by Proposition \ref{propmarkov}. 
 	Since by Proposition \ref{propmarkov} the mapping $y\mapsto Z(t,\omega,y)$ is continuous in $y$, the image in $\bbr$ must be an interval.}
% and s solve the BSDE 
%We will omit the argument $yS$ in the sequel.
\endproof \\
\noindent\textbf{Proof of Proposition \ref{Le:Condition3}}\\
{Set $\cH=\cH_A$ with $A=cl(\text{Im}_t(Z))$. Since $g_z(t,\cdot)$ is decreasing and $U''(X+\zeta)<0$, the continuous function 
$\hat{U}_{t,X,\zeta,M}$
is bijective on the interval $[a,b]=cl(\text{Im}(Z))$ so that $ a_t\vee(\hat{U})^{-1}_{t,X,\zeta,M}(0)\wedge b_t\in  cl(\text{Im}_t(Z)) ,$
by Proposition \ref{h2} is well defined being the unique solution of \begin{equation}
	0 \in-U'(X+\zeta) \nabla {g}_A(t,\cH(t,X,\zeta,M)) +U''(X+\zeta)\big(\cH(t,X,\zeta,M){-Z_t(H_\mathrm{M})}+M)\big).
	\label{eq:H}
\end{equation}  
%Therefore, since $\cH^i=Z^i(H_\mathrm{M})$ for $i=n+1,\ldots,d$ and by Proposition \ref{multidimintervals} $\nabla g^{n+1}(t,Z(H_\mathrm{M}))=\bbr^{d-n}$, \eqref{nablaF} holds. %Furthermore, $\cH$ has at most linear growth in $M$ by Lemma \ref{le:BMO} below.
Finally, by Theorem \ref{Th:inverse} and the bijectivity of $\hat{U}$, if $\cZ^*$ is the optimal solution, and we actually must have $\cZ_t^*=\cH(t,X_t,\zeta_t,M_t)= a_t\vee(\hat{U})^{-1}_{t,X,\zeta,M}(0)\wedge b_t$. 
%Recalling that by Proposition \ref{multidimintervals}  we have that $\cZ_t^{\theta^*,i}=Z^{i}_t(H_{\mathrm{M}})$ for $i=n+1,\ldots,d$, we can conclude that \eqref{Ztheta} holds.%$\cH$ indeed takes values in $Image(\mathcal{Z})$.
\endproof }\\

\noindent\textbf{Proof of Proposition \ref{Le:Condition3next}.}
{It can be seen by simple algorithmic manipulations, that on $\theta^*\neq 0$ we have 
\begin{align}\label{eq:3.20}
	Z_t(-\mathrm{sgn}(\theta^*)S)U''(X+\zeta)(|\theta^*| Z_t(-\mathrm{sgn}(\theta^*)S)+M)-U'(X+\zeta) g(t,Z_t(-\mathrm{sgn}(\theta^*)S))  =0,
\end{align}
for each ($\omega,t,X,\zeta,M)$ if and only if $\theta^*$ satisfies (\ref{lhsp}) or (\ref{rhsp}).} %in case that $g$ is positively homogeneous, 
%In addition, equation (\ref{eq:3.20}) holds on $\theta^*= 0$ with $\leq$ ($\geq$) instead of ``='' with $\mathrm{sgn}(\theta^*(t,X,\zeta,M))$ in (\ref{eq:3.20}) being replaced by $+1(-1)$.}\\ %follows now from solving \eqref{eq:3.20} for $\theta^*$, distinguishing the cases $\mathrm{sgn}(\theta^*(t,X,\zeta,M))=1$ or $\mathrm{sgn}(\theta^*(t,X,\zeta,M))=-1$.\\

{
Next let us show that an optimal $\theta^*$ solves \eqref{eq:3.20}. %with $\leq(\geq)$ on $\theta^*\neq0$, and  \eqref{eq:3.20} with $\sign( \theta^*)=1(-1)$.
By Remark \ref{HM}, $Z(H_\mathrm{M})=0$. Furthermore, by Proposition \ref{propmarkov},} { $\cZ^{\theta^*}=|\theta^*|Z(-\sign(\theta^*)S)$ and hence $A=\text{Im}_t(Z)=\{\theta Z_t(-S)|\theta \geq 0\} \cup \{-\theta Z_t(S)|\theta<0\}$. For $\xi \in \nabla g(t,|\theta| Z_t(-\sign(\theta)S))$ on $\theta>0$, it holds, by the definition of a subgradient, that  
\begin{align*}
	\pm \xi Z_t(-S) &\leq D{g}'_A(t, \theta Z_t(-S),\pm Z(-S))\\
	&= Dg'(t,\theta Z_t(-S),\pm Z_t(-S))\\
	&=\lim_{\lambda\downarrow 0}\frac{1}{\lambda}\bigg(g(t,\theta Z_t(-S)\pm\lambda Z_t(-S))-g(t,\theta Z_t(-S))\bigg)\\
	&=\lim_{\lambda\downarrow 0}\frac{1}{\lambda}\bigg((\theta\pm\lambda)g(t,Z_t(-S))-\theta g(t,Z_t(-S))\bigg)=\pm g(t,Z_t(-S)),
\end{align*} for $\theta>0$ by positive homogeneity of $g$. Multiplying this inequality by $-1$ it follows that $\xi Z_t(-S)=g(t,Z_t(-S))$. Similarly on $\theta<0$, \begin{align*}
\xi Z_t(S)&\leq\lim_{\lambda\downarrow 0}\frac{1}{\lambda}g(t,-\theta Z_t(S)\pm \lambda Z_t(S))-g(t,-\theta Z_t(S))\\
&=\lim_{\lambda\downarrow 0}\frac{1}{\lambda}\bigg\{(-\theta \pm \lambda)g(t, Z_t(S))+\theta g(t,Z_t(S))\bigg\}=\pm g(t,Z_t(S)).
\end{align*} Multiplying both sides by $-1$ we again obtain $\xi Z_t(S)=g(t,Z_t(S))$ on $\theta<0$. In particular, $\xi Z_t(-\sign(\theta)S)=g(Z_t(-\sign(\theta)S))$. Hence, multiplying \eqref{eq:rem} with $Z_t(-\sign(\theta^*)S)$ implies that on $\theta^*_t \neq 0$ we have}
{
\begin{equation} \label{eq:derivativesnew}
	-U'(X^{\theta^*}_t+\zeta_t) g(t,Z_t(-\mathrm{sgn}(\theta_t^*)S)) + Z_t(-\mathrm{sgn}(\theta_t^*)S) U''(X^{\theta^*}_t+\zeta_t) (\cZ_t^{\theta^*} + M^\zeta_t)^\intercal =0.
\end{equation} 
We can use ``$=$'' instead of ``$\in$'' in \eqref{eq:rem}, because, by the argument above, the right-hand side of \eqref{eq:rem} only consists of one point on $\theta^*\neq 0$. Similarly, on $\theta^*_t = 0$ by the argument above the following inequalities hold:
\begin{eqnarray} 
	& U''(X^{\theta^*}_t+\zeta_t) M^\zeta_t Z^\intercal_t(-S) \leq U'(X^{\theta^*}_t+\zeta_t)g(t,Z_t(-S))\label{eq:starstarnew} \\
	& U''(X^{\theta^*}_t+\zeta_t) M^\zeta_t Z^\intercal_t(S) \leq U'(X^{\theta^*}_t+\zeta_t)g(t,Z_t(S))\label{eq:starstar1new}.
\end{eqnarray}
%where $M^\zeta_t:=\d \langle \zeta,W^\intercal \rangle_t/\d t$. 
\eqref{eq:3.20} can now be seen from \eqref{eq:derivativesnew}. Furthermore, if the left-hand side of \eqref{lhsp} is negative and at the same time the left-hand side of \eqref{rhsp} is positive, we must have $\theta^*=0$ and  \eqref{eq:starstarnew}-\eqref{eq:starstar1new} holds. }
%\eqref{eq:derivativesnew}-\eqref{eq:starstarnew}-\eqref{eq:starstar1new} that \eqref{eq:H}-\eqref{eq:3.20} holds.
\endproof

\noindent\textbf{Proof of Theorem \ref{th:generalcase}.} 
{Note that we may assume by \textbf{(H0)} that $cl(\text{Im}(Z))$ is convex, see also the proof of Proposition \ref{h2}. Furthermore, by Proposition \ref{h2}, $\cH$ is a function and the stochastic inclusion FBSDE \eqref{eq:FBSDE} therefore becomes a proper FBSDE.} Consider the problem
\begin{equation}\label{eq:appendC}
\underset{\cZ \: \mbox{\small  takes values in } {cl(\text{Im}(Z)}),\: \cZ\in L^2(d\mathbf{P}\times \d t)}{\sup}\E\bigg[U(X_T^\cZ+H_L)\bigg].
\end{equation}
For the rest of the proof we will call the $\cZ$'s over which the supremum in (\ref{eq:appendC}) is taken admissible.
%Theorem \ref{th:generalcase} would follow from Theorem \ref{Th:FB} provided that the supremum in (\ref{eq:EU.1}) is attained for some $\theta^*$ and provided that $\E[|U(X^{\theta^*}_T+H_L)|]<\infty   $ and $\E[|U'(X^{\theta^*}_T+H_L)|^{1+\tilde{\epsilon}}]<\infty$. 
{Replacing in the proof of Theorem \ref{Th:FB} $\text{Im}_t(Z)$ with $cl(\text{Im}_t(Z))$, it can be seen that} to show Theorem \ref{th:generalcase} it is actually sufficient to prove that the supremum in (\ref{eq:appendC}) is attained for some $\cZ^*$ and that
$\E[|U(X^{\cZ^*}_T+H_L)|]<\infty   $ and $\E[|U'(X^{\cZ^*}_T+H_L)|^{1+\tilde{\epsilon}}]<\infty   $.
To prove that this holds in case of (i), we first show that the space over which the supremum in \eqref{eq:appendC} is taken is weakly compact, and, in addition to being concave, the function to be optimized is upper semi-continuous. \newline
\textbf{Proof of (i).}\newline
First of all, note that $\E[X_T^\mathcal{Z}]$ is finite for any admissible $\cZ$ by \textbf{(Hg)} or \textbf{(HL)}. In particular $\E[ U(X^\cZ_T+H_L)]<\infty$.  Furthermore, since the mapping $\mathcal{Z}\mapsto X_T^{\mathcal{Z}}$ is concave in $\mathcal{Z}$ {and since by \textbf{(H0)}, we may assume that ${cl(\text{Im}(Z))}$ is convex as well}, we have as in \eqref{uniqness2}
\begin{align*}
U\bigg(X_T^{\lambda \mathcal{Z}_1+(1-\lambda)\mathcal{Z}_2}+H_L\bigg)&\geq U\bigg(\lambda X_T^{\mathcal{Z}_1}+H_L+(1-\lambda)X_T^{\mathcal{Z}_2}+H_L\bigg)\nonumber\\ 
&\geq\lambda U(X_T^{\mathcal{Z}_1}+H_L)+(1-\lambda)U(X_T^{\mathcal{Z}_2}+H_L).
\end{align*}
It follows that also the mapping $\mathcal{Z}\mapsto\E[U(X_T^{\mathcal{Z}}+H_L)]$ from $L^2(\d \textbf{P}\times\d s)$ to $\R\cup\{-\infty\}$ is concave. 
Now, to show that a concave functional attains its maximum, it is sufficient to show that we may restrict the maximization problem to a weakly compact set, and that the functional is (strongly) upper semi-continuous.

{\it Weak Compactness}: %In the sequel $K>0$ is a constant which may change from line to line. 
Note that we may restrict ourselves to $\cZ$ such that
$
\E\bigg[U(x_0+H_L)\bigg]\leq\E\bigg[U(X_T^{\cZ}+H_L)\bigg]
\leq a+\tilde{K}\E\bigg[X_T^{\cZ}+H_L\bigg],
$ for a certain $a$ and $\tilde{K}>0$, where the second inequality holds as $U$ grows at most linearly by concavity. 
Hence, we may restrict the optimization problem to wealth processes satisfying
\begin{equation}\label{eq:wealth}
\E[X_T^\cZ+H_L]\geq\frac{\E\big[U(x_0+H_L)\big]}{\tilde{K}}-\frac{a}{\tilde{K}}.
\end{equation}
For a suitable $\hat{K}>0$, the corresponding $\mathcal{Z}\in L^2(\d\P\times\d s)$ of such a wealth process satisfies
\begin{align*}
\E\bigg[\int_0^T\mid \mathcal{Z}_s|^2\d s\bigg]&\leq \hat{K}\bigg(1+\E\bigg[\int_0^T g(s,\mathcal{Z}_s)\d s\bigg]\bigg)\\
%&=K\bigg[1+\E\bigg[-x_0+\int_0^T g(s,\mathcal{Z}_s)\d s-\int_0^T\mathcal{Z}_s\d W_s+x_0+\int_0^T\mathcal{Z}_s\d W_s\bigg]\bigg]\\
&=\hat{K}\bigg(1+x_0+\E[H_\mathrm{M}]-\Pi_0(H_\mathrm{M}) \\
&\hspace{0.5 cm}+\E\bigg[-x_0+\int_0^T g(s,\mathcal{Z}_s)\d s-\int_0^T\mathcal{Z}_s\d W_s-H_\mathrm{M}+\Pi_0(H_\mathrm{M})\bigg]\bigg)\\
&= \hat{K}\bigg(1+x_0+\E[H_\mathrm{M}]-\Pi_0(H_\mathrm{M})+\E[-X^\cZ_T]\bigg)\\
&\leq \hat{K}\bigg(1+x_0+\E[H_\mathrm{M}]-\Pi_0(H_\mathrm{M})+\frac{a}{\tilde{K}}+\E[H_L]-\frac{\E\big[U(x_0+H_L)\big]}{\tilde{K}}\bigg).
\end{align*}
The first inequality holds since $g$ is assumed to grow at least quadratically. The last inequality holds by (\ref{eq:wealth}). Hence, there exists $K'>0$ with $\E\big[\int_0^T|\cZ_s|^2\d s\big]\leq K'$ and we can set $\E\big[U\big(X_T^\cZ+H_L\big)\big]=-\infty$ for any other $\cZ$ without affecting the maximization problem.\par
\noindent {\it Upper semi-continuity}: If $\mathcal{Z}^m\rightarrow\mathcal{Z}$ in $L^2(\d \textbf{P}\times\d s)$ as $m$ tends to infinity, then there exists a subsequence converging a.s., which guarantees that $\mathcal{Z}\text{ takes values in } {cl(\text{Im}(Z)})$ (pointwise closure). By (\textbf{Hg}) or (\textbf{HL}), $|g(s,\mathcal{Z}_s^m)|\leq K(1+|\mathcal{Z}_s^m|^2),$ which is uniformly integrable in $L^1(\d \textbf{P}\times\d s)$. Therefore,
\begin{align*}
-K&\leq\int_0^Tg(s,\mathcal{Z}_s^m)\d s\overset{L^1}{\rightarrow}\int_0^Tg(s,\mathcal{Z}_s)\d s\hspace{0.5cm}\text{ and }\hspace{0.5cm}\int_0^T\mathcal{Z}_s^m\d W_s\overset{L^2}{\rightarrow}\int_0^T\mathcal{Z}_s\d W_s, \quad \mbox{ as }m\to\infty.
%&u(\mathcal{Z})\geq\lim\sup u(\mathcal{Z}^m) .... u(x_0)
\end{align*}
Thus, $X^{\mathcal{Z}^m}_T\rightarrow X^{\mathcal{Z}}_T$ in $L^1$ as $m\to\infty$. Hence, denoting by $(\cdot)^+$ the positive part of a number (and the negative part similarly) we have $\big(U\big(X^{\mathcal{Z}^m}_T+H_L \big)\big)^+\leq\big(a+K(X^{\mathcal{Z}^m}_T+H_L)\big)^+$
%\begin{equation*}
%\big(U\big(X^{\mathcal{Z}^m}+H_L \big)\big)^+\leq\big(a+K(X^{\mathcal{Z}^m}+H_L)\big)^+
%\end{equation*}
is uniformly integrable. Denoting by $m_j$ a subsequence which attains the $\lim\sup$ and by $m_{j_l}$ a subsequence of $m_j$ along which $X^{\mathcal{Z}^{m_{j_l}}}_T$ converges to $X_T$ a.s., we obtain by Fatou's generalized Lemma
\begin{align*}
\underset{m}{\lim\sup}\hspace{0.1cm}\E\big[U\big(X^{\mathcal{Z}^m}_T+H_L\big)\big]=\underset{m_{j_l}}{\lim\sup}\hspace{0.1cm}\E\big[U\big(X_T^{\mathcal{Z}^{m_{j_l}}}+H_L\big)\big]\leq\E\big[U\big(X_T^{\mathcal{Z}}+H_L\big)\big].
\end{align*}
This shows upper semi-continuity.
\vspace{2mm}\newline
Upper semi-continuity  together with the weak compactness of the optimization set guarantees for a concave functional the existence of an optimal $\cZ^*$.
Next note that by {\bf (Hg)} or {\bf (HL)}, $X^{\mathcal{Z}^*}_T \in L^1$ which entails that for the positive utility part we have $\E[U^+(X^{\mathcal{Z}^*}_T+H_L)]<\infty.$ 
Since $\theta^*=0$ is an admissible strategy with corresponding wealth process $x_0$, we obtain $-\infty<U(x_0)\leq \E[U(X^{\mathcal{Z}^*}_T+H_L)]$ so that we can not have that  $\E[U^-(X^{\mathcal{Z}^*}_T+H_L)]=-\infty$. Hence, overall 
$\E[|U(X^{\mathcal{Z}^*}_T+H_L)|]<\infty.$ From this and from the fact that the first moment of the optimal wealth process exists, we obtain from our growth conditions on $U'$ that also $\E[|U'(X^{\mathcal{Z}^*}_T+H_L)|^{1+\tilde{\epsilon}}]<\infty.$ By the remark in the beginning of the proof, (i) follows.\\

\noindent\textbf{Proof of (ii).} Since by assumption $U(x)=U(w_0)$ for all $x\geq w_0$ we may assume that for the optimal %$\theta^*$ with corresponding 
${\cZ}^*$ we have that $X_T^{\cZ^*}\leq w_0$. This can be seen as follows: {Recall that by Remark \ref{HM}, $\cZ^0=Z(H_M)$ corresponds to a $\cZ$ induced by investing nothing, which by \eqref{eq:stern} entails that the corresponding wealth process $X$ is constant.} Denote $\tau=\inf\{t>0\mid X_t\geq w_0\}$ and set \begin{equation}
\label{Ztau}
\cZ^\tau_t={\cZ}^*_t\mathbf{1}_{t\leq\tau}+Z_t(H_{\mathrm{M}})\mathbf{1}_{t>\tau}.\end{equation} Then $\cZ^\tau$ %corresponds to the strategy $\tilde{\theta}_t=\theta^*_t\mathbf{1}_{t\leq\tau}+0\mathbf{1}_{t>\tau}$ and 
is admissible (i.e., square integrable and taking values in $ {cl(\text{Im}(Z))}$) and $X_T^{\cZ^\tau}=X^{\cZ^*}_{T\wedge\tau}\leq w_0$. Furthermore, clearly 
\begin{align*}
\E\big[U(X_T^{\cZ^\tau})\big]&=\E\big[U(X_T^{{\cZ}^*})\mathbf{1}_{T<\tau}\big]+\E\big[U(w_0)\mathbf{1}_{T\geq\tau}\big]\\
&\geq \E\big[U(X_T^{{\cZ}^*})\mathbf{1}_{T<\tau}\big]+\E\big[U(X_T^{{\cZ}^*})\mathbf{1}_{T\geq\tau}\big]= \E\big[U(X_T^{Z^*})\big],
\end{align*}
so that $\cZ^\tau$ is optimal as well. Therefore, it is indeed sufficient to consider $\cZ$ with $X_T^{\cZ}\leq w_0$. 

Next, assume first that \textbf{(HL)} holds and therefore $g$ is Lipschitz-continuous. 
Let $X_T^{\cZ^m}$ be a sequence with   $X_T^{\cZ^m}\leq w_0$ whose expected utility converges as an increasing sequence to the supremum in \eqref{eq:appendC} as $m$ tends to infinity, i.e., $\E\big[U(X_T^{\cZ^1})\big]\leq \E\big[U(X_T^{\cZ^2})\big]\leq \ldots$ and
\begin{equation}
\label{vstar} -\infty < v^*:=\lim_m \E\big[U(X_T^{\cZ^m})\big]= \underset{\cZ \: \text{\small taking values in } {cl(\text{Im}(Z))}}{\sup}\E\big[U(X_T^\cZ)\big]  .
\end{equation}
{Assume by contradiction that $||X_T^{\cZ^m}||_2 \to \infty$ as $m\to\infty$.}
As $X_T^{\cZ^m}$ is bounded from above, we have $||\big(X_T^{\cZ^m}\big)^{-}||_2\rightarrow\infty$ and consequently also $||\big(X_T^{\cZ^m}\big)^-||_2\rightarrow\infty$ as $m\to\infty$. Therefore, as $U^-$ decreases faster than quadratically,
\begin{align*}
\E\big[U(X_T^{\cZ^{m}})\big]&\leq\E\big[-\big(U(X_T^{\cZ^m})\big)^-\big]+U(w_0)\\
&\leq -K'\big|\big|(X_T^{\cZ^m})^-\big|\big|_2^2+K'+U(w_0)\underset{m\rightarrow\infty}{\rightarrow}-\infty,
\end{align*}
for some $K'>0$, which {is a contradiction to \eqref{vstar}}. The sequence $X_T^{\cZ^m}$ must therefore be bounded in $L^2$ entailing that there exists a subsequence (for the sake of simplicity again denoted by $m$) which converges weakly to an $X^*\in L^2.$ 
 
\noindent{\it
\textbf{Claim A:} $X^*$ is admissible (i.e., $X^*=X^{\cZ^*}_T$ for an admissible $\cZ^*$) and attains the supremum in (\ref{eq:appendC}).}

Let us show that Claim A holds. By Mazur's lemma there exists a function ${\displaystyle N:\mathbb {N} \to \mathbb {N} }$ and a sequence of  non-negative numbers 
${\displaystyle \left\{\lambda (m)_{k}:k=m,\dots ,N(m)\right\}}$
such that
${ \sum _{k=m}^{N(m)}\lambda (m)_{k}=1}$ and the sequence
defined by the convex combination
$ \tilde{X}^m_T:=\sum _{k=m}^{N(m)}\lambda (m)_{k} X_T^{\cZ^k}$ has the property that 
$ \tilde{X}^m_T$ converges in $L^2$ to $X^*$ as $m$ tends to infinity. By switching to a subsequence if necessary we may assume that additionally convergence holds a.s. 
Finally, as $X_T^{\cZ^k}\leq w_0$ we have that
\begin{equation}
\label{kleineru0}
\tilde{X}^m_T\leq w_0, \quad\mbox{ and therefore }\quad  X^*\leq w_0,
\end{equation}  
as well.
%By the concavity of $U$ we have that 
%$\geq v^*.$
Denote $\hat{\cZ}^m:=\sum _{k=m}^{N(m)}\lambda (m)_{k}\cZ^k$ which is in $L^2(\d \P\times \d t)$ since the $\cZ^k$ themselves are by construction admissible and hence in $L^2(\d \P\times \d t)$. Furthermore, by assumption, $ {cl(\text{Im}(Z))}$ is convex. Therefore, $\hat{\cZ}^m \: \text{takes values in } {cl(\text{Im}(Z))}$ as a convex combination of elements of $ {cl(\text{Im}(Z))}.$  %(Both assumptions respectively guarantee that $\text{Im}(\cZ)$ is actually an interval.) 
Hence, under our assumptions  $\hat{\cZ}^m$ is admissible.

Next, note that by concavity of $-g$ we have that

\begin{align}\label{nstern}
X^{\hat{\cZ}^m}_T&{=-\int_0^T g(s,\sum _{k=m}^{N(m)}\lambda (m)_{k}\cZ^k_s)\d s +\int_0^T \sum _{k=m}^{N(m)}\lambda (m)_{k}\cZ^k_s \d W_s+\Pi_T(H_\mathrm{M})-\Pi_0(H_\mathrm{M})
}\nonumber \\&
{\geq-\sum _{k=m}^{N(m)}\lambda (m)_{k}\int_0^T g(s,\cZ^k_s)\d s +\sum _{k=m}^{N(m)}\int_0^T \lambda (m)_{k}\cZ^k_s \d W_s+\Pi_T(H_\mathrm{M})-\Pi_0(H_\mathrm{M})
}\nonumber\\
&{=\sum _{k=m}^{N(m)}\lambda (m)_{k}\bigg(-\int_0^T g(s,\cZ^k_s)\d s +\int_0^T \cZ^k_s \d W_s+\Pi_T(H_\mathrm{M})-\Pi_0(H_\mathrm{M})\bigg)
}\nonumber\\
&= \sum _{k=m}^{N(m)}\lambda (m)_{k} X_T^{\cZ^k}=\tilde{X}^m_T.
\end{align}
Thus, 
\begin{equation}
\label{starstar}
\liminf_m \E\big[U(X_T^{\hat{\mathcal{Z}}^m})\big]
\geq \liminf_m \E\big[U(\tilde{X}^m_T)\big]\geq \liminf_m \E\big[U(X_T^{\cZ^m})\big]= v^*,
\end{equation}
where the last inequality follows as $U$ is concave, $\tilde{X}^m$ is a convex combination of $X_T^{\cZ^m},X_T^{\cZ^{m+1}}\ldots,$ and by construction $\E\big[U(X_T^{\cZ^m})\big]$ increases in $m$. Hence, since $\hat{\cZ}^m$ are admissible by the definition of $v^*$ in (\ref{vstar}), all equalities in (\ref{starstar}) must be equalities. In particular, $X_T^{\hat{\mathcal{Z}}^m}$ itself is an admissible maximizing sequence and the same holds true for the sequence which ``stops'' at $w_0$, $X_T^{\mathcal{\hat{\mathcal{Z}}}^{m,\tau^m}}$ with $\tau^m=\inf\{t>0 \mid X^{\hat{\mathcal{Z}}^m}_t\geq w_0\}$ and $\mathcal{\hat{\mathcal{Z}}}^{m,\tau^m}={\mathcal{\hat{\mathcal{Z}}}^m}_t\mathbf{1}_{t\leq\tau^m}+Z_t(H_{\mathrm{M}})\mathbf{1}_{t>\tau^m}$ as defined in (\ref{Ztau}). Furthermore, by (\ref{starstar})
and Fatou's lemma we have that
\begin{equation}
\label{superoptimal}
v^*= \lim_m \E\big[U(\tilde{X}^m_T)\big]\leq \E\big[U(X^*)\big],
\end{equation}
showing that $X^*$ is superoptimal. Now to show Claim A, we will first show the following Claim B:

\noindent{\it \textbf{Claim B:} $X_T^{\hat{\mathcal{Z}}^{m,\tau^m}}$ converges to $X^*$ in $L^2.$}

To see this, the first
key observation is that the admissible wealth process $X_T^{\hat{\mathcal{Z}}^{m,\tau^m}}=X^{\hat{\mathcal{Z}}^m}_{T\wedge\tau}=X_T^{\hat{\mathcal{Z}}^{m}} {\bf 1}_{\tau^m>T}+w_0  {\bf 1}_{\tau^m\leq T}$ is either equal to $X_T^{\hat{\mathcal{Z}}^m}$ (which, {by \eqref{nstern}}, is greater {or equal} than $\tilde{X}_T^m$) or is equal to $w_0$ (which by (\ref{kleineru0}) is greater {or equal} than $\tilde{X}_T^m$ as well). Thus,
\begin{equation}
\label{ui}
\tilde{X}_T^m\leq X_T^{\hat{\mathcal{Z}}^{m,\tau^m}} \leq w_0.
\end{equation}     
As $\tilde{X}_T^m$ converges to $X^*$ in $L^2$, $|X_T^{\hat{\mathcal{Z}}^{m,\tau^m}}|^2$ is uniformly integrable. Next, let us show that $X_T^{\hat{\mathcal{Z}}^{m,\tau^m}}$ converges in probability to $X^*$ as $m$ tends to infinity which would prove Claim B. To this end, recall that $\tilde{X}_T^m$ converges to $X^*$ a.s. and that therefore by (\ref{ui})
\begin{equation}
\label{liminf}
w_0\geq \liminf_m X_T^{\hat{\mathcal{Z}}^{m,\tau^m}}\geq \lim_m \tilde{X}_T^m=X^*.
\end{equation}
Clearly, on the set $A:=\{X^*=w_0\}$ we must have that all inequalities in (\ref{liminf}) are equalities. Since an analogous argument holds {when taking the $\limsup$ instead of the $\liminf$ in \eqref{ui}}, we have that
\begin{equation}
\label{first}
\lim_m X_T^{\hat{\mathcal{Z}}^{m,\tau^m}}=X^*,\quad\mbox{ on } A.
\end{equation}
Last, let us remark that by Lemma \ref{Le: conv}, also on  
$A^c= \{X^*<w_0\}, $ (the complement of $A$) 
we actually must have that 
$X_T^{\hat{\mathcal{Z}}^{m,\tau^m}}$  converges to $X^*$ in probability as well. As shown in Lemma \ref{Le: conv} below, the main idea is that from (\ref{liminf}) we already know that the $\liminf$ of $X_T^{\hat{\mathcal{Z}}^{m,\tau^m}}$ dominates $X^*$ a.s. Observe that this domination can actually not be strict on $A^c$, as this would be a contradiction  
that by (\ref{superoptimal}) $X^*$ is superoptimal. Now, the fact that $X_T^{\hat{\mathcal{Z}}^{m,\tau^m}}$ converges in probability to $X^*$ together with  uniform integrability of $|X_T^{\hat{\mathcal{Z}}^{m,\tau^m}}|^2$ proves Claim B. %A formal proof of this last statement can be found in Appendix \ref{Ap:G}. This shows Claim B.

Next, let us complete proving Claim A. By Claim B, we have that the admissible terminal wealth $X_T^{\hat{\mathcal{Z}}^{m,\tau^m}}$ converges to $X^*$ in $L^2$. From well known results on the continuity of BSDEs with Lipschitz-continuous driver in the terminal condition, it follows then that the tuple 
\begin{align*}
&\big(\Pi_t(-X_T^{\hat{\mathcal{Z}}^{m,\tau^m}}+H_\mathrm{M}-\Pi_0(H_\mathrm{M})),\hat{\mathcal{Z}}^{m,\tau^m}_t\big)\\
&=\Big(-X_T^{\hat{\mathcal{Z}}^{m,\tau^m}}+H_\mathrm{M}-\Pi_0(H_\mathrm{M})-\int_t^Tg(t,\hat{\mathcal{Z}}^{m,\tau^m}_t)\d t +\int_t^T\hat{\mathcal{Z}}^{m,\tau^m}_t\d W_t,\hat{\mathcal{Z}}^{m,\tau^m}_t\Big)
\end{align*} converges to
$$\big(\Pi_t(-X^*+H_\mathrm{M}-\Pi_0(H_\mathrm{M})),\hat{\mathcal{Z}}^{*}_t\big)=\Big(-X^*+H_\mathrm{M}-\Pi_0(H_\mathrm{M})-\int_t^Tg(t,\hat{\mathcal{Z}}^{*}_t)\d t +\int_t^T\hat{\mathcal{Z}}^{*}_t\d W_t,\hat{\mathcal{Z}}^{*}_t\Big)$$
in $\mathcal{S}^2\times L^2(\d\P\times\d s)$ with $\wh{\cZ}^*$ being the unique square-integrable process from (\ref{bsdeX}) {with terminal condition $-X^*+H_\mathrm{M}-\Pi_0(H_\mathrm{M})$}. % \mitjacomment{check if reference is correct}. 
Note that $X^*=X_T^{\wh{\cZ}^*}$ (plug in $t=0$ above and compare the first component of the tuple with \eqref{eq:stern}, see also below \eqref{eq:stern}). By switching to a subsequence if necessary we may assume that $\hat{\mathcal{Z}}^{m,\tau^m}$ converge to $\wh{\cZ}^*$ $\d\P\times\d s$ a.s. so that $\wh{\cZ}^*\: \text{takes values in } {cl\big(\text{Im}(Z)\big)} $. In particular, $X^*$ is an admissible portfolio of problem (\ref{eq:appendC}), and  by (\ref{superoptimal}) attains the optimum, showing (ii) the existence of an optimal $\mathcal{Z}^*$ in case the driver is Lipschitz.

Finally, if the driver is not Lipschitz, it grows at least quadratically { by assumption}. In particular, (i) already shown above applies showing the existence of an optimal $\mathcal{Z}^*$. For both cases \textbf{(HL)} and \textbf{(Hg)} it can be checked as in (i) that the integrability conditions of $U(X_T^{\mathcal{Z}^*})$ and $U'(X_T^{\mathcal{Z}^*})$ are satisfied so that by the remark above, (ii) follows.\\

\noindent\textbf{Proof of (iii).} So assume that $U$ is a CARA utility function, i.e. $U(x)=a-be^{-\gamma x}$ for $a\in\R$ and $b,\gamma>0$. Then, from \eqref{eq:bijective} we deduce that $\cH(t,X,\zeta,M)=\cH(t,M)$ does not depend on $(X,\zeta)$. % is of at most linear growth in $M$ by Lemma \ref{lineargrove} (see Appendix \ref{se:lineargrove}) and given by 
%\begin{equation}\label{eq:ghm}
%0\in \nabla g(t,\cH(t,M))+\gamma\cH(t,M)+\gamma M.
%\end{equation} 
The FBSDE in \eqref{eq:FBSDE} decouples and becomes then 
\begin{equation}
\zeta_t=H_L-\int_t^T M_s\d W_s+\int_t^Tf(s,M_s) \d s, \quad \zeta_T=H_L,
\label{eq:Cara1}
\end{equation}
where, {
$
f(t,M)=\tn{-}\frac{\gamma}{2}\vert \cH(t,M)-Z_t(H_\mathrm{M}))+M\vert^2+g(t,Z_t(H_\mathrm{M}))-g(t,\cH(t,M_t)). $}
%\thaicomment{We need differentiability $g_z$ in order to make sure that $f$ is differentiable in $M$?}\mitjacomment{I added a sentence below. Basically in the Lipschitz-continuous case we might not need it.}Furthermore, by the identity \eqref{eq:ghm} we obtain $g_{zz}(t,\cH(t,z))\cH_z(t,z)+\gamma\cH_z(t,z)+\gamma=0$, which implies that
%$$
%f_z(t,z)=\frac{1}{\gamma}g_{zz}(t,\cH(t,z))g_z(t,\cH(t,z))\cH_z(t,z) +g_z(t,\cH(t,z))\cH_z(t,z)=-g_z(t,\cH(t,z)).
%$$
%%Using that by Assumption \ref{Le:Condition3}, $H$ grows at most linearly 
%
%Hence, there exist a deterministic function $k'_t\in\cL^2[0,T]$ and a constant $C>0$ such that  $
%\vert f_z(t,z)\vert \le k'_t+C \vert z\vert, \quad \forall\, (t,z)\in [0,T]\times \bbr.
%$
%Noting $g(t,0)=0$ we obtain 
%$$g(t,\cH(t,M))=\int_0^{\cH(t,M)} g_z(t,z) \d z=-\int_0^{\cH(t,M)} \gamma (z+ M) \d z=-\gamma (\cH^2(t,M)/2+M\cH(t,M)).$$
%So $f(t,M)=\frac{\gamma}{2}M_t^2. $ 
From Proposition \ref{multidimintervals}, we have that $Z_t(H_\mathrm{M})$ is bounded. Furthermore, by Lemma \ref{lineargrove}, $\cH$ grows at most linearly in $M$, implying that $f$ grows at most quadratically. The existence of a solution to the FBSDE follows then from Theorem 2.3 in \cite{kobylanski2000}.\qed \\
\begin{lemma}\label{Le: conv}
In the proof of Theorem \ref{th:generalcase} (ii), $X^{\hat{\mathcal{Z}}^{m,\tau^m}}_T$ converges to $X^*$ on $A^c$ in probability
\end{lemma}
\label{Ap:G}
\proof
Let us show that on $A^c= \{X^*<{w_0}\} , $
$X_T^{\hat{\mathcal{Z}}^{m,\tau^m}}$ converges to $X^*$ in probability. 
First of all note that if we could show that for any arbitrary $\delta>0$,
\begin{equation}
\label{convprop}
 \lim_m\P[X_T^{\hat{\mathcal{Z}}^{m,\tau^m}}>X^*+\delta,\,\,X^*<w_0 ]=0,
\end{equation}
then
\begin{align*}\lim_m \P[|X_T^{\hat{\mathcal{Z}}^{m,\tau^m}}-X^*        |>\delta,\,\,X^*<w_0]& \leq \limsup_m \P[X_T^{\hat{\mathcal{Z}}^{m,\tau^m}}>X^*+\delta,\,\,X^*<w_0]\\
& \hspace{0.5cm} + 
\limsup_m \P[X_T^{\hat{\mathcal{Z}}^{m,\tau^m}}<X^*-\delta,\,\,X^*<w_0]\\
&=0 ,
\end{align*}
where the second term converges to zero by \eqref{liminf}. So all what is left is to show  (\ref{convprop}). Assume by contradiction that (\ref{convprop}) does not hold and we have that
$ \P[X_T^{\hat{\mathcal{Z}}^{m,\tau^m}}>X^*+\delta,X^*<w_0 ]>2\varepsilon ,$
for a subsequence again indexed by $m$ and an $\varepsilon>0.$ Then there exists $N_0\in\mathbb{N}$ such that for all $m$
$$ \P[C_{m}]:=  \P\bigg [X_T^{\hat{\mathcal{Z}}^{m,\tau^m}}>X^*+\delta,X^*<w_0-\frac{1}{N_0} \bigg]>\varepsilon. $$
Let $C_m^c$ denote the complement of $C_m$. Note that $U(x+y)-U(x)$ is decreasing in $x$ and increasing in $y$. Since $U(x)$ is constant for $x+y\geq \omega_0$, this entails that on $X^*<w_0-\frac{1}{N_0}$ we have
$U(X^*+\delta)-U(X^*)\geq U(w_0)- U(w_0-\min(\delta,\frac{1}{N_0}))   .$
 Therefore,
\begin{align*}
\lim_m \E[& U(X_T^{\hat{\mathcal{Z}}^{m,\tau^m}})]\\
&= \lim_m\bigg\{ \E[U(X_T^{\hat{\mathcal{Z}}^{m,\tau^m}}){\bf 1}_{C_m}]+ \E[U(X_T^{\hat{\mathcal{Z}}^{m,\tau^m}}){\bf 1}_{C^c_m}]\bigg\}\\
&\geq \limsup_m\bigg\{ \E[U(X^*+\delta){\bf 1}_{C_m}]+ \E[U(X_T^{\hat{\mathcal{Z}}^{m,\tau^m}}){\bf 1}_{C^c_m}]\bigg\}\\&\geq
\limsup_m \bigg\{\E[U(X^*){\bf 1}_{C_m}]+[U(w_0)-U(w_0-\min(\delta,\frac{1}{N_0}))]\P[C_{m}]\\
&\hspace{2cm}+ \E[U(X_T^{\hat{\mathcal{Z}}^{m,\tau^m}}){\bf 1}_{C^c_m}]\bigg\}\\
&>\limsup_m  \bigg\{\E[U(X^*){\bf 1}_{C_m}]+ \E[U(X_T^{\hat{\mathcal{Z}}^{m,\tau^m}}){\bf 1}_{C^c_m}] \bigg\}\\
&\geq \liminf_m  \bigg\{\E[U(X^*){\bf 1}_{C_m}]+ \E[U(X^*){\bf 1}_{C^c_m}] \bigg\}= \E[U(X^*)]\geq v^*,
\end{align*}
which is a contradiction to the definition of $v^*$ in (\ref{vstar}), where we used the definition of $C_m$ in the first, (\ref{liminf}) in the fourth and (\ref{superoptimal}) in the last inequality.\\
\endproof

\section{Proofs of Section \ref{se:BSPDE}}\label{Ap:F}

\noindent\textbf{Proof of Proposition \ref{Le:supermart}.}
Let $\theta^0\in\Theta$. Denote by  
$\Theta(\theta^0,t,T)$ the set of all admissible strategies being equal to $\theta^0$ until time $t$, i.e.,
$\theta\in \Theta(\theta^0,t,T)$ if $\theta\in\Theta$ and $\theta_s{\bf 1}_{0\le s\le t}=\theta^0_s{\bf 1}_{0\le s\le t}$. Let us show that the family 
$$
\big\{\Upsilon_t^\theta:=\E[U(x+\cI_\zs{t,T}(\theta)+H_L)\vert \cF_t], \,\theta\in \tilde{{\Theta}}(\theta^0,t,T)\big\}
$$ admits the lattice property. Indeed, for $\theta^1,\theta^2 \in \tilde{{\Theta}}(\theta^0,t,T)$ we define
$$
\theta_s:=\theta^0_s{\bf 1}_{0\le s< t}+\bigg(\theta^1_s{\bf 1}_\zs{\Upsilon_t^{\theta^1}\ge \Upsilon_t^{\theta^2}}+\theta^2_s {\bf 1}_\zs{\Upsilon_t^{\theta^1}< \Upsilon_t^{\theta^2}}\bigg) {\bf 1}_\zs{T\ge s\ge t} .
$$
Note that for any $y_1,y_2\in \bbr^n$ and $A\in \cF_t$, we have $Z_t(-y_1S {\bf 1}_A-y_2S {\bf 1}_{A^c})=Z_t(-y_1S){\bf 1}_A+Z_t(-y_2S) {\bf 1}_{A^c}$. It is then clear that 
$$\mathcal{Z}_s^\theta=\mathcal{Z}^{\theta^0}_s{\bf 1}_\zs{0\le s< t}+\bigg(\mathcal{Z}^{\theta^1}_s {\bf 1}_\zs{\Upsilon_t^{\theta^1}\ge \Upsilon_t^{\theta^2}}+\mathcal{Z}^{\theta^2}_s{\bf 1}_\zs{\Upsilon_t^{\theta^1}< \Upsilon_t^{\theta^2}}\bigg) {\bf 1}_\zs{T\ge s\ge t},$$ and $\theta$ is admissible. Since ${\bf 1}_\zs{\Upsilon_t^{\theta^1}\ge \Upsilon_t^{\theta^2}}$ is $\cF_t$-measurable we deduce that
$
\Upsilon_t^\theta={\bf 1}_\zs{\Upsilon_t^{\theta^1}\ge \Upsilon_t^{\theta^2}}\Upsilon_t^{\theta^1}+{\bf 1}_\zs{\Upsilon_t^{\theta^1}< \Upsilon_t^{\theta^2}} \Upsilon_t^{\theta^2}=\max(\Upsilon_t^{\theta^1},\Upsilon_t^{\theta^2}).
$
Noting that the lattice property allows to interchange essential supremum and conditional expectations, it follows that $V(t,x+\cI_\zs{s,t}(\theta))$ is a supermartingale and \eqref{eq:super} holds.
Finally, the equivalence property can be seen as follows: if $\theta^*$ is optimal then
	\begin{align*}
	V(0,x)=\sup_{\theta\in\Theta} \E[U(x+\cI_\zs{0,T}(\theta)+H_L)]=\E[U(x+\cI_\zs{0,T}(\theta^*)+H_L)]=\E[V(T,x+\cI_\zs{0,T}(\theta^*))].
	\end{align*}
	Hence, the supermartingale process $V(t,x+\cI_\zs{s,t}(\theta^*))$ is a martingale. {Assume now that $V(t,x+\cI_\zs{s,t}(\theta^*))$ is a martingale. We have then
	%$V(0,x)=\E[V(T,x+\cI_\zs{0,T}(\theta^*))]=\E[U(x+\cI_\zs{0,T}(\theta^*)+H_L)],$
	%On the other hand, 
	by the supermartingale property %, it holds that
	$$
	\sup_{\theta\in\Theta} \E[U(x+\cI_\zs{0,T}(\theta)+H_L)]\le V(0,x)=\E[V(T,x+\cI_\zs{0,T}(\theta^*))]=\E[U(x+\cI_\zs{0,T}(\theta^*)+H_L)],
	$$
	which shows that $\theta^*$ is optimal.} \qed\\%\red{GIVE A MORE DETAILED PROOF.}\thaicomment{Mitja, does it now sound ok?}

	\textbf{Proof of Proposition \ref{Le:VlamdaV}.} 
%%see Example \ref{completemarket}. 
%This means every $X_T \in \mathcal{D}_T$ is attainable for a $\theta\in\tilde{\Theta}$ as long as $-\E^g(-X_T) \leq x_0$. Since the $g$-expextation is concave, see \cite{Jiang2008convexity}, the set of attainable portfolios is convex, from which the concavity of $V$ immediately follows. \\ 
%%This means for every $Z \in L^2 (\d P \times \d S)$ giving raise to an $X \in \mathcal{D}_T$ there exists $\theta \in \mathcal{S}$ such that $Z = \mathcal{Z}^\theta$. \\
Note that, $ \E[U(x-\int\limits_0^T g(s,Z_s)\d s + \int\limits_t^T Z_s \d Ws+H_L)]$ is jointly concave in $(x,Z)$. By assumption,  ${cl(\text{Im}(Z))}$, the set where $Z$ takes values in, is a convex set. Thus, $V$ is concave as maximum of a jointly concave function in one of its arguments over a convex set.\\
To show the strict concavity let $\theta^1,\theta^2$ be the corresponding optimal strategy starting with initial wealth $x_1$, $x_2$ respectively, i.e.,
$$
V(t,x_1)=\E\big[U(x_1+\cI_\zs{t,T}(\theta^1)+H_L)\mid\cF_t\big],\quad \mbox{and}\quad V(t,x_2)=\E\big[U(x_2+\cI_\zs{t,T}(\theta^2)+H_L)\mid\cF_t\big].
$$ Let us assume that $\lambda V(t,x_1)+(1-\lambda)V(t,x_2)=V\big(t,\lambda x_1+(1-\lambda)x_2\big)$ for some $\lambda\in[0,1]$ {on a non-zero $\mathcal{F}_t$-measurable set, say $A$}. Due to the strict concavity of $U$ we deduce that
$
x_1+\cI_\zs{t,T}(\theta^1)=x_2+\cI_\zs{t,T}(\theta^2),
$
{on $A$} which leads to 
\begin{equation}
	x_2-x_1=-\int_t^T (g(s,\mathcal{Z}^{\theta^1}_s)-g(s,\mathcal{Z}^{\theta^2}_s))\d s+ \int_t^T (\mathcal{Z}^{\theta^1}_s-\mathcal{Z}^{\theta^2}_s)\d W_s=\int_t^T  (\mathcal{Z}^{\theta^1}_s-\mathcal{Z}^{\theta^2}_s) \d W_s^G,
	\label{eq:x12}
\end{equation}
{on $A$}, {for $T\geq s\geq t$, where $W^G_s:=W_s-W_t-\int_t^s \sum_{j=1}^d G^j(u,\mathcal{Z}^{\theta^1}_u,\mathcal{Z}^{\theta^2}_u)\d u$} and $$G^j(s,z,\tilde{z}):=\frac{g(s,z^1,\ldots,z^{j},\tilde{z}^{j+1},\ldots,\tilde{z}^d)-g(s,z^1,\ldots,z^{j-1},\tilde{z}^j,\ldots,\tilde{z}^d)}{z^j-\tilde{z}^j},$$ with $\frac{0}{0}:=0$. Under Condition ${\bf (HL)}$, $G$ is bounded. 
%In case of Condition ${\bf (Hg)}$ we observe that 
%$
%\vert G(s,\mathcal{Z}^{\theta^1}_s,\mathcal{Z}^{\theta^2}_s)\vert \le K(1+\vert \mathcal{Z}^{\theta^1}_s\vert + \vert \mathcal{Z}^{\theta^2}_s\vert).
%$
%Now it follows directly from (\ref{eq:x12}) that $\mathbf{1}_A(\mathcal{Z}_s^{\theta^1})_{t\leq s\leq T }$ and $\mathbf{1}_A(\mathcal{Z}_s^{\theta^2})_{t\leq s\leq T}$ are $\mbox{BMO}(\P | A)$'s (see Appendix \ref{sec:appBMO}). Thus, we conclude by our growth conditions that $\mathbf{1}_AG(s,\mathcal{Z}^{\theta^1}_s,\mathcal{Z}^{\theta^2}_s) \in \mbox{BMO}(\P |A)$ as well. 
%Hence, by Kazamaki's Theorem (see e.g. \cite{carmona2008} [Th. 3.24, Chapter 3]),
Hence, by Novikov's condition, $W^G$ is a Brownian motion under $\P^G$ conditionally on $A$, where $\P^G$ is defined by 
$\d \P^G/\d \P|\mathcal{F}_t=\cE_T\big(-\int_t^TG(s,\mathcal{Z}^{\theta^1}_s,\mathcal{Z}^{\theta^2}_s)\d W_s\big)$. 
Moreover, it follows from Girsanov's Theorem %and general properties of BMOs 
that 
%since $\int_t^u (\mathcal{Z}^{\theta^1}_s-\mathcal{Z}^{\theta^2}_s) \d W_s$ is a $\mbox{BMO}(\P|A)$ martingale it follows from general properties of BMOs, see \cite{carmona2008}, that 
$\int_t^u  (\mathcal{Z}^{\theta^1}_s-\mathcal{Z}^{\theta^2}_s) \d W_s^G$ is a martingale under $\P^G$. %see Appendix \ref{sec:appBMO}. 
Taking expectations conditional on $\P^G[\cdot|\mathcal{F}_t]$ on both sides in \eqref{eq:x12} we obtain {that $x_2=x_1$} and the proof is complete.\qed \\
	
\noindent\textbf{Proof of Lemma \ref{Le:abc}.} {The first part of the lemma follows as the right-hand side of \eqref{eq:cL} is a lower semi-continuous, coercive function in $\cZ.$ For the second part, note that the right-hand side of \eqref{eq:cL} is strictly convex, as $cl(\text{Im}(Z))$ is convex.} \qed
%Optimizing over $\cZ$ yields that $\upsilon$ is determined by the solution of $-\nabla {g}_A^\intercal(t,\cZ)V_x(t,x)+ V_{xx}(t,x)\cZ+\alpha_x(t,x)=0.$ The left-hand side is decreasing and surjective in $\cZ$. \qed

%\begin{lemma}\label{le:ap2}
%	In the case that $\mathcal{L}^V$ is differentiable in $x$ and (H2) holds we have 
%	\begin{equation}\label{eq:ap2}
%	\mathcal{L}_x^V(t,x) = - g(t,\upsilon(t,x))V_{xx}(t,x) + \frac{1}{2}\upsilon^2(t,x)V_{xxx}(t,x)+\upsilon(t,x)\alpha_{xx}(t,x).
%	\end{equation}
%\end{lemma}
%\proof
%For those $x>0$ for which we have $\hat{\theta}(t,x) \neq 0$, $\upsilon(t,x)=\hat{\theta}(t,x)Z_t(-S)$ is differentiable in $x$, and the proof that \eqref{eq:ap2} holds, follows as in the previous lemma. 
%The proof that \eqref{eq:ap2} also holds for all $x$ for which $\hat{\theta}(t,x)=0$ will be omitted. 
%\endproof
\bigskip
\noindent\textbf{Proof of Theorem \ref{Th:SBPDEss}.}
Using It\^o-Ventzel's formula we can represent 
\begin{align}
V(t,x+\cI_\zs{s,t}(\theta))=&V(s,x)+\int_s^t \Big(G(u,x+\cI_\zs{s,u}(\theta),\mathcal{Z}^\theta_u)-b(u, x+\cI_\zs{s,u}(\theta))\Big) \d u\notag\\
&+\int_s^t \Big(V_x(u,x+\cI_\zs{s,u}(\theta))\mathcal{Z}^\theta_u+\alpha(u,x+\cI_\zs{s,u}(\theta))\Big)\d W_u,
\label{eq:Vform1}
\end{align}
where $G(u,p,z):=-{g}(u,z)V_x(u,p)+\frac{1}{2}\vert z\vert^2 V_\zs{xx}(u,p)+\alpha_x(u,p)z.$

We recall from Proposition \ref{Le:supermart} that for any $x$, the value process $V(t,x+\cI_\zs{s,t}(\theta))$ is a supermartingale. Hence, the finite variation part in \eqref{eq:Vform1} is decreasing. Therefore, for any admissible strategy $\theta$ and $\epsilon>0$
$$
\int_s^{s+\epsilon} b(u, x+\cI_\zs{s,u}(\theta)) \d u\ge \int_s^{s+\epsilon} G(u,x+\cI_\zs{s,u}(\theta),\mathcal{Z}^\theta_u)\d u.
$$
Dividing both sides by $\epsilon$ and letting $\epsilon\to 0$ we obtain
$
 b(s, x)\ge G(s,x,\mathcal{Z}^\theta_s) \quad \d \P\times \d t \quad \text{a.s}.
$
Noting that $V_{xx}\le 0$ and $V_x>0$ a.e. we get from Lemma \ref{Le:abc} for all $x$ 
\begin{equation}
b(t,x)\ge \esssup_\zs{\cZ\in {cl(\text{Im}_t(Z))}} G(t,x,\mathcal{Z})=\cL^V(t,x),
\label{eq:es}
\end{equation}
where $\cL^V(t,x)$ is defined by \eqref{eq:cL}. Now assume that $\theta^*$ is an optimal strategy, i.e.,  $V(t,x+\cI_\zs{s,t}(\theta^*))$ is a martingale. Let $X^{\theta^*}_{s,t}(x)=x+\cI_\zs{s,t}(\theta^*)$ and note that $X^{\theta^*}_{0,t}(x)=x+\cI_\zs{0,t}(\theta^*)= X^{\theta^*}_{t}.$ Using the It\^o-Ventzel formula we get for any $s\in [0,t]$,
$G(t,X^{\theta^*}_{s,t}(x),\mathcal{Z}^{\theta^*}_t)-b(t, X^{\theta^*}_{s,t}(x))=0, \: \d \P\times \d t \ \text{a.s}.$
It follows from \eqref{eq:es} then that $\mathcal{Z}^{\theta^*}_t$ must be the maximizer of $G(t,X^{\theta^*}_{s,t}(x),\mathcal{Z})$ and therefore $\mathcal{Z}^{\theta^*}_t=\upsilon(t,X^{\theta^*}_{s,t}(x))$ for a function $\upsilon$ defined in \eqref{eq:cL}. %%which means that
Hence, for any $s\in[0,t]$ we have $0\in \cU^V(t,\mathcal{Z}^{\theta^*}_t,X^{\theta^*}_{s,t}(x))$. Consequently, taking $s=0$ and $s=t$ we obtain \eqref{eq:Z*gen} and $b(t,x)=\cL^V(t,x)$,
which leads to the BSPDE \eqref{eq:VLV}. 

On the other hand, if $\wt{\theta}$ is a strategy satisfying \eqref{eq:Z*gen}-\eqref{eq:Xoptbgen}, {and $cl(\text{Im}(Z))$ is convex, $\upsilon$ is unique, and} one can verify using the preceding steps that the value process $V(t,x+\cI_{s,t}(\wt{\theta}))$ is a local martingale. Since it by assumption belongs to the class $D$, $V(t,x+\cI_{s,t}(\wt{\theta}))$ is a martingale and $\wt{\theta}$ is optimal by Proposition \ref{Le:supermart}.\qed \\

\noindent\textbf{Proof of Theorem \ref{Th:6.2}.}
First, by the It\^o-Ventzel formula, it can be seen directly that
\begin{align*}
V_x(t,X^{\theta^*}_t)=&V_x(0,x)+\int_0^t \big(\mathcal{Z}^{\theta^*}_sV_{xx}(s,X^{\theta^*}_s) +\alpha_x(s,X^{\theta^*}_s)\big)\d W_s\\
&+ \int_0^t \big(-\cL^V_x(s,X^{\theta^*}_s)+\frac{1}{2}\vert\mathcal{Z}^{\theta^*}_s\vert^2 V_{xxx}(s,X^{\theta^*}_s)-g(s,\mathcal{Z}^{\theta^*}_s)V_{xx}(s,X^{\theta^*}_s)+\mathcal{Z}^{\theta^*}_s\alpha_{xx}(s,X^{\theta^*}_s)\big) \d s.
\end{align*}
From Lemma \ref{le:ap1} below we observe that the finite variation term is zero, which implies that $V_x(t,X^{\theta^*}_t)$ is a local martingale whose decomposition is given by 
\begin{align*}
V_x(t,X^{\theta^*}_t)=V_x(0,x)+\int_0^t \big(\mathcal{Z}^{\theta^*}_sV_{xx}(s,X^{\theta^*}_s) +\alpha_x(s,X^{\theta^*}_s)\big)\d W_s.
\end{align*}
Let $\zeta_t=I(V_x(t,X_t^{\theta^*}))-X_t^{\theta^*}$. Clearly, $V_x(t,X_t^{\theta^*})=U'(X_t^{\theta^*}+\zeta_t)$ and therefore 
$$I'(V_x(t,X_t^{\theta^*}))=\frac{1}{U''(X_t^{\theta^*}+\zeta_t)},\quad 
I''(V_x(t,X_t^{\theta^*}))=-\frac{U^{(3)}(X_t^{\theta^*}+\zeta_t)}{(U'')^3(X_t^{\theta^*}+\zeta_t)}.$$ 
Note that $\zeta_T=I(U'(X_T^{\theta^*}+H_L))-X_T^{\theta^*}=H_L$. By It\^o's formula we get
\begin{align}
\d \zeta_t=& \frac{\mathcal{Z}^{\theta^*}_t V_{xx}(t,X_t^{\theta^*})+\alpha_x(t,X_t^{\theta^*})}{U''(X_t^{\theta^*}+\zeta_t)} \d W_t -\mathcal{Z}^{\theta^*}_t \d W_t \notag \\
&+\frac{1}{2}I''(V_x(t,X_t^{\theta^*})) \vert\mathcal{Z}^{\theta^*}_t V_{xx}(t,X_t^{\theta^*})+\alpha_x(t,X_t^{\theta^*})\vert^2 \d t+g(t,\mathcal{Z}^{\theta^*}_t) \d t. \label{eq:twostar}
\end{align}
Set 
\begin{equation}
M_t=\frac{\upsilon(t,X^{\theta^*}_t)V_{xx}(t,X_t^{\theta^*})+\alpha_x(t,X_t^{\theta^*})}{U''(X_t^{\theta^*}+\zeta_t)}-\upsilon(t,X^{\theta^*}_t).
\label{eq:spdem}
\end{equation}
%\end{appendix}\end{document}
%Since $v$ by assumption takes values in the interior $\nabla {g}_A(t,v(t,X_t^{\theta^*}))=g_z(t,v(t,X_t^{\theta^*}))$ . 
By \eqref{eq:Z*gen} we have $\upsilon(t,X^{\theta^*}_t)V_{xx}(t,X_t^{\theta^*})+\alpha_x(t,X_t^{\theta^*})\in \nabla {g}_A^\intercal(t,\upsilon(t,X^{\theta^*}_t)) V_x(t,X_t^{\theta^*})$ {with $A=cl(\text{Im}_t(Z))$}. \tn{As $\upsilon$ only takes values in the interior of $\text{Im}_t(Z))$, $\nabla{g}_{\text{Im}_t(Z)}(t,\upsilon(t,X^{\theta^*}_t))=g_z(t,\upsilon(t,X^{\theta^*}_t))=g_{cl(\text{Im}_t(Z))}(t,\upsilon(t,X^{\theta^*}_t))$}. Hence, it follows from \eqref{eq:spdem} that
$0\in -\nabla{g}_{\tn{z}}^\intercal(t,\upsilon(t,X^{\theta^*}_t))U'(X_t^{\theta^*}+\zeta_t)+U''(X_t^{\theta^*}+\zeta_t)(M_t+\upsilon(t,X^{\theta^*}_t))$, and we have $\cH_{{\tn{cl(\text{Im}_t(Z))}}}(t,X_t^{\theta^*},\zeta_t,M_t)=\cZ_t^{\theta^*}=\upsilon(t,X^{\theta^*}_t)$, where the last equation holds by Theorem \ref{Th:SBPDEss}. \tn{To see the first equation, note first that by Theorem \ref{Th:FB}, $\cH_{{\tn{cl(\text{Im}_t(Z))}}}(t,X_t^{\theta^*},\zeta_t,M_t)$ is a singleton. Furthermore, as $\cZ^{\theta^*}$ by assumption lies on the interior of $\text{Im}_t(Z)$, $\nabla{g}_{\text{Im}_t(Z)}(t,\cZ^{\theta^*}_t)=g_z(t,\cZ^{\theta^*}_t)=\nabla{g}_{cl(\text{Im}_t(Z))}(t,\cZ^{\theta^*}_t)$ so by Theorem \ref{Th:FB}, $\cZ^{\theta^*}_t\in \cH_{{cl(\text{Im}_t(Z))}}(t,X_t^{\theta^*},\zeta_t,M_t)$.} It is then straightforward to see using (\ref{eq:twostar}) that the triple $(X_t^{\theta^*},\zeta_t,M_t)$ satisfies the FBSDE \eqref{eq:FBSDE} with $H_{\mathrm{M}} =0$.\qed

%M_t(=Y_t U'(X_t^{\theta^*}+\zeta_t))/U''(X_t^{\theta^*}+\zeta_t)-\mathcal{Z}^{\theta^*}_t.
%$$
%Furthermore, 
%$$
%-g(t,\mathcal{Z}^{\theta^*}_t) V_x(t,X^{\theta^*}_t)+\mathcal{Z}^{\theta^*}_t V_{xx}(t,X^{\theta^*}_t)+\alpha_x(t,X^{\theta^*}_t)=0.
%$$
%\end{appendix}\end{document}

\begin{lemma}\label{le:ap1}
	If $\upsilon$ only takes values in the interior of $\text{Im}(Z)$ we have  
	\begin{align*}
	\mathcal{L}_x^V(t,x) = - g(t,\upsilon(t,x))V_{xx}(t,x) + \frac{1}{2}\upsilon^2(t,x)V_{xxx}(t,x)+\upsilon(t,x)\alpha_{xx}(t,x).
	\end{align*}
\end{lemma}
\proof
By Lemma \ref{Le:abc} and since $g$ is differentiable, $\upsilon$ is characterized by 
\begin{equation}\label{eq:ap1}
-g^\intercal_z(t,v(t,x))V_x(t,x)+ v(t,x) V_{xx}(t,x)+\alpha_x(t,x)=\mathcal{U}^V(t,\upsilon(t,x),x)=0.
\end{equation}
Note that $g_z$ is increasing, $V_x>0$, and $V_{xx}<0$. As $\upsilon$ by assumption only takes values in the interior of ${\text{Im}(Z)_t}$ all functions involved are continuously differentiable and the implicit function theorem gives that $\upsilon$ is continuously differentiable in $x$. Hence, taking the derivative of $\mathcal{L}^V$ in \eqref{eq:cL} w.r.t. $x$ we obtain using \eqref{eq:ap1}
\begin{align*}
\mathcal{L}_x^V(t,x) &= \upsilon_x(t,x)(-g^\intercal_z(t,\upsilon(t,x))V_x(t,x) +\upsilon(t,x)V_{xx}(t,x) + \alpha_x(t,x)) \\
&\hspace{1cm} - g(t,\upsilon(t,x))V_{xx}(t,x) + \frac{1}{2}\upsilon^2(t,x)V_{xxx}(t,x)+\upsilon(t,x)\alpha_{xx}(t,x) \\ 
&=  - g(t,\upsilon(t,x))V_{xx}(t,x) + \frac{1}{2}\upsilon^2(t,x)V_{xxx}(t,x)+\upsilon(t,x)\alpha_{xx}(t,x).
\end{align*}

\section{Proofs of Section \ref{se:Regu}}
\noindent\textbf{Proof of Proposition \ref{propcara}.}
By \eqref{eq:Vx} and (\textbf{CV}), the value function can be written as $V(t,x)=U(x)V_t$, where
%\red{Thai: Essinf?}\thaicomment{Indeed, thanks}
\begin{align*}
V_t:&=\essinf_{\zs{\theta \in \Theta, \theta_s, s\in[t,T]} }\E\big[e^{-\gamma \cI_\zs{t,T}(\theta)}\mid\cF_t\big]=\E\big[e^{-\gamma \cI_\zs{t,T}(\theta^*)}\mid\cF_t\big]=e^{\gamma \cI_\zs{0,t}(\theta^*)}\E\big[e^{-\gamma \cI_\zs{0,T}(\theta^*)}\mid\cF_t\big].
\label{eq:VxExp}
\end{align*}
By the martingale representation theorem and since $\cI_\zs{0,t}{(\theta^*)}$ is an It\^o-process, $V_t$ is an It\^o-process as well {and of the form $V_t=V_0-\int_{0}^{t}b_s \d s+\int_0^t \alpha_s \d W_s $.} Therefore, $V(t,x)$ is a special semimartingale. Hence, for $V(t,x)=U(x)V_t$, conditions (a)-(c) in Definition \ref{Def:1} are fulfilled. The definition of $V_t$ and the separation structure allow us to rewrite the BSPDE \eqref{eq:ValueBSPDE} as $V(t,x)=U(x)V_t$ with $\alpha (t,x)=U(x)\alpha_t,$ $b(t,x)=U(x)b_t$, where $V$ satisfies a BSDE of the form $V_t=V_0-\int_0^t {b}_s \d s+\int_0^t \alpha_s \d W_s \text{ and } V_T=1.$
Repeating the steps in the proof of Theorem \ref{Th:SBPDEss} (where, due to the separation structure, we can simply apply It\^o's Lemma rather than the It\^o-Ventzel formula) we obtain the finite variation term $\cL^V_t$ of the proposition.  %\mitjacomment{Is condition (CR) satsified? Do we actually need it in the paper somewhere?}\thaicomment{Well, Condition (CR) is needed for Ito-Ventzel formula however,it seems that we do not need to check (CR), even though, it might be true for this case}
\qed \\

\begin{lemma}\label{Le:der}
For $x=f(y)$, { which is the inverse of the  decreasing  function $U'(x)e^{-\gamma x}{\gamma ^{-1}}$ on $\bbr\to\bbr_+$, we have}
\begin{enumerate}
	\item [(i)]$y\wt{U}''(y)=\frac{\gamma}{R_1(x)+\gamma}e^{\gamma x}$,
	\item [(ii)] $
	y^2\wt{U}'''(y)=-\gamma^2\bigg(\frac{\gamma}{R_1(x)+\gamma}\bigg)^2e^{\gamma x}+\gamma\frac{-R_1(x) R_2(x)+2\gamma R_1(x)-\gamma^2}{ (R_1(x)+\gamma)^3} e^{\gamma x},
	$
	\item [(iii)] if, in addition, the relative risk aversion $R_1(x)$ and $R_2(x)$ are bounded and bounded away from zero and $\E[U(f(\lambda \xi_T))]<\infty,$ and $\E[e^{\gamma f(\lambda \xi_T)}\xi_T]<\infty,$ for any $\lambda>0$,  then $\wt{V}(y)$ is well-defined and three times differentiable for $y>0$.
\end{enumerate}
\end{lemma}

\noindent\textbf{Proof of Lemma \ref{Le:der}.}
It can be derived from \eqref{eq:Utilde} that $\wt{U}'(y)=-e^{\gamma x}$, $\wt{U}''(y)=-\gamma e^{\gamma x}x'$ and $\wt{U}'''(y)=-\gamma^2 e^{\gamma x}(x')^2-\gamma e^{\gamma x} x''$ (suppressing the dependence of the optimal $x$ on $y$). Using the relation $\gamma y=U'(x) e^{-\gamma x}$ we obtain $x'=\frac{\gamma e^{\gamma x}}{U''(x)-\gamma U'(x)}$ and $x''=-\gamma^2 \frac{U'''(x)-2\gamma U''(x)+\gamma^2 U'(x)}{(U''(x)-\gamma U'(x))^3}e^{2\gamma x}$, which implies the first two conclusions. For (iii), it suffices to observe that for $y>0$,
$
\wt{V}(y)=\E[\wt{U}(y\xi_T)]$ and $\wt{V}'(y)=\E[\xi_T\wt{U}'(y\xi_T)]=\E[-\xi_Te^{\gamma f(y\xi_T)}]<\infty
$
are well-defined by assumption. Moreover, using (i) and (ii), we can show that $${\wt{V}''(y)}=\E[\xi_T^2\wt{U}''(y\xi_T)]=\frac{1}{y}\E[\xi_T(y\xi_T)\wt{U}''(y\xi_T)]=\frac{1}{y}\E[\xi_T\frac{\gamma}{R_1(f(y\xi_T))+\gamma}e^{\gamma f(y\xi_T)}]<\infty$$
and, similarly,
\begin{align*}
{\wt{V}'''(y)}=\E[\xi_T^3\wt{U}'''(y\xi_T)]&=\frac{1}{y^2}\E[\xi_T(y\xi_T)^2\wt{U}'''(y\xi_T)]=\frac{1}{y^2}\E[(-\gamma^2)\bigg(\frac{\gamma}{R_1(x)+\gamma}\bigg)^2 \xi_Te^{\gamma f(y\xi_T)}]\\
&-\gamma\E[\frac{-R_1(f(y\xi_T)) R_2(f(y\xi_T))+2\gamma R_1(f(y\xi_T))-\gamma^2}{\gamma^2 (R_1(f(y\xi_T))+\gamma)^3}\xi_T e^{\gamma f(y\xi_T)}]<\infty.
\end{align*}
{Note that derivatives and expectations can be interchanged since by assumption, the relative risk aversion $R_1(x)$ and $R_2(x)$ are bounded and bounded away from zero, $\E[e^{\gamma f(\lambda \xi_T)}\xi_T]<\infty$ for any $\lambda>0$, and $f$ is monotone.} %Since $f$ is monotone one may infer that, $\wt{V}(y)$ is three times differentiable.
%\mitjacomment{Thai: Can you explain this more? Is $f$ convex?}\thaicomment{Does it sound better? It is not clear if $f$ is convex.}\mitjacomment{It seems like it is argued that you can switch derivative and expectation. Why can you do this? Normally you need a that some inner supremum is integrable which would for instance be guaranteed $f$ is monotone.}\thaicomment{Indeed, as mentioned above, $f$ is decreasing. Here, the inner term can be bounded by $e^{\gamma f(\lambda \xi_T)}\xi_T$ (scaled by some constant since $R_1(x)$ and $R_2(x)$ are bounded and bounded away from zero. It is clear enough?}  
\qed\\

%\end{appendix}\end{document}
\begin{lemma}\label{Le:basic}
Under the assumptions of Proposition \ref{Th:quadg} the inverse of the stochastic flow $x\rightarrow X_t(x)$ exists and is three times continuously differentiable for any $t\in[0,T].$ Furthermore, all coefficients of $X_t(x)$ are differentiable and locally Lipschitz. Finally, denoting the stochastic integrand of $X_t(x)$ by $\cZ^*(t,x), \cZ^*$ and $\frac{\partial \cZ^*}{\partial x}$ are locally Lipschitz. %(i.e., the stochastic integrands in the It\^{o} formula) are differentiable and locally Lipschitz. %\thaicomment{Mitja, what is $\sigma$?}\mitjacomment{Defined now above the lemma. }
\end{lemma}
\noindent\textbf{Proof of Lemma \ref{Le:basic}.}
By Proposition \ref{propcomplete} we have that $X_T(x)=f(\lambda(x)\xi_T)$ with $\lambda(x)$ being uniquely determined by the budget constraint
		\begin{equation} 
		\label{implicit}
		-\frac{1}{\gamma}\log\bigg(\E^\Q[e^{-\gamma f(\lambda\xi_T)}]\bigg)-x=0,
		\end{equation}
		and $X_t(x)=-\frac{1}{\gamma}\log\bigg(\E^\Q[e^{-\gamma f(\lambda(x)\xi_T)}\vert \cF_t]\bigg).$
	Clearly, the derivative w.r.t. $\lambda$ of the left hand-side in \eqref{implicit} is not zero for any $x$. Hence, the inverse function theorem entails that $\lambda(x)$ is three times continuously differentiable (as $f$ is). Furthermore, it follows directly by (\ref{implicit}) that $\lambda$ is strictly monotone and converges to $\pm \infty$ as $x$ converges to $\pm\infty$. Note that by the definition of $0\leq \xi_T$ and the fact that $f$ is decreasing we have
	that $\E[e^{\gamma f(\lambda \xi_T)}\xi^p_T]<\infty$ for $p\geq 1$ and for every $\lambda>0$.
	Hence, $X_t(x)$ is invertible and three times continuously differentiable as well, which entails that also its inverse is three times continuously differentiable. Differentiability and local Lipschitz continuity of the coefficients of $X_t(x)$ together with $\cZ^*$ and $\frac{\partial \cZ^*}{\partial x}$ being locally Lipschitz may be seen via expressing $X_t(x)$ by the Feynman-Kac theorem, as a function of $W_t$ and $x$ which can be written as a differentiable integral over a Gaussian kernel, with derivatives differentiable in $x$. \qed\\
%\thaicomment{
\begin{lemma}\label{Le: inverseflow} %\mitjacomment{Thai: Is this condition satisfied?}\disc{ This assumption  holds for exponential utility but may be strict}. 
%Assume that $f'(z)z$ is uniformly bounded, i.e. $\sup_{z\in\bbr_+}\vert f'(z)z\vert\le K$, for some constant $K>0$.   
Under the assumptions of Proposition \ref{Th:quadg}, the inverse stochastic flow $\psi_t(x):={X^*_t}^{-1}(x)$ exists and is an Itô process whose dynamics is given by
\begin{equation}
d \psi_t(x)=a(t,\psi_t(x)) dt+b(t,\psi_t(x)) d W^{\Q}_t,
\label{eq:psiSDE}
\end{equation}
where 
$b(t,\psi_t)=-\cZ^*(t,\psi_t(x))/X_x(t,\psi_t(x))$ and
$$
a(t,\psi_t(x))=-\frac{1}{2}\frac{X_{xx}(t,\psi_t(x)) (\cZ^*(t,\psi_t(x))^{{2}}}{X_{x}^3(t,\psi_t(x))}+\frac{\gamma}{2}\frac{(\cZ^*(t,\psi_t(x)))^{\tn{2}}}{X_x(t,\psi_t(x))}+\frac{\cZ^*(t,\psi_t(x))\cZ^*_x(t,\psi_t(x))}{X_x^2(t,\psi_t(x))}.
$$
\end{lemma}
\proof
{The proof is similar to that in \cite{mania2017} and \cite{nicole2013exact}. By Lemma \ref{Le:basic} the coefficients of $a(t,y)$ and $b(t,y) $ defined above are locally Lipschitz continuous in $y$. Therefore, by a truncating argument in \cite{kunita1997stochastic}, the SDEs \eqref{eq:psiSDE} has a unique maximal solution up to an explosion time $\tau_x\le T$. Using Itô-Ventzel's formula, it is straightforward to see that $dX_t(\psi_t(x))=0$ hence $X_t(\psi_t(x))=x$ a.s. for $t\in[0,\tau_x)$ since $\psi_{\tau_x}(x)=\infty$ if $\tau_x<T$ and $X_t(\infty)=\infty$. On the other hand, by continuity $X_t(\psi_t(x))=x$ if $t=\tau_x<T$. Consequently, we must have $\tau_x=T$ a.s., meaning that $\psi_t(x)=X_t^{-1}(x)$ for $t\in[0,T]$.}\endproof
\\
\noindent\textbf{Proof of Proposition \ref{Le:duality}.}
The proof is similar to that of Theorem 2.0 in \cite{kramkov1999asymptotic} and is omitted.\endproof

Following the ideas of Lemma 4 in \cite{schachermayer2003super}, the conditional version of Proposition \ref{Le:duality} can be derived. In particular, one can show that
\begin{equation}
\frac{1}{\gamma}V'(t,X_t^{\theta^*}(x))e^{-\gamma X_t^{\theta^*}(x)}=y\xi_t,
\label{eq:V'new}
\end{equation}
{where $ y=\frac{1}{\gamma}V'(x)e^{-\gamma x}$}.
Set $X_t(x):=X_t^{\theta^*}(x)=x+\mathcal{I}^{\theta^*}_{0,t}$. Partially generalizing \cite{mania2017}, we show below that for the BSPDE above regularity can be obtained. The optimal wealth is a strictly increasing continuous function
of $x$ $\P$-a.s. Hence, by Lemma \ref{Le:basic} an adapted inverse of $X_t(x)$ exists. \qed %The following result is the main pillar for obtaining regularity.
%Let us assume without loss of generality that $\eta$ is constant. %\red{Thai:martingale random fields?}\thaicomment{yeah, it is}
% and define $\sigma^2:=\int_0^T \eta^2 ds$.

%\tn{Mitja: checking the arguments and Mania's paper, I think we do not need Lipschitz the processes  $(X_T^{\theta^*})'(x)$ and  $(X_T^{\theta^*})''(x)$ are bounded. This can be seen as follows: 
%By Lemmas \ref{Le:der}-\ref{Le:duality}, the value function $V(x)$ is three times differentiable. Furthermore, due to the duality relation \eqref{eq:duality} and $y=\frac{1}{\gamma}V'(x)e^{-\gamma x}$ we obtain
%$$
%U'(X_T^{\theta^*})e^{-\gamma X_T^{\theta^*}}=V'(x) \xi_T e^{-\gamma x}.
%$$
%By differentiating both sides of the last equality one gets
%$$
%(U''(X_T^{\theta^*}(x))-\gamma U'(X_T^{\theta^*}(x)))(X_T^{\theta^*})'(x)e^{-\gamma X_T^{\theta^*}(x)}=(V''(x) -\gamma V'(x))\xi_T e^{-\gamma x}.
%$$
%Therefore ( by replacing $\xi_T=\frac{U'(X_T^{\theta^*})e^{-\gamma X_T^{\theta^*}}}{V'(x)  e^{-\gamma x}}$)
%\begin{equation}
%(X_T^{\theta^*})'(x)=\frac{-\frac{V''(x)}{V'(x)} +\gamma }{-\frac{U''(X_T^{\theta^*}(x))}{U'(X_T^{\theta^*}(x))}+\gamma}
%\label{eq:}
%\end{equation}
%From $\gamma y=U'(x) e^{-\gamma x}=V'(x)e^{-\gamma x}$ it can be seen that $ \frac{V''(x)}{V'(x)}=y'/y+\gamma=\frac{U''(x)}{U'(x)}=-R_1(x)$, which implies that $(X_T^{\theta^*})'(x)$ is bounded. Similarly, we can show that $(X_T^{\theta^*})''(x)$ is bounded.
%}
%\end{appendix}\end{document}
\noindent\textbf{Proof of Proposition \ref{Th:quadg}.} \tn{The proof is similar to the case without market price impact in \cite{mania2017}.}}
In particular, define the {martingale} random fields $M(t,x):=\E[U(X_T(x))\vert\cF_t]$ and $\wt{M}(t,x):=\E[U'(X_T(x))\vert\cF_t].$
%Since $X_T(x)$ is locally Lipschitz, $M$ and $\wt{M}$ can, by Kolmogorov's criterion, be chosen up to a zero set which is independent of $x$. 
%\tn{Note that by Proposition \ref{propcomplete}, $\cZ^*$ is deterministic and the value function $V(t,x)$ is deterministic}. 
By optimality, \tn{using Proposition \ref{Le:supermart}} %Theorem \ref{Th:SBPDEss}, 
$V(t,X_t(x))$ is a martingale and 
$$V(t,X_t(x))=\E[V(T,X_T(x))\vert\cF_t]=\E[U(X_T(x))\vert\cF_t]=M(t,x).$$ 
Moreover, it follows from \eqref{eq:V'new} that 
$
\frac{1}{\gamma}V'(t,X_t(x))e^{-\gamma X_t(x)}=y\xi_t,
$
%\red{Thai: Is $y$ independent of $t$?} \thaicomment{I think so, because $ y=\frac{1}{\gamma}V'(x)e^{-\gamma x}$}
which by Lemma \ref{Le:basic}  and Proposition \ref{Le:duality} implies that 
 $${M}'(t,x)=V'(t,X_t(x)) X_t'(x)=\gamma y\xi_t e^{\gamma X_t(x)}X_t'(x)=V'(x)e^{-\gamma x}\xi_t e^{\gamma X_t(x)}X_t'(x).$$
By (ii) and (iii) of Proposition \ref{Le:duality} we obtain
 $$\wt{M}(t,x)=\E[U'(X_T(x))\vert\cF_t]=\E[\gamma y\xi_T e^{\gamma X_T(x)}\vert\cF_t]=\gamma y\xi_t e^{\gamma X_t(x)}=V'(t,X_t(x))$$
by using the identity \eqref{eq:V'new}. 
% and  %$ y=\frac{1}{\gamma}V'(x)e^{-\gamma x}$) $$\wt{M}'(t,x):=\E[U'(X_T(x))\vert\cF_t]=V''(x)\xi_t e^{-\gamma x} e^{\gamma X_t(x)}.$$
It follows that ${M}'(t,x)= \wt{M}(t,x) X_t'(x).$
%\begin{equation*}
%{M}'(t,x)= \wt{M}(t,x) X_t'(x).
%\label{eq:}
%\end{equation*}
{From Lemma \ref{Le:der}, $\wt{V}$ is three-times differentiable, which implies that its conjugate (see Proposition \ref{Le:duality}) is also $V(x)$ is three-times differentiable}. %\mitjacomment{Thai: can you explain why?}\disc{ it now Makes sense?} 
Therefore, we can conclude that the martingale random fields $M(t,x)$ and $\wt{M}(t,x)$ are two times differentiable. Let 
$M'(t,x)=V'(x)+\int_0^t h(s,x) \d W_s$
be the GKW decomposition of the martingale random field $M'(t,x)$. %\mitjacomment{In order to show that we can write $ds$ in the drift part, do you use that ${M}'(t,x)= \wt{M}(t,x) X_t'(x)$?} \thaicomment{ Not yet, we get this decomposition from the fact that $M'(t,x)$ is a martingale. Is this ok Mitja?} \mitjacomment{Would the drift term be then zero by the martingale property? This seems not true.}\disc{I see, drift is set zero as in the previous version}
 Since $V(t,x)=M(t,X_t^{-1}(x))$ we obtain by It\^o-Ventzel's formula that
\begin{align*}
V(t,x)=V(0,x)&+\int_0^t M(\d s, \psi_s)+\int_0^t M'(s, \psi_s) \d \psi_s+\frac{1}{2} \int_0^t M''( s, \psi_s) d \langle \psi\rangle_s\\
&+\bigg\langle \int_0^t M'(s, \psi_s) \d \psi_s, \psi\bigg\rangle_t,
\end{align*}
where $\psi_t:=X_t^{-1}(x)$ and $\langle\cdot,\cdot\rangle$ denotes the covariation. From the assumption that $zf'(z)$ is uniformly bounded, Lemma \ref{Le: inverseflow} implies that the diffusion term $\cZ(t,x)$ of $X_t(x)$ is $\cF_t$-measurable and bounded and by \cite{kunita1997stochastic} (chapter 3), the inverse stochastic flow $\psi_t(x):={X_t}^{-1}(x)$ exists and is an Itô process
%whose dynamics is given by
$
d \psi_t(x)=a(t,\psi_t(x)) dt+k(t,\psi_t(x)) d W_t,
$
where 
$k(t,\psi_t)=-\cZ(t,\psi_t(x))/X_x(t,\psi_t(x))$ and
$$
a(t,\psi_t(x))=-\frac{1}{2}\frac{X_{xx}(t,\psi_t(x))(\cZ(t,\psi_t(x)))^2}{X^{\tn{3}}_{x}(t,\psi_t(x))}+\frac{\gamma}{2}\frac{(\cZ(t,\psi_t(x)))^2}{X_x(t,\psi_t(x))}+\frac{\cZ(t,\psi_t(x))\cZ_x(t,\psi_t(x))}{X_x^2(t,\psi_t(x))}.
$$%the dynamcis of $\psi_t(x)$ is given by
%\red{Thai: Is $\psi_t$ an It\^o process? If so why?}\thaicomment{I guess it is Ito's process. Can we confirm this from  Lemma \ref{Le:basic}?}\mitjacomment{How can it be seen from Lemma \ref{Le:basic}?} \thaicomment{See Lemma \ref{Le: inverseflow} above}
%{By Lemma \ref{Le: inverseflow}, 
Now, using Itô-Ventzel's formula, we may see directly that the integrand of the finite variation term (in the definition of regular semimartingales) of $V(t,x)$ is given by
$$
M'(t,\psi_t(x)) a(t,\psi_t(x))+\frac{1}{2}M''(t,\psi_t(x)) k^2(t,\psi_t(x))+h(t,\psi_t(x))k(t,\psi_t(x)),
$$
%where the coefficients $a,b$ of $\psi_t(x)$ are given in Lemma \ref{Le: inverseflow}, 
and hence (b) is confirmed.
%\end{appendix}\end{document}
Similarly, since $V'(t,x)=\wt{M}(t,X_t^{-1}(x))$ and $\wt{M}$ is twice differentiable we can use It\^o-Ventzel's formula to obtain that $V'(t,x)$ is a special semimartingale with progressively measurable finite variation and $V'(t,x)$ admits the representation
\begin{equation}
V'(t,x)=V'(0,x)+\int_0^t \ov{b}(s,x)\d s+ \int_0^t \ov{k}( s, x) \d W_s
\label{eq:Vprime}
\end{equation}
for integrable processes $\ov{b},\ov{k}$ being continuous in $x$. Integrating both sides of \eqref{eq:Vprime} with respect to the space argument on $[0,x]$ and applying the stochastic Fubini's theorem we can conclude that (c) holds. Hence, conditions (a)-(c) are fulfilled and $V(t,x)$ is a regular family of semimartingales.\qed
\end{appendix}

\end{document}